\RequirePackage{fix-cm}
\documentclass[smallextended]{svjour3}       
\smartqed  
\usepackage{latexsym,graphics,amssymb,amsmath}

\usepackage[pdfstartview=FitH,colorlinks=true,citecolor=blue,bookmarks=false]{hyperref}

\usepackage{subcaption}

\usepackage{supertabular}
\usepackage{xcolor}

\usepackage{overpic}

\usepackage[shortlabels]{enumitem}

\usepackage{graphicx,float,wrapfig}

\usepackage{mathtools}
\usepackage{dsfont}

\definecolor{Saddle}{HTML}{8b4513}
\definecolor{NewGrey}{HTML}{85888D}
\definecolor{green1}{rgb}{0,0.5,0} 

\usepackage[normalem]{ulem}

\numberwithin{equation}{section}

\usepackage{units}
\usepackage{cases}

\newcommand{\sk}{s_{\kappa}}
\newcommand{\sQ}{s_\textit{\tiny Q }}
\newcommand{\sG}{s_\textit{\tiny G}}

\renewcommand{\epsilon}{\varepsilon}
\renewcommand{\leq}{\leqslant}
\renewcommand{\geq}{\geqslant}

\newcommand{\be}{\begin{equation}}
\newcommand{\ee}{\end{equation}}
\newcommand{\eq}[1]{\eqref{#1}}

\renewcommand{\phi}{\varphi}
\renewcommand{\epsilon}{\varepsilon}

\newcommand{\cO}{\mathcal{O}}

\newcommand{\Nhomeo}{N^{h}}
\newcommand{\Gonehomeo}{G_1^{h}}
\newcommand{\Gtwohomeo}{G_2^{h}}
\newcommand{\Gprod}{G_{\!\textit{prod}}}
\newcommand{\VN}{V_{\!N_{\!M\!}}}
\newcommand{\aNM}{a_{N_{\!M}}}
\newcommand{\tauNM}{\tau_{N_{\!M\!}}}
\newcommand{\gammaNM}{\gamma_{N_{\!M}}}
\newcommand{\bNP}{b_{N_{\!P}}}
\newcommand{\tauNP}{\tau_{N_{\!P}}}
\newcommand{\etaNP}{\eta_{N_{\!P}}}
\newcommand{\etaNPhomeo}{\eta_{N_{\!P}}^{h}}

\newcommand{\tauQ}{\tau_Q}
\newcommand{\tauN}{\tau_N}

\newcommand{\kappaNhomeo}{\kappa^{h}}
\newcommand{\kappaN}{\kappa}
\newcommand{\NR}{N_{\!R}}
\newcommand{\NRhomeo}{\NR^{h}}

\newcommand{\Ntott}{[\NR(t)+N(t)]}

\newcommand{\gammaNR}{\gamma_{N_{\!R}}}

\newcommand{\ftrans}{\phi_{N_{\!R}}}
\newcommand{\ftranshomeo}{\phi_{N_{\!R}}^{h}}

\newcommand{\kANC}{k_\textit{\tiny ANC}}
\newcommand{\kP}{k_\textit{\tiny P}}
\newcommand{\kcirc}{k_\textit{circ}}
\newcommand{\EDrug}{E_\textit{Drug}}

\newcommand{\TG}{t_\textit{\tiny G}}
\newcommand{\TN}{t_\textit{\tiny N}}

\newcommand{\brackthat}{(\hspace{.05em}\hat{t}\hspace{.05em})}
\newcommand{\that}{\hat{t}}

\newcommand{\TimeDeriv}{\frac{\textrm{d}}{\textrm{d}t}}
\newcommand{\TimeDerivD}{\dfrac{\textrm{d}}{\textrm{d}t}}

\journalname{\small Journal of Pharmacokinetics and Pharmacodynamics}

\begin{document}

\title{Transit and lifespan in neutrophil production: implications for drug intervention
}


\author{Daniel~C\^amara~De~Souza* \and
	Morgan~Craig*        \and
	Tyler~Cassidy \and
         Jun~Li \and
         Fahima~Nekka \and
         Jacques~B\'elair \and
          Antony~R~Humphries
}

\authorrunning{Daniel C\^amara De Souza, Morgan Craig et al.} 

\institute{Daniel C{\^a}mara De Souza \at
            *Co-first author\\
              Department of Mathematics \& Statistics, McGill University, Montreal, QC, Canada, H3A 0B9\\
              \email{daniel.desouza@mail.mcgill.ca}
           \and
	Morgan Craig \at
	             *Co-first author\\
              Program for Evolutionary Dynamics, Harvard University, Cambridge, MA, USA, 02138\\
              \email{morganlainecraig@fas.harvard.edu}          
           \and
           Tyler Cassidy \at
              Department of Mathematics \& Statistics, McGill University, Montreal, QC, Canada, H3A 0B9\\
              \email{tyler.cassidy@mail.mcgill.ca}
            \and
           Jun Li \at
              Facult{\'e} de Pharmacie, Universit{\'e} de Montr{\'e}al, Montr{\'e}al, QC, Canada, H3C 3J7\\
              \email{jun.li.2@umontreal.ca}
              \and
           Fahima Nekka \at
              Facult{\'e} de Pharmacie, Universit{\'e} de Montr{\'e}al, Montr{\'e}al, QC, Canada, H3C 3J7\\
              \email{fahima.nekka@umontreal.ca}
            \and
           Jacques B{\'e}lair \at
              D{\'e}partement de math{\'e}matiques et de statistique, Universit{\'e} de Montr{\'e}al, Montr{\'e}al, QC, Canada, H3T 1J4\\
              \email{jbelair@dms.umontreal.ca}
               \and
           Antony R Humphries \at
              Departments of Mathematics \& Statistics, and Physiology, McGill University, Montreal, QC, Canada, H3A 0B9\\
              \email{tony.humphries@mcgill.ca}
}

\date{}

\maketitle

\begin{abstract}
We compare and contrast the transit compartment ordinary differential equation modelling approach with distributed and discrete delay differential equation models.
We focus on Quartino's extension to the Friberg transit compartment model of myelosuppression, widely relied upon in the pharmaceutical sciences to predict the neutrophil response after chemotherapy, and on a QSP delay differential equation model of granulopoiesis. We extend the Quartino model by considering a general number of transit compartments and introduce an extra parameter which allows us to decouple the maturation time from the production rate of cells,
and review the well established linear chain technique from the delay differential equation (DDE) literature which can be used to reformulate transit compartment models with constant transit rates as distributed delay DDEs. We perform a state-dependent time rescaling of the Quartino model in order to apply the linear chain technique and rewrite the Quartino model as a distributed delay DDE, which yields a discrete delay DDE model in a certain parameter limit. We then perform stability and bifurcation analyses on the models to situate such studies in a mathematical pharmacology context.

We show that both the original Friberg and the Quartino extension model incorrectly define the mean maturation time, essentially treating the proliferative pool as an additional maturation compartment, which can have far reaching consequences on the development of future models of myelosuppression in PK/PD.

\end{abstract}

%
%

\section{Introduction}
\label{sec:intro}

In the pharmaceutical sciences, the concept of lag time, or the delay between the administration and the absorption of a drug, is a
well-established phenomenon which is often accounted for \cite{Steimer1982}.
Physiologically-based pharmacokinetic models incorporating absorption models like the ACAT or ADAM  \cite{Agoram2001,Jamei2009} were indeed conceived and developed in part to account for the enterohepatic circulation that contributes to the delay in drug concentrations in the blood after oral administration. However, regardless of the administration of a xenobiotic, various forms of delays are present throughout physiological systems. In addition to pharmacokinetic lags, systems-level delays play an important role in determining the pharmacodynamic response to treatment. As examples, intracellular and intrinsic viral delays contribute to more complicated viral load decay in patients with human immunodeficiency virus being treated with antiretroviral drugs \cite{Dixit2004}, and the hematopoietic system displays multiple delays along the pathways from the pluripotent hematopoietic stem cells (HSCs) to terminally differentiated circulating cells \cite{Mackey1990}.

Granulopoiesis, the process of neutrophil production, in particular, exhibits multiple delays and has been studied in depth owing to the role neutrophils play in the innate (and adaptive) immune response \cite{Mantovani2011}. Within the pharmaceutical context, neutropenia is a toxic side effect of chemotherapy, and impacts heavily on treatment success and overall survival outcomes \cite{Sternberg2006,Gruber2011}.
There is therefore an established interest in mathematical models that can predict the response to chemotherapeutic drugs \cite{Friberg2002,Foley2008,Schirm2014} and accurately represent the feedback mechanisms regulating neutrophil homeostasis \cite{Craig2016c,Krinner2013}.

To maintain basal circulating neutrophil concentrations, multipotent progenitor HSCs in the bone marrow differentiate into the myeloid lineage on their way to becoming circulating neutrophils. After commitment, cells proliferate and undergo several divisions during a phase where cell numbers increase exponentially. After proliferation, neutrophil progenitors no longer divide. Instead, they grow in size and number of receptors before being sequestered into a marrow reservoir \cite{Rankin2010a}, where they either die through apoptosis or transit into circulation \cite{Christopher2007}. Once they exit from the bone marrow,  neutrophils circulate very transiently, with a half-removal time on the order of 7-10 hours \cite{VonVietinghoff2008}, as they either rapidly die or marginate into tissues \cite{Rankin2010a}. Granulopoiesis is controlled by various cytokines, of which granulocyte colony-stimulating factor (G-CSF) is the principal actor \cite{Ward1999a}. By binding to receptors on the neutrophil membranes, G-CSF regulates the rate at which neutrophils are released into circulation, and modulates up-stream factors (differentiation into the myeloid lineage, proliferation of upstream neutrophil progenitors, speed of maturation) to replenish and regulate the concentration of neutrophils in the bone marrow reservoir. G-CSF is then internalised by the neutrophils and removed from circulation. In the case of elevated circulating concentrations, G-CSF is also cleared via a linear, renal pathway \cite{Lyman2011} and these dual routes of elimination are important determinants of the PKs of G-CSF \cite{Craig2016c}. An overview of the process of neutrophil production is given in Figure~\ref{fig:NeutrophilDevelopment}.
\begin{figure}[!ht]
\begin{center}
\includegraphics[scale=0.4]{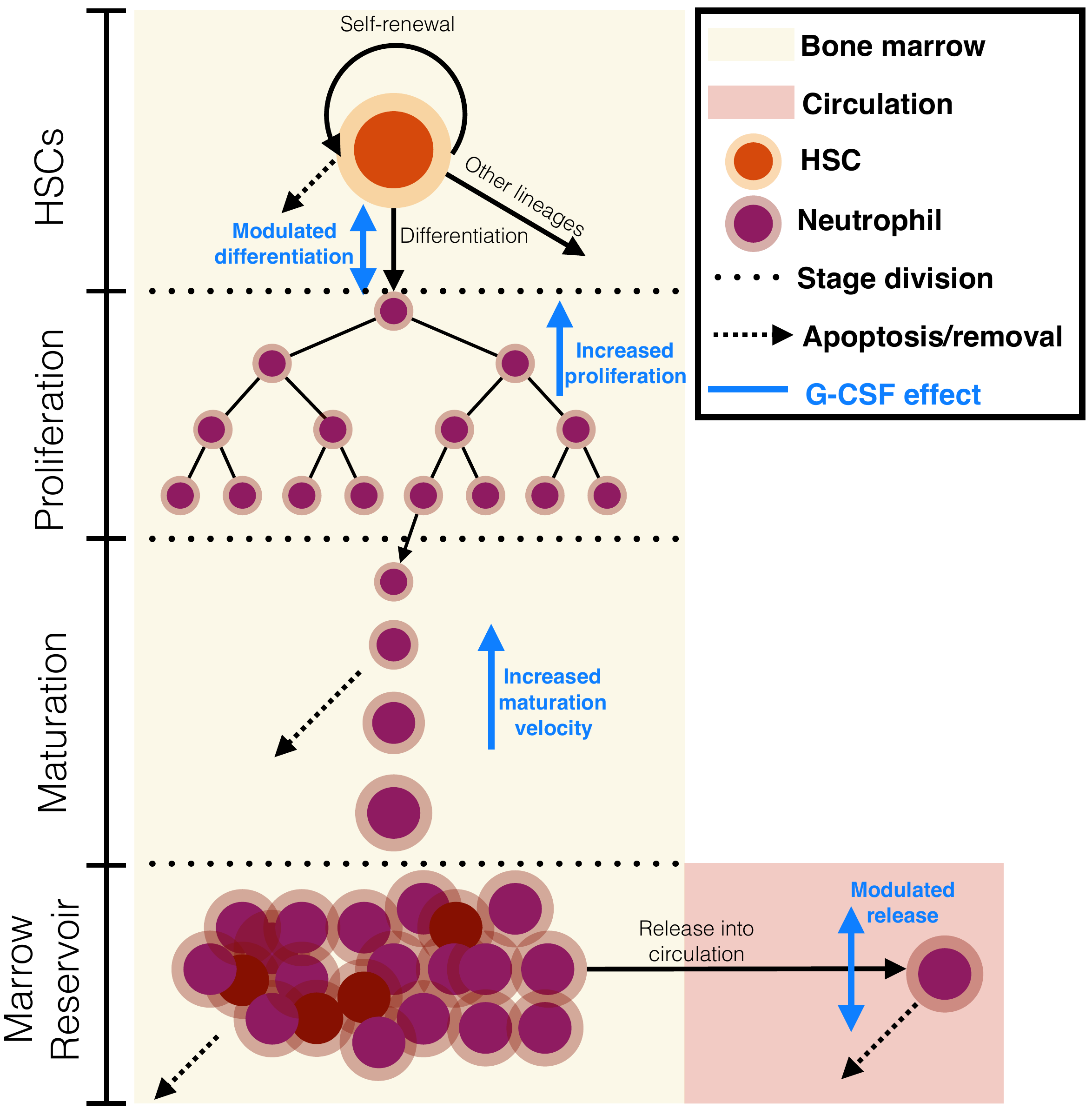}
\caption{An overview of granulopoiesis. As with all blood cells, neutrophils begin as hematopoietic stem cells (HSCs--orange circle) in the bone marrow (pale yelow background), where they develop. HSCs are capable of self-renewal and are subject to cell death (dashed arrows). HSCs may also differentiate into one of the blood cell lines, including the neutrophils (purple circles). After commitment to the neutrophil lineage, cells undergo a period of proliferative expansion after which they no longer divide. Post-mitotic neutrophils then mature, growing in size and gaining receptors. At the end of the maturation process, cells are then stored in the bone marrow reservoir from which they egress to reach circulation (pale red background) before removal (by margination or death). G-CSF acts to modulate the rate of exit from the marrow reservoir, increase the rates of maturation and proliferation, and to modulate the rate of differentiation into the neutrophil lineage (G-CSF actions represented by blue vertical arrows). Figure reproduced from ``Towards quantitative systems pharmacology models of chemotherapy-induced neutropenia'',  \textit{CPT: Pharmacometrics and Systems Pharmacology}, 2017 (to appear), Craig, M., \cite{Craig2017} with the permission of Wiley.}
\label{fig:NeutrophilDevelopment}
\end{center}
\end{figure}

Mathematical representations of granulopoiesis (and other similar physiological delay systems) fall into three classes: transit compartment models where delays are represented via a chain of first-order ordinary differential equations (ODEs), distributed delay systems where integro-differential equations represent a delay that takes a range of values determined through some probability distribution \cite{Campbell2009,Hearn1998,Adimy2007,Vainstein2005}, or delay differential equation (DDE) systems where the present state depends on past states via fixed or state-dependent delays \cite{Foley2008,Brooks2012,Craig2016c} (for more detailed discussions on the various models used in modelling hematopoiesis and chemotherapy-induced neutropenia, see \cite{Pujo-Menjouet2016} and \cite{Craig2017}, respectively).

Here we focus on two models of granulopoiesis in particular: the Quartino model \cite{Quartino2014} and the
Quantitative systems pharmacology (QSP) model of Craig~\cite{Craig2016c}.
The Quartino transit compartment ODE model accounts for the effects and PKs of endogenous G-CSF
and is an extension of the widely-used Friberg model \cite{Friberg2002,Friberg2003}, while the
QSP granulopoiesis model of~\cite{Craig2016c} is a state-dependant delay DDE model that incorporates the concentrations of unbound G-CSF and G-CSF bound to its neutrophil receptors. We will show that the
Quartino model \cite{Quartino2014} can be reformulated as a distributed delay DDE, which becomes a discrete-delay DDE in a certain parameter limit. This reformulation of the Quartino model leads to some additional insight on parameter choices and will lead us to generalise this model.

Since the maintenance of homeostasis or the pathogenic shift towards disease-states depend on the longterm behaviour of a given system's steady states, stability is an integral concept in physiology. In what follows, we will study the stability of these three major granulopoiesis model-types (transit compartment, distributed and discrete delay, and QSP) by demonstrating the relationships and equivalencies between all three formalisms and analysing the resulting distributed delay model to provide a better understanding of the role model selection plays within a treatment context. Accordingly, we will discuss how these stability results can impact the incorporation and delineation of the effects of interindividual variability. We will also provide a historical context for the origins of transit compartment models from distributed delay models and DDEs.

This paper is divided as follows. We begin in Section~\ref{sec:QuartinoModel} with an extension to the common ODE transit compartment model before introducing the more general distributed and discrete delay models in Section~\ref{sec:GammaDelayModels}. Therein, we discuss the linear chain technique (Section~\ref{sec:lct}) used to recover a transit compartment model from distributed delay systems. We then briefly introduce our previously published QSP model including endogenous G-CSF negative feedback in Section~\ref{sec:sddde}. The stability of the transit compartment/distributed/discrete delay models and the QSP model is analysed in Section~\ref{sec:stability} before we undertake bifurcation analyses (Section~\ref{sec:bifurcation}), which we discuss within the pharmaceutical sciences context in Section~\ref{sec:variability}. We conclude by discussing our results in Section~\ref{sec:discussion}. Many of the proofs are provided in the appendices at the end.


\section{Modelling granulopoiesis: three different approaches to handling delays}
\subsection{Transit compartment model with endogenous G-CSF}
\label{sec:QuartinoModel}

The Friberg model \cite{Friberg2002} is perhaps the most well-known model of myelosuppression after chemotherapy in the pharmaceutical sciences \cite{Craig2017}. Five compartments are used to represent the HSCs and early progenitors, circulating neutrophils, and the transit between the proliferative and circulative states. A feedback mechanism on the rate of proliferation determines the extent of myelosuppression of the chemotherapeutic agent. The model has been shown to generically represent a variety of chemotherapeutic drugs \cite{Friberg2003} and has been widely adopted in PK/PD studies of anti-cancer drugs. We write a generalised version of this model as
\begin{align} \notag
\frac{\textrm{d}P}{\textrm{d}t}&=\left(\kP(1-\EDrug)\left(\frac{N_0}{N(t)}\right)^{\!\gamma}-k_{tr}\right)P\\ \label{gdd2}
\frac{\textrm{d}T_1}{\textrm{d}t}&=k_{tr}P-aT_1\\ \notag
\frac{\textrm{d}T_j}{\textrm{d}t}&=a(T_{j-1}-T_j), \qquad j=2,\ldots,n\\ \notag
\frac{\textrm{d}N}{\textrm{d}t}&= aT_n-\kcirc N,
\end{align}
which reduces to the Friberg model if we set $\kP=k_{tr}=a$ and $n=3$.
Here, $P$ is the concentration of proliferating progenitors, $T_j$ is the $j^{\text{th}}$ post-mitotic transit compartment, and $N$ is the circulating neutrophil concentration (all in units of 10$^9$ cells/L), while $\kP$ is the rate of proliferation in the progenitor cell pool, $k_{tr}$ and $a$ are the transit rates between the maturation compartments, and $\kcirc$ is the rate of neutrophil exit from circulation (all in units of h$^{-1}$).

\begin{figure}[t]
\begin{center}
\includegraphics[width=\textwidth]{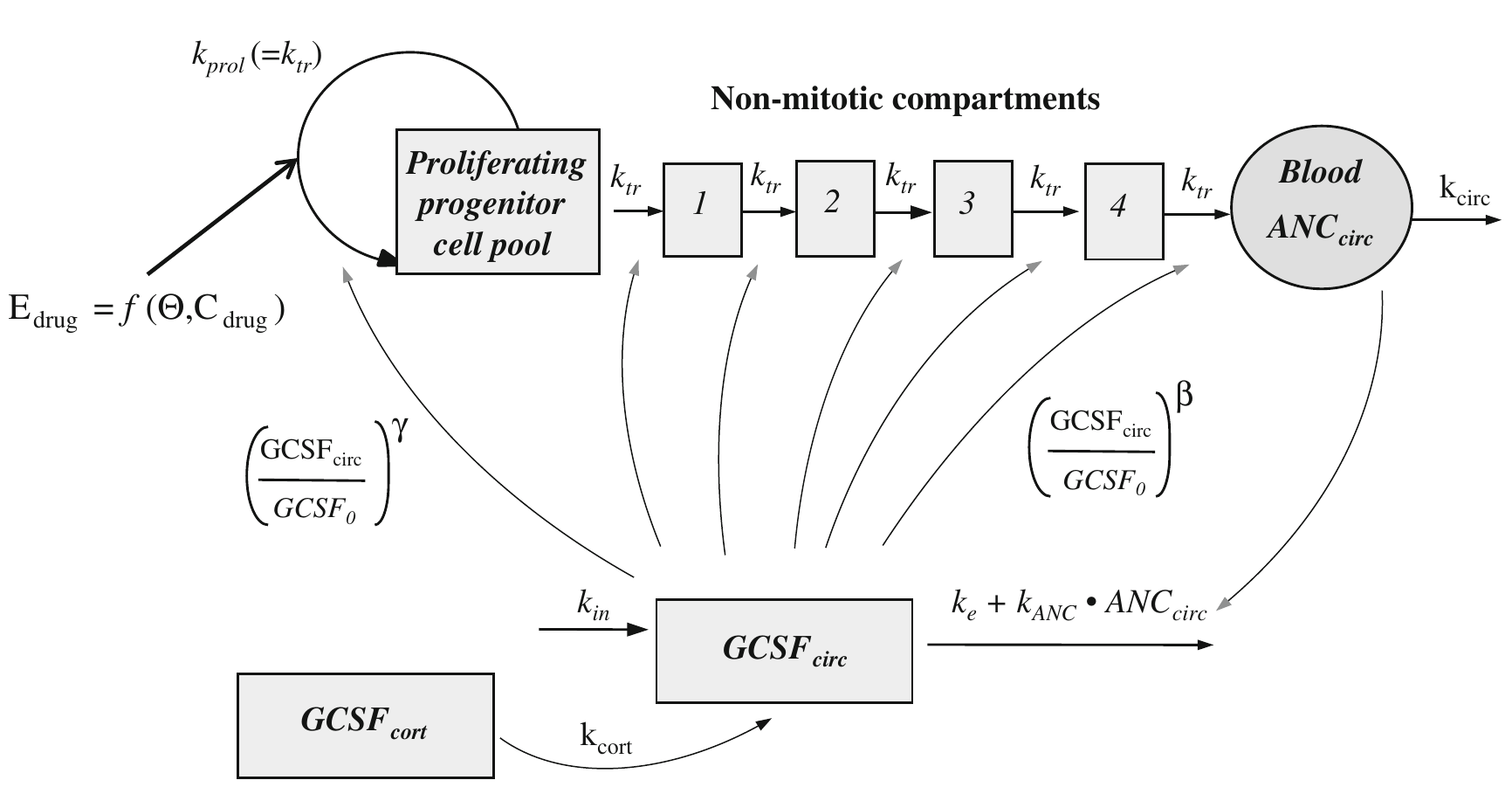}
\caption{The integrated G-CSF-myelosuppression model describing the dynamics of endogenous G-CSF and neutrophils following chemotherapy. For the myelosuppression model the parameters are baseline neutrophil count (ANC$_0$), mean maturation time (MMT=5/k$_{\text{tr}}$), the half-life of neutrophils in circulation (t$_{\nicefrac{1}{2}_{\text{circ}}}$=$\ln$(2)/$\kcirc$), the feedback parameters of G-CSF on neutrophil proliferation ($\gamma$) and transit time ($\beta$) and the drug related effect (E$_{\text{drug}}$). The estimated parameters for the G-CSF turnover model are baseline G-CSF (GCSF$_0$), nonspecific elimination rate constant (k$_\text{e}$) and ANC-dependent elimination rate constant ($\kANC$) and cortisol-induced G-CSF release (DOSE$_\text{cort}$) and the half-life of cortisol-induced G-CSF release (t$_{\nicefrac{1}{2}_{\text{cort}}}$=$\ln(2)/k_{\text{cort}}$). Figure reproduced from ``Characterization of endogenous G-CSF and the inverse correlation to chemotherapy-induced neutropenia in patients with breast cancer using population modeling'',  \textit{Pharmaceutical Research, 31}, 2014, pp. 3396, Quartino, A.L. et~al., \cite{Quartino2014} with the permission of Springer.}
\label{fig:QuartinoFigure}
\end{center}
\end{figure}

An extension to the Friberg model, which we will refer to as the Quartino model, is presented in \cite{Quartino2014} and models the myelosuppressive effects of chemotherapy on progenitor and circulating neutrophils, the endogenous G-CSF response, and the effect of the administration of a glucocorticoid to induce a rapid increase in G-CSF.
A model schematic is given in Figure~\ref{fig:QuartinoFigure}.
For our purposes, we can discount the administration of the glucocorticoid prior to chemotherapy and ignore the corresponding model terms. We write a generalised version of the model as
\begin{subequations} \label{eq:QuartinoModelLimited}
\begin{align} \label{eq:QuartinoModelLimiteda}
\frac{\textrm{d}P}{\textrm{d}t}&=P\left(\kP(1-\EDrug)\left (\frac{G}{G_{0}}\right)^\gamma-k_{tr}\left(\frac{G}{G_{0}}\right)^\beta\right)\\ \label{eq:QuartinoModelLimitedb}
\frac{\textrm{d}T_1}{\textrm{d}t}&=\left(\frac{G}{G_{0}}\right)^\beta (k_{tr}P-aT_1)\\ \label{eq:QuartinoModelLimitedc}
\frac{\textrm{d}T_j}{\textrm{d}t}&=a\left(\frac{G}{G_{0}}\right)^\beta (T_{j-1}-T_j), \qquad j=2,\ldots,n\\ \label{eq:QuartinoModelLimitedd}
\frac{\textrm{d}N}{\textrm{d}t}&=a\left (\frac{G}{G_{0}}\right)^\beta T_n-\kcirc N\\ \label{eq:QuartinoModelLimitede}
\frac{\textrm{d}G}{\textrm{d}t}&=k_{in}-(k_{e}+\kANC N)G,
\end{align}
\end{subequations}
where $G$ is the circulating G-CSF concentration (ng/L), $\kANC$ is the neutrophil-dependent rate of G-CSF elimination (h$^{-1}$), $k_e$ is the G-CSF nonspecific elimination rate (h$^{-1}$), $(G/G_0)^\gamma$ is the feedback on the 
proliferation rate
from circulating G-CSF concentrations, and $(G/G_0)^\beta$ reflects the G-CSF feedback on the
maturation rate.
In most of the current work we do not consider the chemotherapeutic agent and set $\EDrug=0$, unless otherwise stated.

We let $P_0$, $N_0$ and $G_0$ denote the homeostasis values of $P$, $N$ and $G$ respectively, obtained by setting
\be \label{Qdteq0}
\frac{\textrm{d}P}{\textrm{d}t}=\frac{\textrm{d}T_j}{\textrm{d}t}=\frac{\textrm{d}N}{\textrm{d}t}=\frac{\textrm{d}G}{\textrm{d}t}=0
\ee
in \eq{eq:QuartinoModelLimited}. In both the Friberg and Quartino models, it is a modelling assumption that
\be \label{eq:kinc}
\kP=k_{tr}.
\ee
The condition \eq{eq:kinc} is required in the Friberg model \eq{gdd2} to ensure that $N=N_0$ at homeostasis, and in the Quartino model \eq{eq:QuartinoModelLimited} to ensure that $G=G_0$ at homeostasis.
If $G_0$ were not the homeostasis value of $G$, it would be hard to justify the $(G/G_0)^\beta$ terms
appearing throughout the model, and the model ought to take a different form. Consequently we enforce the condition \eq{eq:kinc}
throughout, and always assume that
$k_{tr}=\kP$ as in \cite{Quartino2014}.

To see why we generalise the model by including a new parameter $a$, note that at homeostasis the rate of production of proliferating cells, the rate that cells leave proliferation to enter the first transit compartment, the rate they leave the last transit compartment to enter circulation and the rate that they leave circulation must all be equal. In both
models this results in
\be \label{eq:homeoprod}
\kP P_0 = k_{tr} P_0 = \kcirc N_0.
\ee
The production rate in \eq{eq:homeoprod} is completely independent of the maturation time of the cells; provided cells both enter and leave maturation at the rate given by \eq{eq:homeoprod}, changing the maturation time $\tau$ would only change the total number of cells that are in maturation (which is $\tau\kP P_0$), but will not change the production rate in \eq{eq:homeoprod}. We will see below that at homeostasis the maturation time $\tau$ for both models is given by
\be \label{mattime}
\tau=\frac{n}{a}.
\ee
Fixing $a=k_{tr}$ leads to two related modelling problems. First if we regard $a=k_{tr}$ as known then, since $n$ is an integer, equation \eq{mattime} only allows for certain discrete values of the delay $\tau$. On the other hand, if as is more usual we suppose that $\tau$ is known then choosing $n$ an integer and imposing that $a=k_{tr}$ in \eq{mattime} uniquely determines the value of $k_{tr}$ in \eq{mattime}, which in turn determines the production rate in \eq{eq:homeoprod}. But we already noted that the production rate at homeostasis $\kP P_0$ and the maturation time $\tau$ are independent.

\begin{table}[t]
\centering
\begin{tabular}{|c|c|c|}
\hline
\textbf{Parameter} & \textbf{Units}   & \textbf{Typical estimate (\% RSE)} \\
\hline
$N_0$              & 10$^9$ cells/L   & 3.53 (5)                           \\
$k_{tr}=\kP$       & h$^{-1}$         & 0.03759                            \\
$\gamma$           & -                & 0.444 (4)                           \\
$\beta$            & -                & 0.234 (8)                          \\
$G_0$              & ng/L             & 24.3 (8)                           \\
$k_e$              & h$^{-1}$         & 0.592 (32)                         \\
$\kANC$            & h/10$^9$ cells/L & 5.64                               \\
$\kcirc$           & h$^{-1}$         & 0.099     \\
$k_{in}$           & h$^{-1}$               &  498.1792       \\
\hline
\end{tabular}
\caption{Parameter values from \cite{Quartino2014} for the parameters of interest in this study. The mean value of the distributed delay in \eqref{lctsd3} is given by $\tau=n/a$, or $\tau=106.41$ hours.}
\label{tab:QuartinoValues}
\end{table}

In Quartino \cite{Quartino2014} a mean maturation time is defined by $MMT=(n+1)/k_{tr}$. Presumably the authors counted $n$ transit compartments plus one proliferation compartment. By showing the equivalence of the generalised Quartino model \eq{eq:QuartinoModelLimited} to a distributed delay DDE in Section~\ref{sec:GammaDelayModels} we will find both the mean and variance of the delay, and show that even if $a=k_{tr}$ the correct formula for the
mean maturation time should be $MMT=n/k_{tr}$, corresponding to \eq{mattime}, and not the formula used in Quartino \cite{Quartino2014}.

In the following sections we will consider general values of the parameters $a$ and $n$, but will take the values of the remaining parameters from \cite{Quartino2014}; these values are tabulated in Table~\ref{tab:QuartinoValues}.
To satisfy the homeostasis conditions \eq{Qdteq0} we obtain
\be \label{TjPss}
T_j=\frac{N_0 \kcirc}{a},\; j=1,\ldots,n, \qquad  P_0=\frac{N_0\kcirc}{k_{tr}},
\ee
and the parameter constraint
\be \label{eq:kinb}
k_{in}=G_0(k_e+ \kANC N_0).
\ee
At homeostasis the total number of cells in the $n$ maturation compartments is $N_0\kcirc n/a$. Dividing this by the production rate given by \eq{eq:homeoprod} gives the average maturation time $\tau=n/a$ as stated in \eq{mattime}.

Notice that if $a=k_{tr}$ as in \cite{Quartino2014} then at steady state we have $T_j=P$ for all the transit compartments. In \cite{Quartino2014} the model \eq{eq:QuartinoModelLimited} is considered with initial conditions at time $t=0$ equal to the steady-state values (which is natural for a chemotherapy study before the chemotherapeutic agent is administered), but we will consider the behaviour of the model for general non-negative initial conditions. Proof of the positivity of solutions to the Quartino system~\eqref{eq:QuartinoModelLimited} can be found in Appendix~\ref{sec:QuartinoPositivity}.
%



%
%

\subsection{Gamma-distributed delay representation of the transit compartment granulopoiesis models}
\label{sec:GammaDelayModels}

Distributed delay DDEs come in many varieties, but a reasonably general form is
\begin{align} \label{lct1}
\frac{\textrm{d}N}{\textrm{d}t}
&=f\!\left(t,N(t),\!\int_{-\infty}^t \hspace{-0.7em}P(s)g_a^p(t-s)ds\right)\\
&=f\!\left(t,N(t),\!\int_0^\infty \hspace{-0.7em}P(t-u)g_a^p(u)du\right). \notag
\end{align}
In simpler examples $P(t)\equiv N(t)$, but $P(t)$ can also be a separate variable defined by its own differential equation (as is the case in the granulopoiesis models considered in this work). The function $g_a^p(u)$ is a probability density with
\be \label{lct1c}
\int_0^\infty \hspace{-0.7em}g_a^p(u)du=1.
\ee
So, rather than the dynamics of $N(t)$ being determined by the current value of $P(t)$, the integral distributes the effect of $P$ across its previous values. In this work we will restrict attention to the gamma distribution, though other distributions do arise, in particular the uniform distribution.
We write the probability density function $g_a^p$ of the gamma distribution as
\be \label{lct2}
g_a^p(t)=\frac{a^pt^{p-1}e^{-at}}{\Gamma(p)},
\ee
where $\Gamma(p)$ is the gamma function. When $n$ is a positive integer $\Gamma(n)=(n-1)!$, and the gamma function
generalises the factorial function to real numbers $p$ with
$\Gamma(p)=(p-1)\Gamma(p-1)$ for any $p>0$. The real positive parameters $a$ and $p$ determine the shape and rate of the distribution with the mean delay $\tau$ given by
\be \label{gapmean}
\tau=p/a,
\ee
and standard deviation $\sigma^2=p/a^2$. If $p$ and $a$ are taken to infinity with their ratio $\tau$ held constant then
the variance decreases to zero and
the probability density function $g_a^p(t)$ becomes narrower and taller and approaches the $\delta$-function $\delta(t-\tau)$. In this limit the distributed delay DDE \eq{lct1} reduces to a discrete delay DDE
\be \label{lct.int}
\frac{\textrm{d}N}{\textrm{d}t}=f(t,N(t),P(t-\tau)).
\ee
So discrete delay DDEs can be thought of as a limiting case of distributed delay DDEs.  We will see below that when $p=n$ an integer, we can rewrite a gamma distributed DDE as an ODE, so gamma distributed DDEs provide a link between ODEs and discrete delay DDE models.

\subsubsection{The Linear Chain Technique}
\label{sec:lct}

The linear chain technique is used to convert some distributed delay differential equations (DDEs) into a corresponding system of ordinary differential equations (ODEs), or \emph{vice versa}. The technique dates back at least to the work of Vogel in the 1960s \cite{Vogel63,Vogel65}, and first appears in the English literature in the work of MacDonald \cite{MacDonald78,MacDonald89} who called the method the \emph{linear chain trick}.
Most authors continue to use that name, but we prefer \emph{linear chain technique}, because, as we will see,
there is a true equivalency between the differential equation systems, and no trick is involved. It is usually more convenient to formulate problems as ODEs for numerical simulation, but sometimes more convenient to formulate them as DDEs for analysis. The linear chain technique is well-known and used in population biology and mathematical epidemiology, but is as yet not as well-known in other fields. The method has been independently rediscovered several
times over the decades, being referred to as the fixed boxcartrain method by Goudriaan~\cite{Goudriaan1986},
and recently used by Krzyzanski~\cite{Krzyzanski2011} in a pharmaceutical sciences setting.
There are several variants on this technique, and descriptions can be found in many places including
\cite{Jacquez2002,MacDonald89,Smith2011},
but the simplest application is for a gamma distributed delay, for which we will detail the steps here.

The probability density function \eq{lct2}
has the property that for $p\ne1$
\begin{align} \notag
\TimeDeriv g_a^p(t)&=\frac{(p-1)a^pt^{p-2}e^{-at}}{\Gamma(p)}-\frac{a^{p+1}t^{p-1}e^{-at}}{\Gamma(p)}\\ \notag
&=a\left(\frac{a^{p-1}t^{p-2}e^{-at}}{\Gamma(p-1)}-\frac{a^{p}t^{p-1}e^{-at}}{\Gamma(p)}\right)\\
&=a(g_a^{p-1}(t)-g_a^p(t)). \label{lct3}
\end{align}
While for $p=1$
\be \label{lct3a}
\TimeDeriv g_a^1(t)=\TimeDeriv (ae^{-at})=-a g_a^1(t).
\ee
Models of the form \eq{lct1},\eq{lct2} can in principle be considered for any real positive value of $p$, but in practice nearly all authors only consider $p=n$ a positive integer (one exception is \cite{Campbell2009}), because
then equations \eq{lct3},\eq{lct3a} allow the distributed DDE to be reduced to an ODE. To do this let
\be \label{lct4}
T_j(t)=\int_{-\infty}^t \hspace{-0.5em}P(s)g_a^j(t-s)\,ds = \int_0^\infty \hspace{-0.5em}P(t-u)g_a^j(u)\,du,
\qquad j=1,\ldots,n.
\ee
Then equation \eq{lct1} can be rewritten as an ODE
\be \label{lct1a}
\frac{\textrm{d}N}{\textrm{d}t}=f\left(t,N(t),T_n(t)\right).
\ee
Differentiating \eq{lct4}, using Leibniz rule for $j>1$ (noting that $g_a^j(0)=0$ for $j>1$) we obtain
\begin{align}
\frac{\textrm{d}T_j}{\textrm{d}t}  \notag
&=P(t)g_a^j(0)+\int_{-\infty}^t \hspace{-0.5em} P(s)\TimeDeriv g_a^j(t-s)\,ds\\ \notag
&=\int_{-\infty}^t \hspace{-0.5em}P(s)a(g_a^{j-1}(t-s)-g_a^j(t-s))\,ds\\
&=a(T_{j-1}(t)-T_j(t)), \quad j=\{2,3,\ldots,n\}.  \label{lct5}
\end{align}
While for $j=1$ (noting that $g_a^1(0)=a$)
\be \label{lct5a}
\frac{\textrm{d}T_1}{\textrm{d}t}
= P(t)g_a^1(0)+\int_{-\infty}^t\hspace{-0.5em} P(s)\TimeDeriv g_a^1(t-s)\,ds \\
 = a(P(t)-T_1(t)).
 \ee
Together equations \eq{lct1a},\eq{lct5},\eq{lct5a} redefine the (nonlinear) distributed delay DDE \eq{lct1} as a system of $n+1$ ODEs. General DDEs can be posed as infinite dimensional dynamical systems, which introduces considerable mathematical difficulties, so being able to reduce some DDE models to finite-dimensional ODEs is mathematically very advantageous.

To complete the relationship between the distributed delay DDE \eq{lct1} and the system of ODEs \eq{lct1a},\eq{lct5},\eq{lct5a} we should take some care with the initial conditions. The distributed DDE \eq{lct1} has infinite memory, and so to solve as an initial value problem from time $t=0$ we need to define a history function $P(t)$ for all $t\leq0$, so that the right hand-side of \eq{lct1} can be evaluated. With $P(t)$ so defined, for the DDE and ODE reduction to have equivalent
solutions, by \eq{lct4} the ODE must have initial conditions
\be \label{lct40}
T_j(0)=\int_{-\infty}^0 \hspace{-0.5em}P(s)g_a^j(-s)\,ds
=\int_0^{\infty} \hspace{-0.5em}P(s)g_a^j(s)\,ds
\qquad j=1,\ldots,n.
\ee
If it is assumed that $P(t)=P_0$, a constant for all $t\leq0$ then, using \eq{lct1c}, we see that
\eq{lct40} reduces to
\be \label{lct40a}
T_j(0)=P_0\int_0^\infty \hspace{-0.5em}g_a^j(s)\,ds=P_0,
\qquad j=1,\ldots,n.
\ee

It is natural to ask if we can also go the other way; does a solution of the system of ODEs \eq{lct1a},\eq{lct5},\eq{lct5a}, define a solution of the distributed DDE \eq{lct1}? It follows immediately from \eq{lct40a} that a solution of the ODE system with initial conditions $P(0)=T_j(0)$ for $j=1,\ldots,n$ does define a solution of \eq{lct1}.
The equivalence for more general initial conditions for the ODE has also been established; in that case the ODE initial conditions define a finite number of constraints on the history function $P(t)$ for $t\leq0$, which do not uniquely define $P(t)$, and the
ODE defines a solution of the distributed DDE \eq{lct1} for all choices of $P(t)$ that satisfy the constraints
\cite{Cooke1982,MacDonald89}.

\subsubsection{Gamma-distributed and discrete delay representations of transit compartment granulopoiesis models}
\label{sec:GammaDelayQuartino}

The linear chain technique of Section~\ref{sec:lct} can be applied to establish the equivalence between transit compartment ODE models and corresponding distributed delay DDEs.
Consider first the distributed DDE system
\be \label{gdd1}
\begin{array}{l}
\displaystyle\frac{\textrm{d}P}{\textrm{d}t}=\left(\kP(1-\EDrug)\left(\frac{N_0}{N(t)}\right)^{\!\gamma}-k_{tr}\right)P\\
\displaystyle\frac{\textrm{d}N}{\textrm{d}t}=-\kcirc N+k_{tr}\!\int_{-\infty}^t \hspace{-0.25em}P(s)g_a^n(t-s)ds.\rule[-1mm]{0mm}{8mm}
\end{array}
\ee
We define $T_j(t)$ by
\be \label{lct4a}
T_j(t)=\int_{-\infty}^t \hspace{-0.25em}\frac{k_{tr}}{a}P(s)g_a^j(t-s)\,ds,
\qquad j=1,\ldots,n,
\ee
which corresponds to \eq{lct4} with $k_{tr}P(s)/a$ replacing $P(s)$. Writing the equation for $N(t)$ as
\begin{equation}
\label{eq:FribergDist}
\frac{\textrm{d}N}{\textrm{d}t}=-\kcirc N+a\!\int_{-\infty}^t \hspace{-0.25em}\frac{k_{tr}}{a}P(s)g_a^n(t-s)ds,
\end{equation}
and applying the linear chain technique of Section~\ref{sec:lct} we obtain
the generalised Friberg transition compartment model of myelosuppression \eq{gdd2}.
Taking $a=k_{tr}$ and $n=3$ gives the Friberg model as stated in \cite{Friberg2002}, as has already
been noted in \cite{Belair2015}.

While it is necessary to set $a=k_{tr}$ in \eq{gdd1} to recover the model as stated in \cite{Friberg2002}, equation \eq{gdd2} defines a transit compartment model for other values of $a$ also, and both the system of ODEs \eq{gdd2} and the distributed DDE \eq{gdd1} can be considered for general values $a>0$.

The extended Quartino endogenous G-CSF model \cite{Quartino2014} as stated in \eq{eq:QuartinoModelLimited} and discussed in Section~\ref{sec:QuartinoModel} cannot be stated simply as a distributed delay DDE via the linear chain technique. The maturation time in the Quartino model instead of being constant is state-dependent with the rate constants for the passage through each transit compartment given by
$$a\left(\frac{G(t)}{G_0}\right)^\beta,$$
which varies as $G(t)$ varies; it reduces to the same value 
as for the Friberg model 
only if $G(t)=G_0$.
In contrast, the derivation of \eq{lct3}, which is essential in the linear chain technique, requires that the rate constant $a$
(and the power $p$) be constant, so to apply the linear chain technique the profile of the probability density function must remain constant and cannot vary with time or the solution. Thus, while it might be tempting to consider a distributed DDE of the form
\begin{equation}
\label{eq:DistributedDelayModel}
\begin{split}
\frac{\textrm{d}P}{\textrm{d}t}&=P\left(\kP\left (\frac{G}{G_{0}}\right)^\gamma-k_{tr}\left(\frac{G}{G_{0}}\right)^\beta\right)\\
\frac{\textrm{d}N}{\textrm{d}t}&=-\kcirc N + k_{tr}\left(\frac{G}{G_{0}}\right)^\beta\int_{-\infty}^t \hspace{-0.5em}P(s)g_a^n(t-s)\,ds\\
\frac{\textrm{d}G}{\textrm{d}t}&=k_{in}-(k_{e}+\kANC N)G,
\end{split}
\end{equation}
if we set $a = k_{tr}(G(t))/G_0)^\beta$, then it is not possible to reduce this model to a system of ODEs because the derivation of \eq{lct3} fails when $a$ is time-dependent. Instead we could consider the model
\begin{equation} \label{lctq1}
\begin{split}
\frac{\textrm{d}P}{\textrm{d}t}&=P\left(\kP\left (\frac{G}{G_{0}}\right)^\gamma-k_{tr}\left(\frac{G}{G_{0}}\right)^\beta\right)\\
\frac{\textrm{d}N}{\textrm{d}t}&=-\kcirc N + k_{tr}\int_{-\infty}^t \hspace{-0.25em}\left(\frac{G(s)}{G_{0}}\right)^\beta P(s)g_a^n(t-s)\,ds\\
\frac{\textrm{d}G}{\textrm{d}t}&=k_{in}-(k_{e}+\kANC N)G,
\end{split}
\end{equation}
and apply the linear chain technique with
\be \label{lctq2}
T_j(t)=\int_{-\infty}^t \hspace{-0.5em}\frac{k_{tr}}{a}\left(\frac{G(s)}{G_{0}}\right)^\beta  P(s)g_a^j(t-s)\,ds, \qquad j=1,\ldots,n.
\ee
to obtain the transit compartment model
\be \label{lctq3}
\begin{split}
\frac{\textrm{d}P}{\textrm{d}t}&=P\left(\kP\left (\frac{G}{G_{0}}\right)^\gamma-k_{tr}\left(\frac{G}{G_{0}}\right)^\beta\right)\\
\frac{\textrm{d}T_1}{\textrm{d}t} & = k_{tr}\left(\frac{G}{G_{0}}\right)^\beta P(t)-aT_1(t)\\
\frac{\textrm{d}T_j}{\textrm{d}t}&=a(T_{j-1}(t)-T_j(t)), \quad j=\{2,3,\ldots,n\}\\
\frac{\textrm{d}N}{\textrm{d}t}&=a T_n(t)-\kcirc N\\
\frac{\textrm{d}G}{\textrm{d}t}&=k_{in}-(k_{e}+\kANC N)G,
\end{split}
\ee
which is similar to the Quartino model \eq{eq:QuartinoModelLimited}, but missing the $(G(t)/G_0)^\beta$ factors in
all the $T_j$ transit terms, and consequently does not model the effect of G-CSF on the maturation rate.
%

To write the Quartino model \eq{eq:QuartinoModelLimited} as a distributed DDE, we first remove the state-dependency of the delays by rescaling time.
Define a new time $\hat{t}(t)$ by
\be \label{lctsd1}
\frac{\textrm{d}\hat{t}}{\textrm{d}t}=\left(\frac{G(t)}{G_{0}}\right)^\beta, \qquad \hat{t}(0)=0.
\ee
By Theorem~\ref{ODETheorem}
the right-hand side of \eq{lctsd1} is strictly positive for $t>0$ so $\frac{\textrm{d}\hat{t}}{\textrm{d}t}>0$ and the new time variable $\hat{t}(t)$ is a strictly monotonic increasing function of $t$. Then we see that
$$\frac{\textrm{d}T_j}{\textrm{d}\hat{t}}=\frac{\textrm{d}t}{\textrm{d}\hat{t}}\frac{\textrm{d}T_j}{\textrm{d}t}=\left(\frac{G_0}{G\brackthat}\right)^\beta
a\left(\frac{G\brackthat}{G_{0}}\right)^\beta(T_{j-1}-T_j)=a(T_{j-1}-T_j).$$
Strictly speaking we should define new variables $\widetilde{G}\brackthat=G(t)$, but following common practice we suppress the tildes and reuse the same variable names.
Applying the same time-rescaling to all the equations we rewrite the Quartino model \eq{eq:QuartinoModelLimited} as
\be \label{lctsd2}
\begin{split}
\frac{\textrm{d}P}{\textrm{d}\hat{t}}&=\left(\kP\left (\frac{G\brackthat}{G_{0}}\right)^{\gamma-\beta}\hspace{-.5em}-\,k_{tr}\right)P\brackthat\\
\frac{\textrm{d}T_1}{\textrm{d}\hat{t}}&= k_{tr}P\brackthat-aT_1\brackthat\\
\frac{\textrm{d}T_j}{\textrm{d}\hat{t}}&=a(T_{j-1}\brackthat-T_j\brackthat), \qquad j=2,\ldots,n\\
\frac{\textrm{d}N}{\textrm{d}\hat{t}}&=a T_n\brackthat-\left(\frac{G_0}{G\brackthat}\right)^\beta\kcirc N\brackthat\\
\frac{\textrm{d}G}{\textrm{d}\hat{t}}&= \left(\frac{G_0}{G\brackthat}\right)^\beta\Bigl(k_{in}-(k_{e}+\kANC N\brackthat)G\brackthat\Bigr).
\end{split}
\ee
We refer to \eq{lctsd2} as the time-rescaled Quartino model. Since the time rescaling satisfies $\that(0)=0$, the initial conditions for the Quartino model \eq{eq:QuartinoModelLimited} at $t=0$ and the
time-rescaled Quartino model \eq{lctsd2} at time $\that=0$ are the same, and these two equations given equivalent solutions.

The time-rescaled Quartino model \eq{lctsd2} has constant transition rates between the transit compartments, and consequently we can apply the linear chain technique to derive  \eq{lctsd2} from
\be \label{lctsd3}
\begin{split}
\frac{\textrm{d}P}{\textrm{d}\hat{t}}&=\left(\kP\left (\frac{G\brackthat}{G_{0}}\right)^{\gamma-\beta}-k_{tr}\right)P\brackthat\\
\frac{\textrm{d}N}{\textrm{d}\hat{t}}&=-\left(\frac{G_0}{G\brackthat}\right)^\beta\kcirc N\brackthat
+ k_{tr}\int_{-\infty}^{\hat{t}} \hspace{-0.5em}P(s)g_a^n(\hat{t}-s)\,ds\\
\frac{\textrm{d}G}{\textrm{d}\hat{t}}&= \left(\frac{G_0}{G\brackthat}\right)^\beta\Bigl(k_{in}-(k_{e}+\kANC N\brackthat)G\brackthat\Bigr),
\end{split}
\ee
by letting
\be \label{lctsd4}
T_j\brackthat=\int_{-\infty}^{\hat{t}} \hspace{-0.5em}\frac{k_{tr}}{a}P(s)g_a^j(\hat{t}-s)\,ds, \qquad j=1,\ldots,n.
\ee
To define an initial value problem for the distributed delay DDE \eq{lctsd3} we need to specify $N(0)$, $G(0)$ and $P\brackthat$ for $\that\leq0$. This in turn defines initial conditions for both the time-rescaled Quartino model \eq{lctsd2}
and the Quartino model \eq{eq:QuartinoModelLimited}
with $T_j(0)$ given by evaluating \eq{lctsd4} with $\that=0$.
If $P\brackthat$ is constant for $\that\leq0$ then \eq{lctsd4} implies that $T_j(0)=k_{tr}P(0)/a$, so there is an immediate equivalence between all three models for such initial conditions. Even if the Quartino model \eq{eq:QuartinoModelLimited} were considered with different initial conditions, there is still a direct equivalence to the time-rescaled Quartino model \eq{lctsd2}, and as noted at the end of Section~\ref{sec:lct}, also to the distributed DDE model \eq{lctsd3}. Consequently we have three equivalent forms of the same model, with a direct correspondence between the solutions of the differential equation systems \eq{eq:QuartinoModelLimited} and \eq{lctsd2} and \eq{lctsd3}.

Recalling \eq{gapmean} the mean value of the distributed delay in \eq{lctsd3} is $\tau=n/a$. The time rescaling \eq{lctsd1} is trivial at homeostasis when $G(t)=G_0$, so this also implies that the mean maturation delay is $\tau=n/a$ in the
Quartino model \eq{eq:QuartinoModelLimited} (and fact that we already derived by a different argument in \eq{mattime}).
Fixing $a=k_{tr}$ only allows a very granular control of the mean delay in the ODE model by varying the integer $n$.
Mathematically it is more convenient to fix the delay $\tau>0$ and use $n$ and $a$ to control the shape of the distribution. For the distributed DDE model \eq{lctsd3} we do not even need $n$ to be an integer. Recalling \eq{lct1c}, in the limit as $n\to\infty$ and
$a\to\infty$ with $\tau=n/a$ fixed, the distributed delay DDE \eq{lctsd3} reduces to the discrete delay DDE
\be \label{lctsd5}
\begin{split}
\frac{\textrm{d}P}{\textrm{d}\hat{t}}&=\left(\kP\left (\frac{G\brackthat}{G_{0}}\right)^{\gamma-\beta}-k_{tr}\right)P\brackthat\\
\frac{\textrm{d}N}{\textrm{d}\hat{t}}&=-\left(\frac{G_0}{G\brackthat}\right)^\beta\kcirc N\brackthat
+ k_{tr}P(\hat{t}-\tau)\\
\frac{\textrm{d}G}{\textrm{d}\hat{t}}&= \left(\frac{G_0}{G\brackthat}\right)^\beta\Bigl(k_{in}-(k_{e}+\kANC N\brackthat)G\brackthat\Bigr).
\end{split}
\ee
We remark that in the discrete delay DDE \eq{lctsd5} the delay $\tau$ is constant in the rescaled time-variable $\hat{t}$, just as the (same) mean delay $\tau=n/a$ is constant in the distributed delay DDE~\eq{lctsd3}. In contrast the mean maturation time $\alpha(t)$ in the Quartino model \eq{eq:QuartinoModelLimited} varies with $G(t)$ and satisfies
$$\hat{t}(t-\alpha(t))=\hat{t}(t)-\tau,$$
where $\hat{t}(t)$ satisfies \eq{lctsd1}. If $G$ is held constant (but not necessarily equal to $G_0$), this gives a
mean maturation time $\alpha$ in the  Quartino model \eq{eq:QuartinoModelLimited} of
$$\alpha=\frac{n}{a(G/G_0)^\beta}=\frac{\tau}{(G/G_0)^\beta}.$$
For the case of time-varying $G(t)$, the evolution of the mean maturation delay $\alpha(t)$
is defined by a differential equation \eq{alphaev},
which we derive in Appendix~\ref{app.timerescale},
where we also show the similarities between this state-dependency and the explicit state-dependency in the QSP model \eq{eq:DDEmodel}.
But, in the current work, the time-rescaling equation \eq{lctsd1} will be sufficient for our purposes.

Notice that while the derivation of \eq{lctsd5} makes sense when considering the limit of the shape
of the probability density functions
$g_{n/\tau}^n(t)$ as they approach the $\delta$-function $\delta(t-\tau)$ when $n\to\infty$, it is problematical
if one interprets the distributed DDE \eq{lctsd3} via the ODE system \eq{lctq3} or \eq{lctsd2} since then the limiting process would correspond to taking the number of compartments $n$ to infinity while increasing the rate constants $a$ to infinity also.

Since the system \eqref{lctsd2} corresponds to the Quartino model \eq{eq:QuartinoModelLimited} with time rescaled by \eq{lctq1}, positivity of solutions is guaranteed by Theorem~\ref{ODETheorem}.
The correspondence between the
distributed delay system \eqref{lctsd3} and the system \eqref{lctsd2} ensures positivity of solutions of \eqref{lctsd3} for integer $n$ only, but actually positivity can be established for all real $n>0$. The proof of the positivity of solutions to \eqref{lctsd2} is given in Theorem~\ref{DDETheorem}
in Appendix~\ref{sec:QuartinoPositivity}.

\subsection{A QSP model of granulopoiesis and its regulation by G-CSF}
\label{sec:sddde}

As previously mentioned, DDEs are frequently relied upon to model granulopoiesis given the delays inherent to hematopoiesis. Here we focus on the Quantitative systems pharmacology
model of \cite{Craig2016c}, which has been shown to account for the dynamics of neutrophil production and its negative feedback relationship with G-CSF--both bound to receptors on the surface of neutrophils and freely circulating--in a variety of scenarios. The model is written as
\begin{subequations}\label{eq:DDEmodel}
\begin{align}
\TimeDerivD Q(t) &= -\bigl(\kappaN(G_1(t)) +\kappa_\delta + \beta(Q(t))\bigr)Q(t) \notag \\
&+ A_Q(t)\beta\left(Q(t-\tauQ)\right)Q(t-\tauQ)   \label{eq:HSCs} \\
\TimeDerivD \NR(t)&=  A_N(t) \kappaN(G_1(t-\tauN(t))) Q(t-\tauN(t))
 \frac{\VN(G_1(t))}{\VN(G_1(t-\tauNM(t)))} \notag \\
&-\bigl (\gammaNR +\ftrans(G_{BF}(t))\bigr )\NR(t) \label{eq:Reservoir} \\
\TimeDerivD N(t) &=   \ftrans(G_{BF}(t))\NR(t)-\gamma_N N(t), \label{eq:Neutrophils}\\
\TimeDerivD G_{1\hspace{-0.1em}}(t) & = I_G(t)
+\Gprod-k_{ren}G_{1\hspace{-0.1em}}(t)\\ \label{eq:FreeGCSF}
& -k_{12}(\Ntott V-G_{2\hspace{-0.05em}}(t))G_{1\hspace{-0.1em}}(t)^{\sG}\hspace{-0.1em}
+k_{21}G_{2\hspace{-0.05em}}(t) \\ \label{eq:BoundGCSF}
\TimeDerivD G_{2\hspace{-0.05em}}(t) & = -k_{int}G_{2\hspace{-0.05em}}(t)+k_{12}\bigl(\hspace{-0.05em}
[\NR(t) \hspace{-0.1em}+\hspace{-0.05em}N\hspace{-0.1em}(t)] V\hspace{-0.4em}-G_{2\hspace{-0.05em}}(t)\hspace{-0.1em}\bigr)G_{1\hspace{-0.1em}}(t)^{\sG}\hspace{-0.3em}-k_{21}G_{2\hspace{-0.05em}}(t)
\end{align}
\end{subequations}
Here $Q(t)$ is the concentration of HSCs ($10^6$ cells/kg), $\NR(t)$ the concentration of neutrophils in the bone marrow reservoir ($10^9$ cells/kg), $N(t)$ the concentration of circulating neutrophils ($10^9$ cells/kg), $G_1(t)$ the circulating G-CSF concentration (ng/mL), and $G_2(t)$ the bound G-CSF concentration (ng/mL). Here, and throughout, the superscript $^{h}$ denotes the homeostasis value of a quantity. The system~\eqref{eq:DDEmodel} is subject to the initial conditions (ICs) and history functions
\begin{equation}
\label{FullModel}
\begin{aligned}
Q(s) = & \phi_1(s) \quad \text{for } s \in [-\tauQ,0] \\
\NR(0) = & N_{R,0}  \\
N(0) = & N_0 \\
G_1(s) = & \phi_2(s) \quad \text{for } s \in [-\tau,0] \\
G_2(0) = & G_{2,0}, \\
\end{aligned}
\end{equation}
where $\phi_{1,2}(t) \in \mathcal{C}^0$ and
 \begin{equation}
 \tau = \sup \limits_{t \geq 0} \tauN(t).
 \end{equation}
$I_G(t)$ models the administration of exogenous G-CSF. As described in \cite{Craig2016c}, the self-renewal $\beta(Q)$ and amplification factor $A_Q(t)$ of the HSCs are given by
$$
\beta(Q) = f_Q\frac{\theta_2^{\sQ}}{\theta_2^{\sQ} + Q^{\sQ}},\qquad\quad
A_Q(t) =  A_Q^{h} = 2e^{-\gamma_Q\tauQ},
$$ 
and the rate at which HSCs differentiate into neutrophil precursors is determined by the circulating concentration of G-CSF
\begin{equation*} \label{eq:NewKappa}
\kappaN(G_1) = \kappaNhomeo + (\kappaNhomeo-\kappaN^\textit{min})\left[\frac{G_1^{\sk}-(\Gonehomeo)^{\sk}}{G_1^{\sk}+(\Gonehomeo)^{\sk}}\right].
\end{equation*}
The rate at which the neutrophil progenitors proliferate is given by
\be \label{eq:etaNP}
\etaNP(G_1(t)) =  \etaNPhomeo + (\etaNPhomeo-\etaNP^\textit{min})\frac{\bNP}{\Gonehomeo}
\left(\frac{G_1(t)-\Gonehomeo}{G_1(t)+\bNP}\right),
\ee
where $\tauNP$ days is the time it takes for proliferation. After exiting proliferation, cells mature with rate
\begin{equation*}
\VN(G_1(t)) = 1+(V_{max}-1)\frac{G_1(t)-\Gonehomeo}{G_1(t)-\Gonehomeo+b_V},
\end{equation*}
where the maximal age of maturing neutrophils is $\aNM$. Given $\VN(G_1(t))$ depends on the circulating concentration of G-CSF, the time it takes neutrophils to mature satisfies
\begin{equation} \label{tauNMthres}
\int_{t-\tauNM(t)}^{t}\VN(G_1(s))ds   =   \aNM,
\end{equation}
and the total time for the process of granulopoiesis is then the sum of the time to completion of each process, given by
\begin{equation*}
\label{eq:tauN}
\tauN (t)=\tauNP+\tauNM(t).
\end{equation*}
Maturing neutrophils are assumed to be subject to a constant death rate $\gammaNM$, and their amplification factor $A_N(t)$ is given by the integral equation
\begin{equation} \label{eq:AN}
A_N(t)
= \exp \left[\int_{t-\tauN(t)}^{t-\tauNM(t)} \etaNP(G_1(s)) d s-\gammaNM\tauNM(t)\right].
\end{equation}

The fraction of G-CSF bound to neutrophil receptors given by
\begin{equation*}
G_{BF}(t)=\frac{G_2(t)}{V\Ntott}\in[0,1], \qquad G_{BF}^{h}=\frac{\Gtwohomeo}{V[\NRhomeo+\Nhomeo]}.
\end{equation*}
regulates the rate with which cells exit the marrow reservoir as
\begin{equation*}
\label{eq:nu}
\ftrans(G_{BF}(t))=\ftranshomeo+(\ftrans^\textit{max}-\ftranshomeo)\frac{G_{BF}(t)-G_{BF}^{h}}{G_{BF}(t)-G_{BF}^{h}+b_G}.
\end{equation*}
Mature neutrophils die from the marrow reservoir with rate $\gammaNR$. Cells that transit into circulation are removed with constant rate $\gamma_N$.

Proofs of the existence, uniqueness, positivity, and boundedness of solutions to \eqref{eq:DDEmodel} are provided in Appendices~\ref{sec:PositivitySolutionsDDE}
and~\ref{sec:ExistUniqueDDE}.

\section{Stability Analysis}
\label{sec:stability}

We now perform stability analyses of the different models derived in
the last section to determine what parameter values, if any,
will render an equilibrium point unstable, most frequently by having
sustained oscillations about it. This is done using a well-established
technique, namely linearising about this equilibrium and then
calculating at which parameter values the ensuing characteristic
equation will have roots with positive real parts. This same technique
is traditionally applied to systems of ODEs in which case the
characteristic equation is a polynomial.

In general, the characteristic equation associated with an arbitrary
distribution is transcendental and possesses an infinite number of
roots. As we shall see, the advantage of an integer-order gamma
distribution is to yield a characteristic equation which is also a
polynomial, reflecting the fact, mentioned in Section (2.2.1)  that the
gamma distribution yields a system of ODEs. In the context of comparing
and contrasting the different models, we obtain a ``continuity" result
of sorts as we determine that approximation in distribution does lead to
approximation in stability diagrams (see \cite{Belair2015} for a similar
continuity argument).

\subsection{Characteristic equations for the Quartino endogenous G-CSF Models}

Consider first the generalised Quartino model~\eqref{eq:QuartinoModelLimited}.
Let
\be \label{eq:XVector}
\mathbf{X}(t):=(P(t),T_1(t),\dots,T_n(t),N(t),G(t))^{\intercal}\in\mathds{R}^{n+3}
\ee
be the vector of solutions
so that \eqref{eq:QuartinoModelLimited} can be rewritten in vector form as
\begin{equation}
\label{eq:VectorQuartino}
\frac{d\mathbf{X}}{dt}=\mathbf{F}(\mathbf{X}),
\end{equation}
where $\mathbf{F}(\mathbf{X})$ represents the right hand side of the Quartino model \eqref{eq:QuartinoModelLimited}.
Let $\mathbf{X}^*$ be an equilibrium of the system (that is that $\mathbf{F}(\mathbf{X}^*)=0$), then
define $\mathbf{Z}=\mathbf{X}(t)-\mathbf{X}^*$ and let $\mathds{J}(\mathbf{X}^*)$ be the Jacobian of \eqref{eq:QuartinoModelLimited} evaluated at $\mathbf{X}^*$ ($\mathds{J}(\mathbf{X})=d\mathbf{F}/d\mathbf{X}$). Then linearising about $\mathbf{X}^*$
we obtain
\begin{equation} \label{linQuart}
\frac{d\mathbf{Z}}{dt}=\mathds{J}(\mathbf{X}^*)\mathbf{Z}, 
\end{equation}
where nonlinear terms of order $\mathcal{O}(\|\mathbf{Z}\|^2)$ are neglected.
Seeking a nontrivial exponential solution $\mathbf{Z}(t) = \mathbf{C} e^{\lambda t}$ of \eq{linQuart}
with $\mathbf{C}\in\mathds{R}^{n+3}$, a vector of constants, and $\lambda\in\mathbb{C}$, we obtain the characteristic equation
\begin{equation}\label{char.eq.quart.det}
\det(\lambda\mathds{I}-\mathds{J})=0,
\end{equation}
where $\mathds{I}\in\mathds{R}^{(n+3)\times(n+3)}$ is the identity matrix. Evaluating the determinant in \eq{char.eq.quart.det} leads to the characteristic equation, which is stated
as Eq.~\eq{char.eq.quart.un.X}
in Appendix~\ref{sec:CharEqns}.
This gives
a polynomial of degree $n+3$ in $\lambda$ for the Quartino model \eqref{eq:QuartinoModelLimited}, and a polynomial of degree $7$ if we set $n=4$, as in \cite{Quartino2014}.
A steady state is unstable if any of the roots of this polynomial have positive real part.

The characteristic polynomial for the time rescaled Quartino model \eq{lctsd2} is also a polynomial in $\lambda$ of degree $n+3$, and actually has a simpler form than the characteristic polynomial for Quartino model \eqref{eq:QuartinoModelLimited}. But to derive this characteristic polynomial it is convenient to first consider the characteristic functions of the discrete delay DDE
\eq{lctsd5} and the distributed DDE model \eq{lctsd3}.

Let $\mathbf{Y}(t):=(P(t),N(t),G(t))^{\intercal}$ denote the vector of solutions of the discrete delay DDE ~\eqref{lctsd5},
and $\mathbf{Y}_{\tau}:=\mathbf{Y}(t-\tau)$ be the vector of delayed solutions. Then we can rewrite the
the DDE \eqref{lctsd5} as
\begin{equation}\label{disc.vec}
\dfrac{d\mathbf{Y}}{\textrm{d}t} = \mathbf{F}(\mathbf{Y},\mathbf{Y}_{\tau}),
\end{equation}
in vector form. Similar to the ODE case, let $\mathbf{F}(\mathbf{Y}^*,\mathbf{Y}^*)=\mathbf{0}$ be a generic steady state. Define the variables $\mathbf{Z} \coloneqq \mathbf{Y}-\mathbf{Y}^{*}$ and $\mathbf{Z}_{\tau} \coloneqq \mathbf{Y}_{\tau}-\mathbf{Y}^{*}$ and denote the linearisation matrices of \eqref{disc.vec} computed at $(\mathbf{Y},\mathbf{Y}_{\tau})=(\mathbf{Y}^{*},\mathbf{Y}^{*})$ by $\mathds{A}$ and $\mathds{B}$. Linearising \eqref{disc.vec} about $\mathbf{Y}^{*}$ and using  the variables $\mathbf{Z}$ and $\mathbf{Z}_{\tau} $ yields
\begin{equation}\label{disc.jac}
\frac{\textrm{d}\mathbf{Z}}{\textrm{d}t} = \mathds{A}\mathbf{Z} +  \mathds{B}\mathbf{Z}_{\tau}.
\end{equation}
The linearisation matrices $\mathds{A}$ and $\mathds{B}$ from \eqref{disc.jac}
are calculated in Appendix~\ref{app.jac.disc}.
Seeking a nontrivial exponential solution $\mathbf{Z}(t) = \mathbf{C} e^{\lambda t}$ for equation \eqref{disc.jac}, with constant
$\mathbf{C}\in\mathds{R}^{3}$ and $\lambda\in\mathbb{C}$, we obtain the characteristic equation
\begin{equation}\label{char.eq.det}
\det(\lambda\mathds{I}-\mathds{A}-e^{-\lambda\tau}\mathds{B})=0,
\end{equation}
where $\mathds{I}$ is the identity matrix. Evaluating the determinant in \eqref{char.eq.det} gives the transcendental
characteristic equation
\begin{equation}\label{char.eq.disc}
\lambda^3+a_2\lambda^2+a_1\lambda+a_0=be^{-\lambda\tau},
\end{equation}
where the coefficients $a_2$, $a_1$, $a_0$ and $b$ are computed in Appendix~\ref{app.jac.disc}.

In general equation \eq{char.eq.disc} has infinitely many roots, corresponding to the infinite dimensional nature of DDEs.
The treatment of these equations is made tractable because although there can be infinitely many complex numbers $\lambda$ that satisfy \eq{char.eq.disc}, it is well known that for any real number $\sigma$ there can only be finitely many solutions $\lambda$ with $Re(\lambda)>\sigma$ (see for example Lemma 4.2 in \cite{Smith2011}). To determine stability we need to ascertain whether all the roots have $Re(\lambda)<0$.

Comparing the discrete delay DDE \eq{lctsd5} with the distributed delay DDE \eq{lctsd3}, we see that they differ in only one term. Thus the linearisation of the distributed delay DDE \eq{lctsd3} follows exactly the steps taken for the discrete delay DDE \eq{lctsd5}. Then, following MacDonald \cite{MacDonald89}, the characteristic equation for the distributed DDE \eq{lctsd3} corresponds to \eq{char.eq.disc} with the term $e^{-\lambda\tau}$ replaced by
\be \label{LapGam}
G(\lambda)= \int_0^\infty \hspace{-0.7em}e^{-\lambda u}g_a^n(u)du =\frac{a^n}{(a+\lambda)^n},
\ee
where $G(\lambda)$ is the Laplace transform of the gamma probability density function, and hence we obtain
\begin{equation} \label{char.eq.distr}
\lambda^3+a_2\lambda^2+a_1\lambda+a_0-\frac{a^nb}{(a+\lambda)^n}=0,
\end{equation}
where the coefficients $a_j$ and $b$ computed in Appendix~\ref{app.jac.disc}
are the same as those for \eq{char.eq.disc}.
Notice that if $b=0$ then \eq{char.eq.disc} and \eq{char.eq.distr} both reduce to the same cubic polynomial.

If $n$ is an integer, equation \eq{char.eq.distr} is the characteristic equation of both the distributed DDE \eq{lctsd3}
and the equivalent time-rescaled Quartino ODE model \eq{lctsd2}. In that case, for $b\ne0$, equation
\eq{char.eq.distr} can be written as
\begin{equation} \label{char.eq.quart}
(1+\lambda/a)^n(\lambda^3+a_2\lambda^2+a_1\lambda+a_0)-b=0,
\end{equation}
a polynomial of degree $n+3$, which can be used to determine the stability of the steady-states of these models.
But since the time rescaling \eq{lctsd1} is monotonic this will also determine the stability of the steady-states of the
Quartino model \eq{eq:QuartinoModelLimited}. Characteristic equations which reduce to polynomials, such as
\eq{char.eq.distr} with $n$ an integer, are said to be \emph{reducible} \cite{MacDonald89}.

The characteristic equation \eq{char.eq.distr} is also valid for the distributed DDE \eq{lctsd3} when $n>0$ is not an integer. For general irrational $n$, equation \eq{char.eq.distr} can have infinitely many roots, as is the case for the discrete DDE \eq{char.eq.disc}. But, if $n=p/q$ is rational, where $p$ and $q$ are co-prime integers then we observe rather odd
behaviour. For an integer $m$, suppose that $n$ is rational with $n=p/q\in(m,m+1)$, which implies that
$p\in(qm,qm+q)$. Solutions of \eq{char.eq.distr} with $b\ne0$ then satisfy
\be \label{char.eq.rat}
(1+\lambda/a)^{p}(\lambda^3+a_2\lambda^2+a_1\lambda+a_0)^q-b^q=0,
\ee
(though not all solutions of \eq{char.eq.rat} will necessarily solve \eq{char.eq.distr} if $q$ is even).
Since \eq{char.eq.rat} is a polynomial in $\lambda$ of degree $p+3q$, the discrete delay DDE \eq{lctsd3}
has at most $p+3q\in(q(m+3),q(m+4))$ characteristic values $\lambda$ which satisfy \eq{char.eq.distr} when $n=p/q$ is rational.

It is natural to think of the discrete DDE \eq{lctsd5} as the limit as $n\to\infty$ with $a=n/\tau$ of the
distributed DDE \eq{lctsd3}, and indeed with $a=n/\tau$ we have
\be \label{disttodisc}
(1+\lambda/a)^{n}=(1+\lambda\tau/n)^{n} \; \to \; e^{\lambda\tau} \; \text{as} \; n\to\infty,
\ee
so the characteristic equation \eq{char.eq.quart} for the distributed DDE approaches the
characteristic equation \eq{char.eq.disc} of the discrete DDE as $n\to\infty$.
However, if one considers $n$ varying across the real numbers this is not a smooth limit as $n$ transitions between the rationals and irrationals. Consequently, even though the distributed DDE model \eq{lctsd3}
is valid for general real $n$, most authors, even when considering the behaviour as $n$ is varied or as $n\to\infty$
mainly restrict attention to the case where $n$ is an integer \cite{Beretta16,Campbell2009,MacDonald89}.

\subsection{Stability analysis for the Quartino endogenous G-CSF model and related forms}
\label{subsec:Stability}

The generalised Quartino model~\eqref{eq:QuartinoModelLimited}, has two steady states.
Assuming that $k_{tr}=\kP$ as in \eq{eq:kinc} for the reasons already stated, and considering the
model in the form \eq{eq:VectorQuartino} with vector solution
$\mathbf{X}(t)=(P(t),T_1(t),\dots,T_n(t),N(t),G(t))\in\mathds{R}^n$
these are given by
\begin{align}
\label{eq:Equilibrium1}
\mathbf{X}^*_1&=[P,T_1,\dots,T_n,N,G]=\left[0,0,\dots,0,0, \frac{k_{in}}{k_e}\right],\\
\label{eq:Equilibrium2}
\mathbf{X}^*_2 
&=\left[\frac{k_{circ}N_0}{k_{tr}},\frac{k_{circ}N_0}{a},\dots,\frac{k_{circ}N_0}{a},N_0,G_0\right]\!,
\end{align}
where $N_0$ is given by \eqref{eq:kinb}.

The time-rescaled Quartino model \eq{lctsd2} has the same steady states $\mathbf{X}^*_1$ and $\mathbf{X}^*_2$, since a monotonic rescaling of time does not affect equilibria.

The discrete and distributed delay Quartino DDE models
\eqref{lctsd3} and \eqref{lctsd5} have related equilibria, but in fewer space dimensions, since these models do not include transit compartments. In the $\mathbf{Y}(t)=(P(t),N(t),G(t))\in\mathds{R}^3$ notation of \eq{disc.vec}
these are given by
\begin{align} \label{lctdsdeq1}
\mathbf{Y}^*_1&=\left[0,0, \frac{k_{in}}{k_e}\right],\\ \label{lctsdeq2}
\mathbf{Y}^*_2&=\left[\frac{k_{circ}N_0}{k_{tr}},N_0,G_0\right].\!
\end{align}
If $n$ is a positive integer the distributed DDE \eqref{lctsd3} is equivalent to the Quartino model \eq{eq:QuartinoModelLimited} and the steady states
$\mathbf{Y}^*_1$ and $\mathbf{Y}^*_2$ correspond exactly to $\mathbf{X}^*_1$ and $\mathbf{X}^*_2$ as defined in
\eq{eq:Equilibrium1} and \eq{eq:Equilibrium2} for the appropriate $n$, and
with the values of $T_j$ following from \eq{lctsd4}.
We have the following stability result for these equilibria.

\begin{theorem}\label{stab.prop.disc}
Provided the parameters satisfy the constraints \eq{eq:kinc} and \eq{eq:kinb}
\begin{enumerate}
\item
For the distributed delay DDE \eqref{lctsd3} and the discrete delay DDE \eqref{lctsd5} the steady state $\mathbf{Y}^*_1$ is locally asymptotically stable if $\gamma<\beta$ and unstable if $\gamma>\beta$, and the steady state $\mathbf{Y}^*_2$ is unstable if $\gamma<\beta$.
\item
For the Quartino model \eq{eq:QuartinoModelLimited} and the time rescaled Quartino model \eq{lctsd2}
the steady state $\mathbf{X}^*_1$ is locally asymptotically stable if $\gamma<\beta$ and unstable if $\gamma>\beta$, and the steady state $\mathbf{X}^*_2$ is unstable if $\gamma<\beta$.
\end{enumerate}
\end{theorem}

\begin{proof}
At $\mathbf{Y}^*_1$ from \eq{coef.disc1}
we have $b=0$, thus from \eq{char.eq.disc} and \eq{char.eq.distr}
the characteristic equation for both the discrete and distributed DDE models reduces to
\be \label{chareq.X1}
h_1(\lambda):=\lambda^3+a_2\lambda^2+a_1\lambda+a_0=0,
\ee
and the stability of $\mathbf{Y}^*_1$ is the same for both models. If $\gamma>\beta$ then from \eq{coef.disc1}
we have
$0>a_0=h_1(0)$, while $h_1(\lambda)\to+\infty$ as $\lambda\to+\infty$. Hence, by the intermediate value theorem
there exists $\lambda>0$ such that $h_1(\lambda)=0$, and thus the steady state is unstable.

If $\gamma<\beta$ from \eqref{coef.disc2}
we have $a_{2}>0$, $a_{0}>0$ and $a_{2}a_{1}>a_{0}$ and it follows from
the Routh-Hurwitz criteria \cite{MacDonald89} that $Re(\lambda)<0$ for all characteristic roots of \eq{chareq.X1},
and hence $\mathbf{Y}^*_1$ is stable.

For the steady state $\mathbf{Y}^*_2$ when $\gamma<\beta$
equation \eqref{coef.disc2}
yields $a_{0}=0$, and $b>0$. For the
discrete delay DDE \eq{lctsd5} the characteristic equation reduces to
$$h_2^\infty(\lambda):=\lambda^3+a_2\lambda^2+a_1\lambda-be^{-\lambda\tau}=0.$$
Then $h_2^\infty(0)=-b<0$, while $h_2^\infty(\lambda)\to+\infty$ as $\lambda\to+\infty$, and the intermediate value theorem again implies that the steady state is unstable. For the distributed DDE \eq{lctsd3} a similar proof shows instability
using the characteristic function \eq{char.eq.quart} becomes
$$h_2^n(\lambda)=(1+\lambda/a)^n(\lambda^3+a_2\lambda^2+a_1\lambda+a_0)-b,$$
which has a positive leading coefficient and is negative when $\lambda=0$, so again the intermediate value theorem shows that the steady state is unstable.

The steady states $\mathbf{X}^*_1$ and $\mathbf{X}^*_2$ of the Quartino model \eq{eq:QuartinoModelLimited} have the same stability as those of the time-rescaled Quartino model \eq{lctsd2}, as the monotonic time-rescaling does not change the stability, though it will change the values of the characteristic roots. But the time rescaled model \eq{lctsd2} has its characteristic roots given by the degree $(n+3)$ polynomial \eq{char.eq.quart} which has the same roots as the characteristic equation \eq{char.eq.distr} of the distributed delay DDE, and hence $\mathbf{X}^*_j$ and $\mathbf{Y}^*_j$ have the same stability.
\qed
\end{proof}

\begin{figure}[ht!]
\begin{center}
\includegraphics[width=0.5\textwidth]{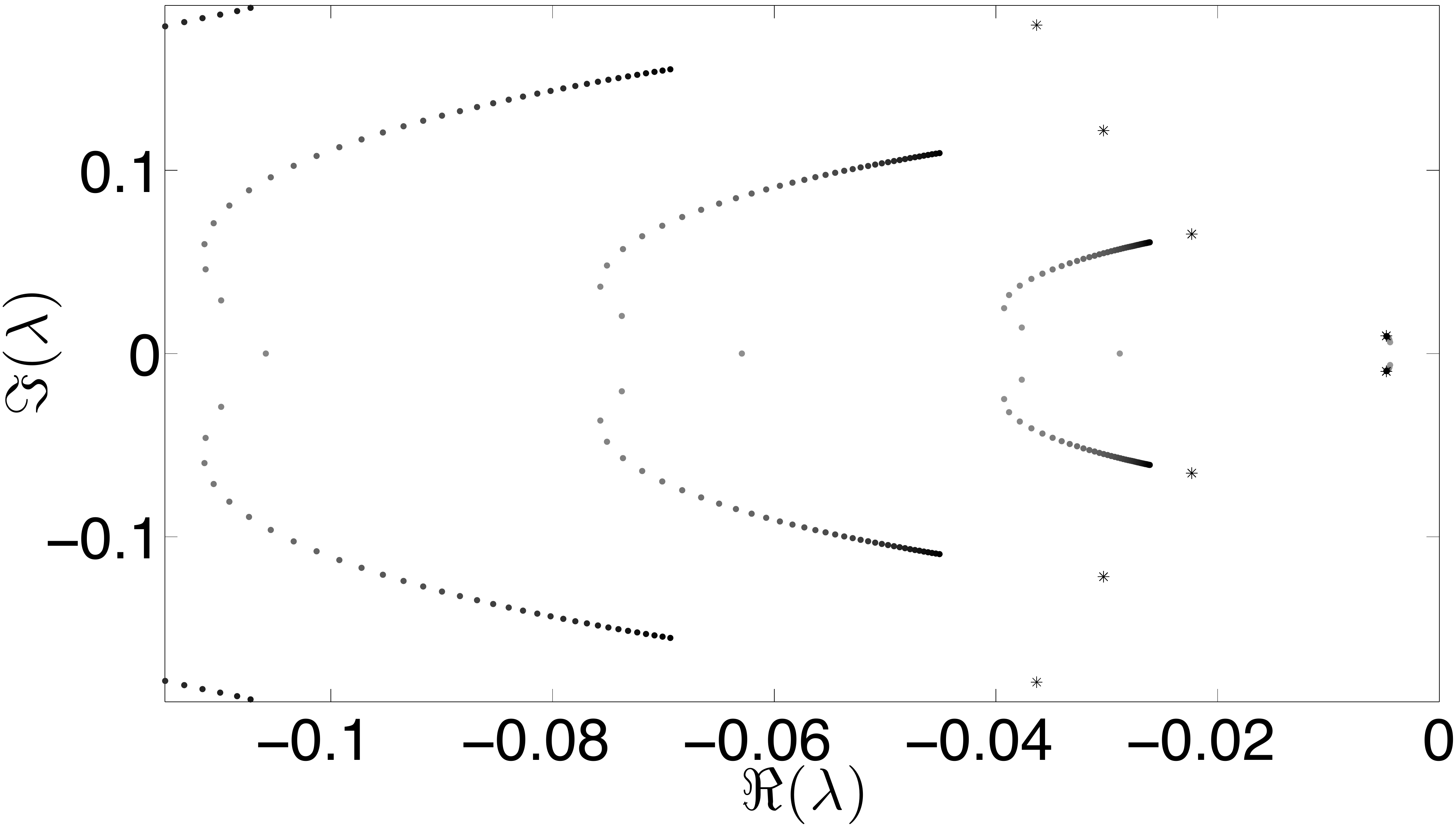}\includegraphics[width=0.5\textwidth]{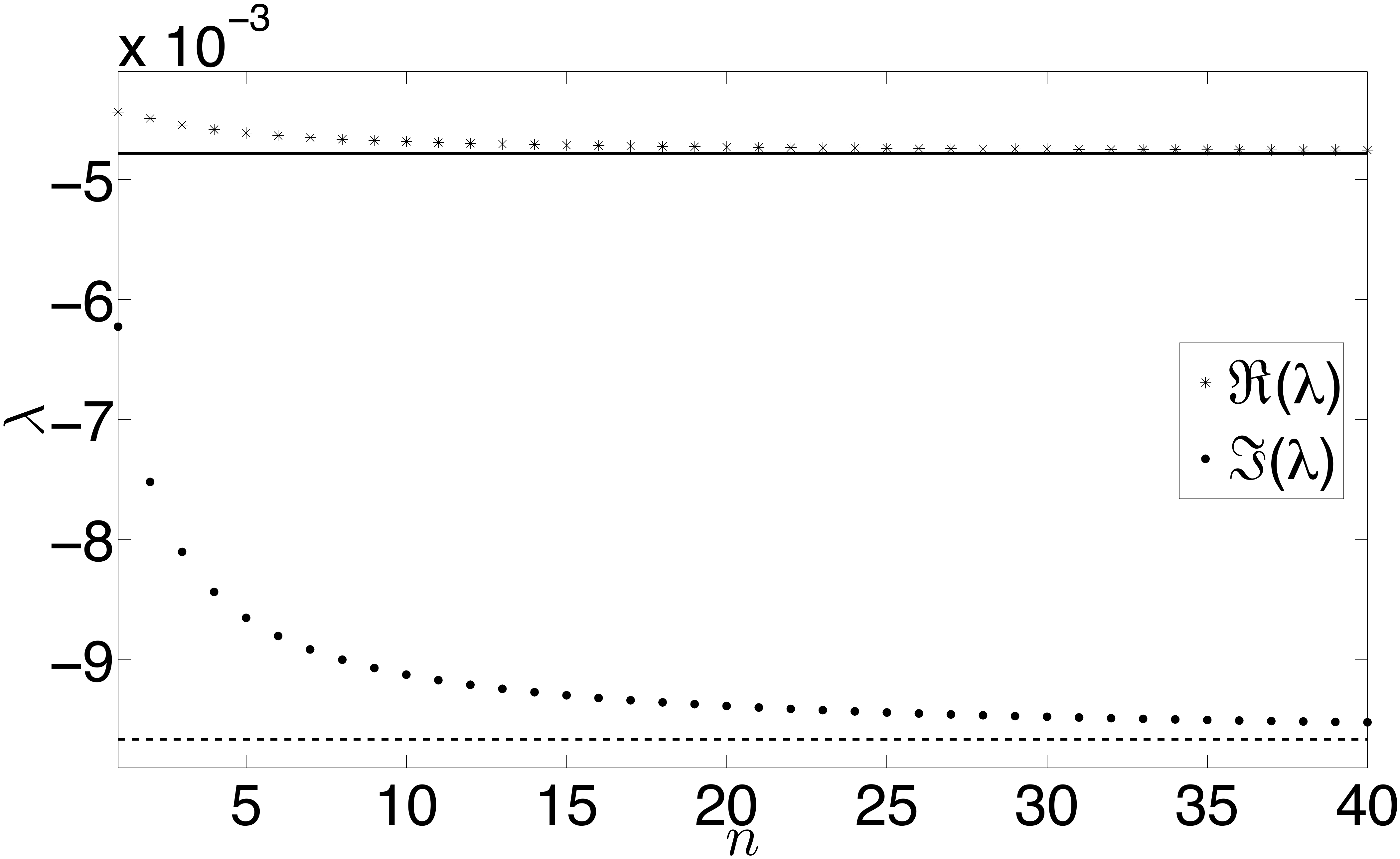}
\caption{At the steady state $\mathbf{X}_2^*$ of the time-rescaled
Quartino model \eq{lctsd2} and the corresponding steady state $\mathbf{Y}_2^*$ of the distributed and discrete delay models
\eq{lctsd3} and \eq{lctsd5} with all parameters from Table~\ref{tab:QuartinoValues}.
(Left): Asterisk for the discrete delay \eq{lctsd5} characteristic roots that satisfy \eq{char.eq.disc}, and in gray transitioning to black as $n$ is increased from $1$ to $40$ characteristic roots for the time-rescaled Quartino model \eq{lctsd2} and the distributed delay DDE \eq{lctsd3} which both satisfy \eq{char.eq.quart}.
(Right): Convergence of the real and negative imaginary parts of the rightmost characteristic root, which determines stability as $n$ increases.}
\label{fig:RootsDistributedDelayLimit1}
\end{center}
\end{figure}

For the standard parameters, as given in Table~\ref{tab:QuartinoValues}, we have $\gamma>\beta$ so
Theorem~\ref{stab.prop.disc} implies that the neutropenic steady states $\mathbf{X}^*_1$ and $\mathbf{Y}^*_1$ are unstable in all these models. Proving directly that the homeostatic steady states $\mathbf{X}^*_2$ and/or $\mathbf{Y}^*_2$ are stable when $\gamma>\beta$ is difficult, but we can compute the roots of the characteristic equations, and these are shown in Figure~\ref{fig:RootsDistributedDelayLimit1}. We see that the  homeostatic steady states are indeed stable when
$\gamma>\beta$. Moreover the characteristic roots for the time-rescaled transit compartment Quartino model converge to the
characteristic roots for the discrete delay DDE as $n$ increases. Although the steady state is stable for all the models, we see that it becomes less stable as $n$ increases, with the real part of the characteristic values tending to increase with $n$. The phenomenon of loss of stability for fixed delay $\tau$ as $n$ is increased has long been known, but remains an area of active interest \cite{MacDonald89,Campbell2009,Beretta16}.

In Section~\ref{sec:bifurcation} we will study how increasing $n$ can make the system more susceptible to loss of stability through bifurcations. However, here we point out that the change in stability observed in Theorem~\ref{stab.prop.disc}
depending on whether $\gamma>\beta$ or $\gamma<\beta$ is \emph{not} a bifurcation in the usual sense. When $\gamma=\beta$ the model is degenerate with the progenitor equation reducing to $\frac{\textrm{d}P}{\textrm{d}t}=0$.

In the proof of Theorem~\ref{stab.prop.disc} we made use of the relationships between the different model formulations to greatly simplify the derivation of the stability results. In particular the simpler forms of the characteristic equation for the time-rescaled Quartino \eq{lctsd2} and DDE models \eq{lctsd3} and \eq{lctsd5} makes these much easier to work with.
It is nevertheless possible to directly derive stability results for the Quartino model \eq{eq:QuartinoModelLimited}
at least for the steady state $\mathbf{X}^*_1$, though the proofs are much more involved. We include those results in Appendix~\ref{sec:QuartinoStability}
for completeness.

\subsection{Stability of the QSP model of granulopoiesis}
\label{sec:DDEBifurcation}

Here we perform the linear stability analysis of the steady states of the QSP model defined by the DDE system~\eqref{eq:DDEmodel},
without any exogenous G-CSF, i.e. $I_{G}(t)=0$. Similar to Section~\ref{subsec:Stability}, let $\mathbf{X}(t)\coloneqq (Q(t),\NR(t),N(t),G_{1}(t),G_{2}(t))^{\intercal}$ be the vector solution of~\eqref{eq:DDEmodel} and $\mathbf{X}_{\sigma}\coloneqq\mathbf{X}(t-\sigma)$ denote a vector of delayed solutions. Then the DDE system defining the QSP model \eqref{eq:DDEmodel}
can be rewritten in vector form as
\begin{equation}\label{qsp.dde.vec}
\dfrac{d\mathbf{X}}{\textrm{d}t} = \mathbf{F}(\mathbf{X},\mathbf{X}_{\tauQ},\mathbf{X}_{\tauN},\mathbf{X}_{\tauNM}).
\end{equation}

Parameters changes to the model lead to different steady states in equation~\eqref{qsp.dde.vec}.
Let the steady state computed at homeostasis be written as $\mathbf{X}^{h}\equiv (Q^{h},\NR^{h},N^{h},G_{1}^{h},G_{2}^{h})$, and denote a generic steady state by $\mathbf{X}^{*}\equiv (Q^{*},\NR^{*},N^{*},\!$ $G_{1}^{*},G_{2}^{*})$.

To linearize \eqref{qsp.dde.vec} around a steady state $\mathbf{X}^{*}$ we define the variables $\tau(t)=\tauN(t)-t$ and $u=s+\tau(t)$ to rewrite the amplification factor \eqref{eq:AN} as

\begin{equation}\label{A.Gu}
A_{N}(t) = \exp{\left[\int_{0}^{\tauNP}\etaNP(G_{1}(u-\tau(t)))du-\gammaNM\tauNM(t)\right]}.
\end{equation}
Thus we approximate the amplification factor \eqref{A.Gu} through the linearisation of the proliferation function \eqref{eq:etaNP} given by
\begin{equation}\label{approx.etaNP}
\etaNP(G_{1}) = \etaNP(G_{1}^{*})+\etaNP^{\prime}(G_{1}^{*})(G_{1}-G_{1}^{*}) + \cO(|G_{1}-G_{1}^{*}|^2),
\end{equation}
where $\etaNP^{\prime}\equiv d\etaNP/dG_{1}$. Further, since it does not affect the local stability of the steady state \cite{Cooke1996}, we freeze the state-dependent delay at its steady state value
\begin{equation}\label{approx.tauNM}
\tauNM(t) = \tauNM^{*}.
\end{equation}
Using \eqref{approx.etaNP} and \eqref{approx.tauNM} together with the distributed delay variable defined by
\begin{equation}\nonumber
\tilde{G}_{1}(t) \coloneqq \int_{0}^{\tauNP}\frac{G_{1}(u-\tau(t))}{\tauNP}du,
\end{equation}
equation~\eqref{A.Gu} becomes
\begin{equation}\label{approxAN}
\tilde{A}_{N}(t) = \exp{\left[\etaNP^{*}\tauNP-\gammaNM\tauNM^{*}
+\etaNP^{\prime}(G_{1}^{*})\tauNP(\tilde{G}_{1}(t)-G_{1}^{*})\right]}.
\end{equation}
As a consequence of the approximation in \eqref{approxAN}, we can rewrite \eqref{qsp.dde.vec} as
\begin{equation}\label{dde.vec2}
\dfrac{d\mathbf{X}}{\textrm{d}t} = \mathbf{f}(\mathbf{X},\mathbf{X}_{\tauQ},\mathbf{X}_{\tauN},\mathbf{X}_{\tauNM},\tilde{\mathbf{X}}),
\end{equation}
where
\begin{equation}\nonumber
\tilde{\mathbf{X}}(t) \coloneqq \int_{0}^{\tauNP}\frac{\mathbf{X}(u-\tau(t))}{\tauNP}du.
\end{equation}
Let $\mathbf{X}^*$ be a generic steady state of~\eq{dde.vec2}, defined by $\mathbf{f}(\mathbf{X}^{*},\mathbf{X}^{*},\mathbf{X}^{*},\mathbf{X}^{*},\mathbf{X}^{*})\!$ $=\mathbf{0}$. Define the variables $\mathbf{Z} \coloneqq \mathbf{X}-\mathbf{X}^{*}$, $\mathbf{Z}_{\sigma} \coloneqq \mathbf{X}_{\sigma}-\mathbf{X}^{*}$ and $\tilde{\mathbf{Z}} \coloneqq \tilde{\mathbf{X}}-\mathbf{X}^{*}$ and denote the linearisation matrices of \eqref{dde.vec2} with regards to $\mathbf{X}$, $\mathbf{X}_{\tauQ}$, $\mathbf{X}_{\tauN}$, $\mathbf{X}_{\tauNM}$, $\mathbf{X}_{\tau}$ and computed at $\mathbf{X}^{*}$, respectively by $\mathds{A}$, $\mathds{B}$, $\ldots$, $\mathds{E}$. Linearising \eqref{disc.vec} about $\mathbf{X}^{*}$ and using the variables $\mathbf{Z}$, $\mathbf{Z}_{\sigma}$ and $\tilde{\mathbf{Z}}$ yields
\begin{equation} \label{syst.JC}
\frac{\textrm{d}\mathbf{Z}}{\textrm{d}t} = \mathds{A}\mathbf{Z}+\mathds{B}\mathbf{Z}_{\tauQ}+\mathds{C}\mathbf{Z}_{\tauN}+\mathds{D}\mathbf{Z}_{\tauNM}+\mathds{E}\tilde{\mathbf{Z}}. 
\end{equation}
The linearisation matrices $\mathds{A}$, $\mathds{B}$, $\ldots$, $\mathds{E}$ from \eqref{syst.JC} are
computed in Appendix~\ref{app.jacDDE}.
Seeking a nontrivial exponential solution $\mathbf{Z}(t) = \mathbf{C} e^{\lambda t}$ for equation \eqref{syst.JC}, with constant $\mathbf{C}\in\mathds{R}^{5}$ and $\lambda\in\mathbb{C}$, we obtain the characteristic equation
\begin{equation}\label{char.eq.det.qsp}
\det (\lambda\mathds{I}-\mathds{A}-e^{-\lambda\tauQ}\mathds{B}-e^{-\lambda\tauN^{*}}\mathds{C}-e^{-\lambda\tauNM^{*}}\mathds{D}-f(\lambda)\mathds{E})=0,
\end{equation}
where $\mathds{I}$ is the identity matrix and
\begin{equation*}
f(\lambda) \coloneqq \frac{(e^{\lambda\tauNP}-1)}{\lambda\tauNP e^{\lambda\tauN^{*}}} = \frac{e^{-\lambda\tauNM^{*}}-e^{-\lambda\tauN^{*}}}{\lambda\tauNP}.
\end{equation*}
For $\lambda=re^{i\theta}$, with $r,\theta\in\mathds{R}$, we have $\lim_{r\rightarrow 0} f(re^{i\theta}) = 1$. We rearrange equation~\eq{char.eq.det.qsp} in the form $\det (\mathds{F})=0$. To calculate the matrix $\mathds{F}$, with terms $F_{ij}$ for $i,j=\{1,2,\ldots,5\}$, note that some terms of the linearisation matrices are symmetric while others are antisymmetric, namely $A_{43}=A_{42}$, $A_{52}=-A_{42}$, $A_{53}=-A_{42}$, $A_{35}=-A_{25}$, and $D_{24}=-A_{24}$ (see Appendix~\ref{app.jacDDE}).
Using this fact, we can then compute the terms
\begin{equation} \nonumber
F_{11}(\lambda) = A_{11}-\lambda + B_{11}e^{-\lambda\tauQ},\ F_{ii}(\lambda)=A_{ii}-\lambda\  \mbox{for}\ i=2,3,4,5,
\end{equation}
\begin{equation} \nonumber
F_{21}(\lambda) = C_{21}e^{-\lambda\tauN^{*}},\quad  F_{24}(\lambda) = A_{24}(1 - e^{-\lambda\tauNM^{*}}) + C_{24}e^{-\lambda\tauN^{*}}  + E_{24}f(\lambda),
\end{equation}
and write the matrix $\mathds{F}$ as
 \begin{equation*}
	\mathds{F} =
	\begin{bmatrix}
  	F_{11}(\lambda) & 0 & 0 & A_{12} & 0 \\
	F_{21}(\lambda) & F_{22}(\lambda) & A_{23} & F_{24}(\lambda) & A_{25}\\
	0 & A_{32} & F_{33}(\lambda) & 0 & -A_{25} \\
	0 & A_{42} & A_{42} & F_{44}(\lambda) & A_{45} \\
	0 & -A_{42} & -A_{42} & A_{54} & F_{55}(\lambda) \\
	\end{bmatrix}.
\end{equation*}
Defining the constants
\begin{eqnarray*}
&& K_{1} = -A_{45}A_{54},\qquad  K_{2} =  -A_{42}A_{25},\qquad K_{3} = -A_{42}A_{25}A_{54},\\[2mm]
&& K_{4} = A_{42}A_{25},\qquad K_{5} = - A_{32}A_{23},\qquad K_{6} = (A_{23}-A_{32} )A_{25}A_{42},\\[2mm]
&& K_{7} = A_{42}A_{25}A_{54},\qquad  K_{8} = [A_{32}(A_{23}A_{45} - A_{42}A_{25}) + A_{42}A_{23}A_{25}]A_{54},\\[2mm]
&& K_{9} = -A_{42}A_{45},\qquad K_{10} = A_{32}A_{42},\qquad K_{11} = A_{32}A_{42}A_{45},
\end{eqnarray*}
and the functions
\begin{eqnarray}\nonumber
\rho(\lambda)&=& F_{22}(\lambda)[F_{33}(\lambda)F_{44}(\lambda)F_{55} (\lambda)+ K_{1}F_{33}(\lambda) + K_{2}F_{44}(\lambda) + K_{3}] + \notag\\[2mm]
&&F_{44}(\lambda)[K_{4}F_{33}(\lambda) + K_{5}F_{55}(\lambda) + K_{6}] + K_{7}F_{33}(\lambda) + K_{8} + \notag\\[2mm]
&&F_{24}(\lambda)[K_{9}F_{33}(\lambda) + K_{10}F_{55}(\lambda) - A_{42}F_{33}(\lambda)F_{55}(\lambda) + K_{11}],\notag
\end{eqnarray}
\begin{equation}\nonumber
\psi(\lambda) = F_{21}(\lambda)[A_{42}F_{33}(\lambda)F_{55}(\lambda)-K_{9}F_{33}-K_{10}F_{55}(\lambda)-K_{11}]A_{12},
\end{equation}
the characteristic equation $\det(\mathds{F})=0$ becomes
\begin{equation}\label{char.dde}
F_{11}(\lambda)\rho(\lambda)+\psi(\lambda) = 0.
\end{equation}

\begin{figure}[ht!]
\begin{center}
\includegraphics[width=0.7\textwidth]{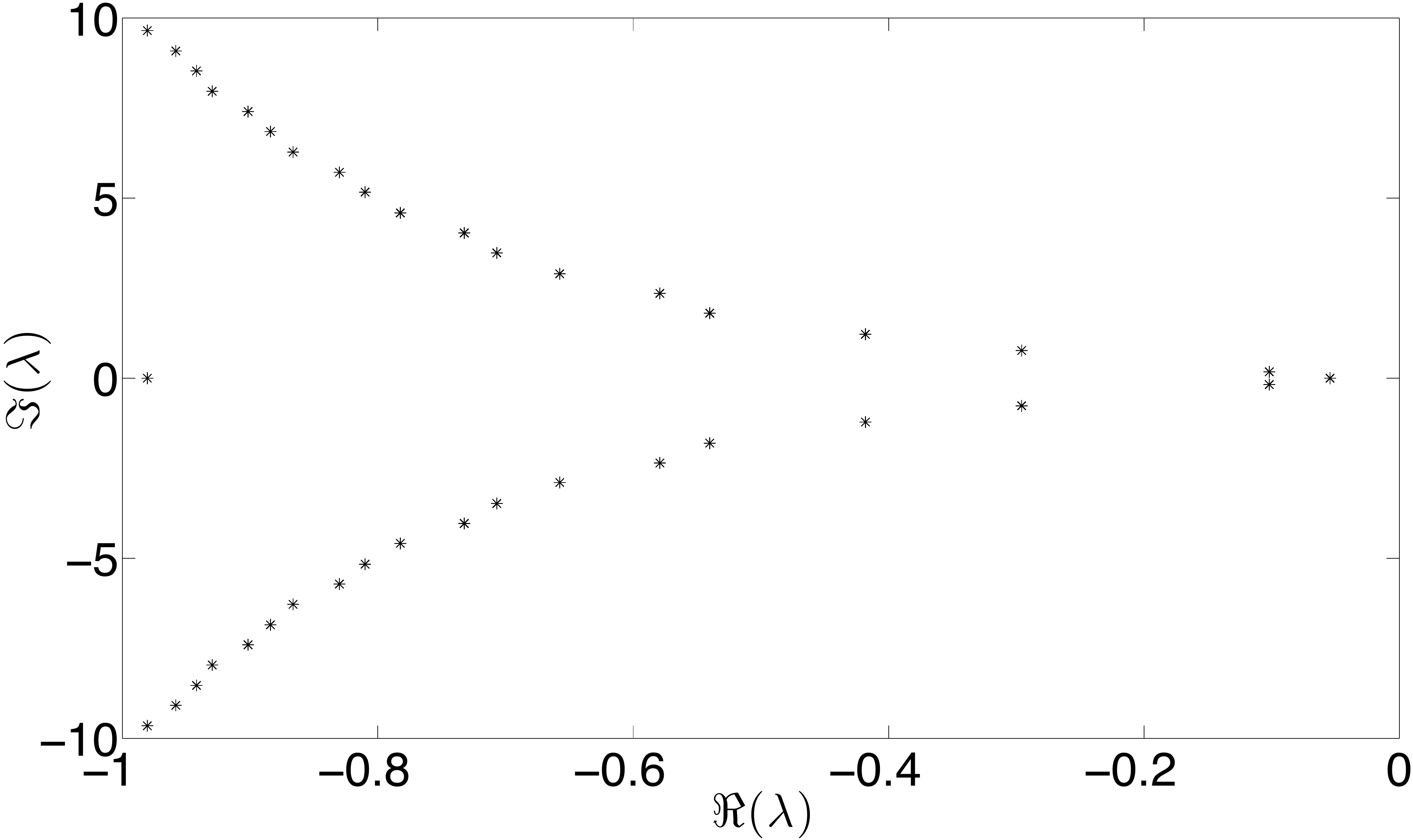}
\caption{Roots of the characteristic equation \eq{char.dde} of the QSP granulopoiesis model \eqref{eq:DDEmodel} evaluated at the homeostasis steady state $\mathbf{X}^{h}$ all have negative real part in the complex plane.}
\label{fig:RootsDDEGranulopHomeo}
\end{center}
\end{figure}

The solutions $\lambda\in\mathbb{C}$ to \eq{char.dde}, the characteristic roots, determine the stability of the steady state $\mathbf{X}^{*}$ from \eqref{qsp.dde.vec}/\eqref{syst.JC}.
To evaluate the stability numerically, we write $\lambda = \sigma + i\omega$, with $\sigma\in\mathds{R}$ and $\omega\in\mathds{R}$,
and then compute the roots of \eqref{char.dde} in the $(\sigma,\omega)$-plane using the Matlab subroutine \textit{fsolve} \cite{Mathworks}. As illustrated in Figure~\ref{fig:RootsDDEGranulopHomeo}, at
homeostasis all the characteristic roots $\lambda$ have negative real part, and so the homeostatic steady state $\mathbf{X}^{h}=(Q^{h},\NR^{h},N^{h},G_{1}^{h},G_{2}^{h})$, defined by
$\mathbf{F}(\mathbf{X}^{h},\mathbf{X}^{h},\mathbf{X}^{h},\mathbf{X}^{h})=\mathbf{0}$ in \eq{qsp.dde.vec}
is locally asymptotically stable.

\section{Bifurcation studies}
\label{sec:bifurcation}

Bifurcation analysis, or the study of the qualitative changes to the behaviour of a system given a change to parameter values, is a fundamental dynamical systems concept \cite{Meiss2007}. Accordingly, studying  bifurcation points can be a powerful tool in the life sciences to shed light on underlying parameter relationships and better understand the robustness of a system with regards to stability.

Historically, bifurcation analysis has been applied to study hematological pathologies and has provided valuable insight into the origins of disorders like cyclic neutropenia, a disease associated with dangerously low neutrophil counts and mouth blistering \cite{Dale2015} where a patient's ANCs oscillate with a period of around 21 days. These oscillations have been shown to correspond to a periodic orbit that appears through a loss of stability after the system undergoes a Hopf bifurcation \cite{Colijn:2,Foley2009b}. In the following sections, we perform bifurcation analyses on the equivalent forms of the Quartino endogenous G-CSF model and the QSP granulopoiesis model \eq{eq:DDEmodel}
to ascertain how changing parameter values modify the stability of each system, giving insight into the potential effects of PK variability on a physiological system and helping to understand pathophysiology of diseases.

\subsection{Bifurcation in the equivalent expressions of the Quartino endogenous G-CSF model}
We begin by investigating whether parameter changes in the equivalent expressions of the Quartino endogenous G-CSF model can lead their steady states $\mathbf{X}^*_2$ and $\mathbf{Y}^*_2$ given by \eq{eq:Equilibrium2} and \eq{lctsdeq2}, respectively, to lose stability. For this, we let $\lambda = \sigma + i\omega$, where $\sigma\in\mathds{R}$ and $\omega\in\mathds{R}$, and computed the roots of the characteristic equations \eq{char.eq.quart.un.X2},
\eq{char.eq.disc} and \eq{char.eq.quart} in the $(\sigma,\omega)$-plane using the Matlab subroutines \textit{roots} and \textit{fsolve} \cite{Mathworks}.

We saw in Section~\ref{subsec:Stability} that the homeostasis steady states $\mathbf{X}^*_2$  and $\mathbf{Y}^*_2$ are locally asymptotically stable in all the versions of the Quartino model that we consider, and that the models are degenerate when $\beta=\gamma$, consequently here we will study the bifurcations that occur as parameters are varied
from their homeostasis values with $\gamma>\beta$.

We begin by studying the Quartino model \eqref{eq:QuartinoModelLimited} starting from parameters used in \cite{Quartino2014}, so $n=4$, $a=k_{tr}$ and all the other parameters as in Table~\ref{tab:QuartinoValues}.
We observed that changes to $\gamma$, the parameter relating the feedback of circulating G-CSF concentrations on the proliferating pool, and $a$, the transit rate between maturation compartments, can lead to a loss of stability giving rise to a periodic orbit via a Hopf bifurcation. Table~\ref{tab.quartino} summarises the parameter pair values $(\gamma,a)$ necessary to induce such a loss in stability in $\mathbf{X}^*_2$, and the resulting period of the emerging periodic orbits.

If we let $\gamma$ be the bifurcation parameter and keep the remaining parameters at their homeostasis values (see Table~\ref{tab:QuartinoValues}), there is
a Hopf bifurcation point at $\gamma=0.86766$.
We verified numerically that $\mathbf{X}^*_2$ is locally asymptotically stable for $\beta<\gamma<0.86766$ and unstable if $\gamma>0.86766$. Of particular interest, as reflected in the bolded row of Table~\ref{tab.quartino}, we found a periodic orbit characteristic of cyclical neutropenia \cite{Dale2015}. Using the relation $a=n/$MMT with $n=4$ gives the value MMT$=123.33\,\text{hours}$, which is close to the mean maturation time of 133 hours for a patient under chemotherapy treatment reported by Quartino~\cite{Quartino2014}.

In Figure \ref{fig.cont.quart} (left) we show the Hopf bifurcation curve for $\mathbf{X}^*_2$ on parameter space $(a,\gamma)$ for the Quartino model  \eq{eq:QuartinoModelLimited}. The steady state is stable in the region below the Hopf curve and unstable otherwise. In Figure \ref{fig.cont.quart} (right) we see that increasing $\gamma$ for small values of the transit rate parameter $n$ lead to solutions with long period.
\begin{figure}[!htbp]
\begin{center}
\includegraphics[width=0.5\textwidth]{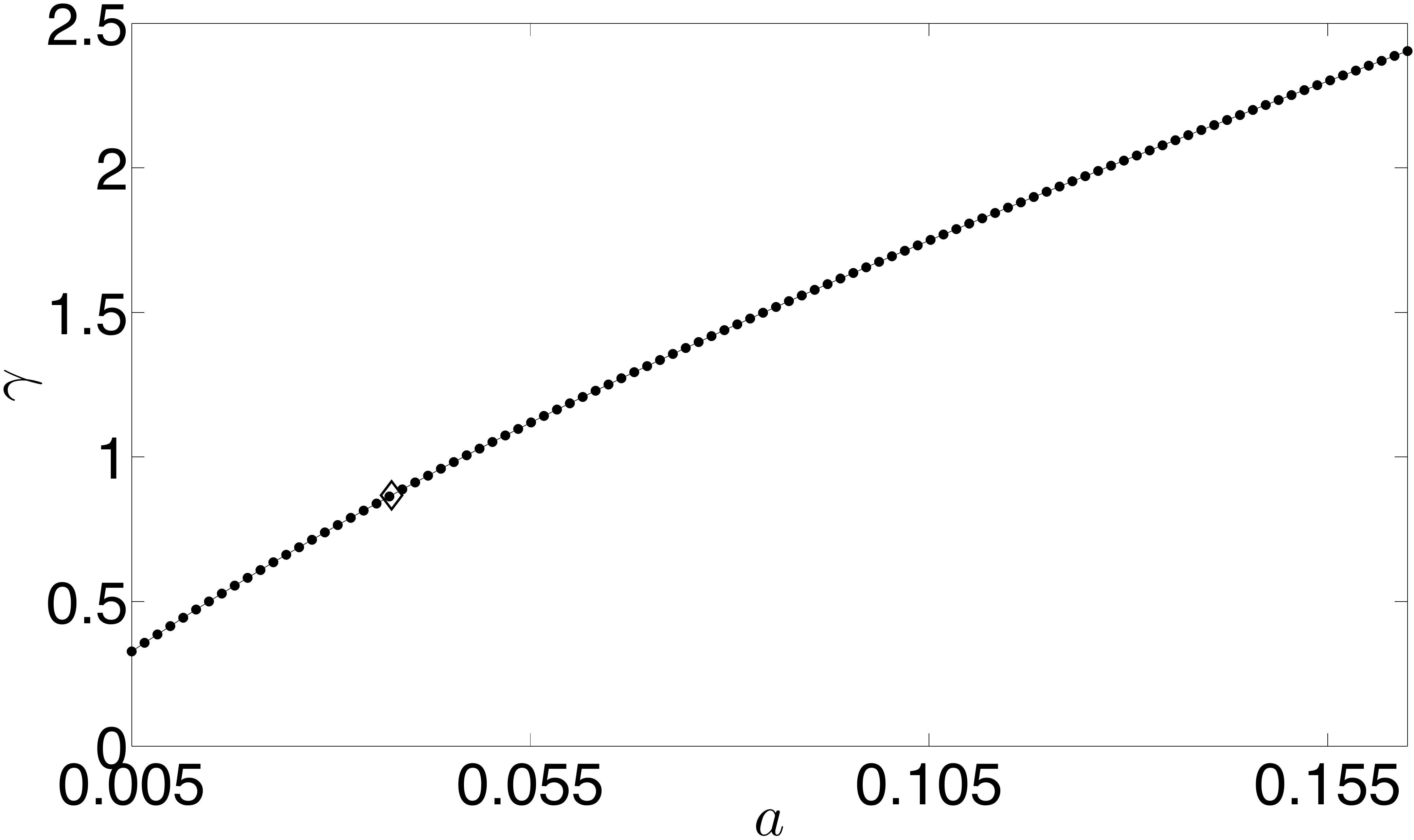}\includegraphics[width=0.5\textwidth,height=0.297\textwidth]{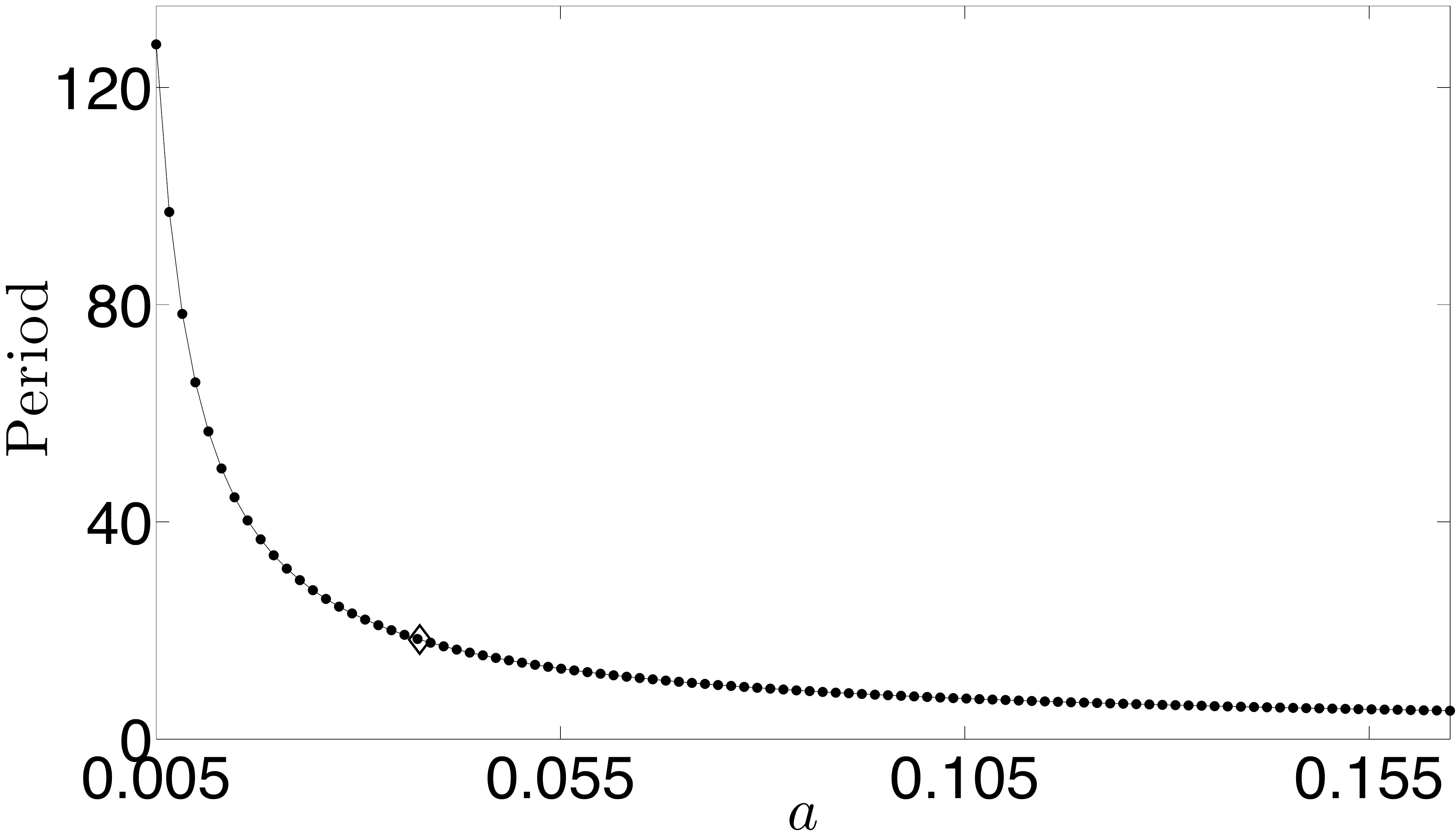}
\caption{Hopf bifurcation curve for $\mathbf{X}^*_2$ on parameter space $(a,\gamma)$ (Left) and the respective period (in days) as function of $a$ (Right) for the Quartino model  \eq{eq:QuartinoModelLimited}. The remaining parameters are as in Table~\ref{tab:QuartinoValues} and the diamond dot correspond to $a$ at its homeostasis value 0.03759 hours$^{-1}$.}
\label{fig.cont.quart}
\end{center}
\end{figure}
\begin{table}[!htbp]
\centering
\hspace*{-0mm}
\begin{tabular}{|c|c|c|c|}
\hline
$\gamma$ (-) & $a$ (hours$^{-1}$) & $\tau$ (hours) & Period (days)\\
\hline
0.87851 & 0.03831 & 104.4 & $18.00$ \\
0.86766 & 0.03759 & 106.41 & $18.32$ \\
$\mathbf{0.78911}$ & $\mathbf{0.03243}$  & $\mathbf{123.33}$ & $\mathbf{21.0}$ \\
\hline
\end{tabular}
\caption{Hopf bifurcation points for varying $\gamma$, $a$, and $\tau$ for $\mathbf{X}^*_2$ computed as in Figure~\ref{fig.cont.quart},
Values highlighted in bold correspond to periods characteristic of patients with cyclic neutropenia.}
\label{tab.quartino}
\end{table}

In the same vein, we also computed bifurcation points for the equilibrium $\mathbf{X}^*_2$ of distributed DDE model \eq{lctsd3} using the characteristic equation \eq{char.eq.quart} with $n=4$. As expected, given that the
this model is simply a time-rescaling of the Quartino model \eq{eq:QuartinoModelLimited} we obtain the same bifurcation
points shown in Table~\ref{tab.quartino}.
\begin{table}[!htbp]
\centering
\hspace*{-0mm}
\begin{tabular}{|c|c|c|}
\hline
$\gamma$ (-) & $n$ (-)  & Period (days)\\
\hline
4.61186 & 1  & 7.767 \\
1.69292 & 1.5  &13.53 \\
1.21571 & 2  & 15.95 \\
0.90252 & 3.5  & 18.08 \\
$\mathbf{0.86766}$ &$ \mathbf{4}$  &$\mathbf{18.32}$\\
0.75312 & 10  & 18.98 \\
0.72367 & 20  & 19.08 \\
0.71461 & 30  & 19.09 \\
0.71021 & 40  & 19.10 \\
0.69754 &``$\infty$''  &19.11\\
\hline
\end{tabular}
\caption{Hopf bifurcation points for varying $\gamma$ and $n$ for $\mathbf{Y}^*_2$ computed as in Table~\ref{tab.hopf2} but considering the characteristic equation of the distributed DDE model \eq{char.eq.quart}. The other parameters were fixed at their homeostasis values given in Table~\ref{tab:QuartinoValues}. The last row corresponds to the discrete delay case.}
\label{tab.distributed}
\end{table}

In Figure \ref{fig.cont.dist} (left) we show the Hopf bifurcation curve for $\mathbf{Y}^*_2$ on parameter space $(n,\gamma)$ for the distributed DDE model \eq{lctsd3}. The steady state is stable in the region below the Hopf curve and unstable otherwise. Increasing $n$ lead to $\gamma$ and the period of the Hopf bifurcation converge to the bifurcation point of the discrete DDE model \eq{lctsd5} shown in the bolded row of Table~\ref{tab.discrete}.
\begin{figure}[!htbp]
\begin{center}
\includegraphics[width=0.5\textwidth]{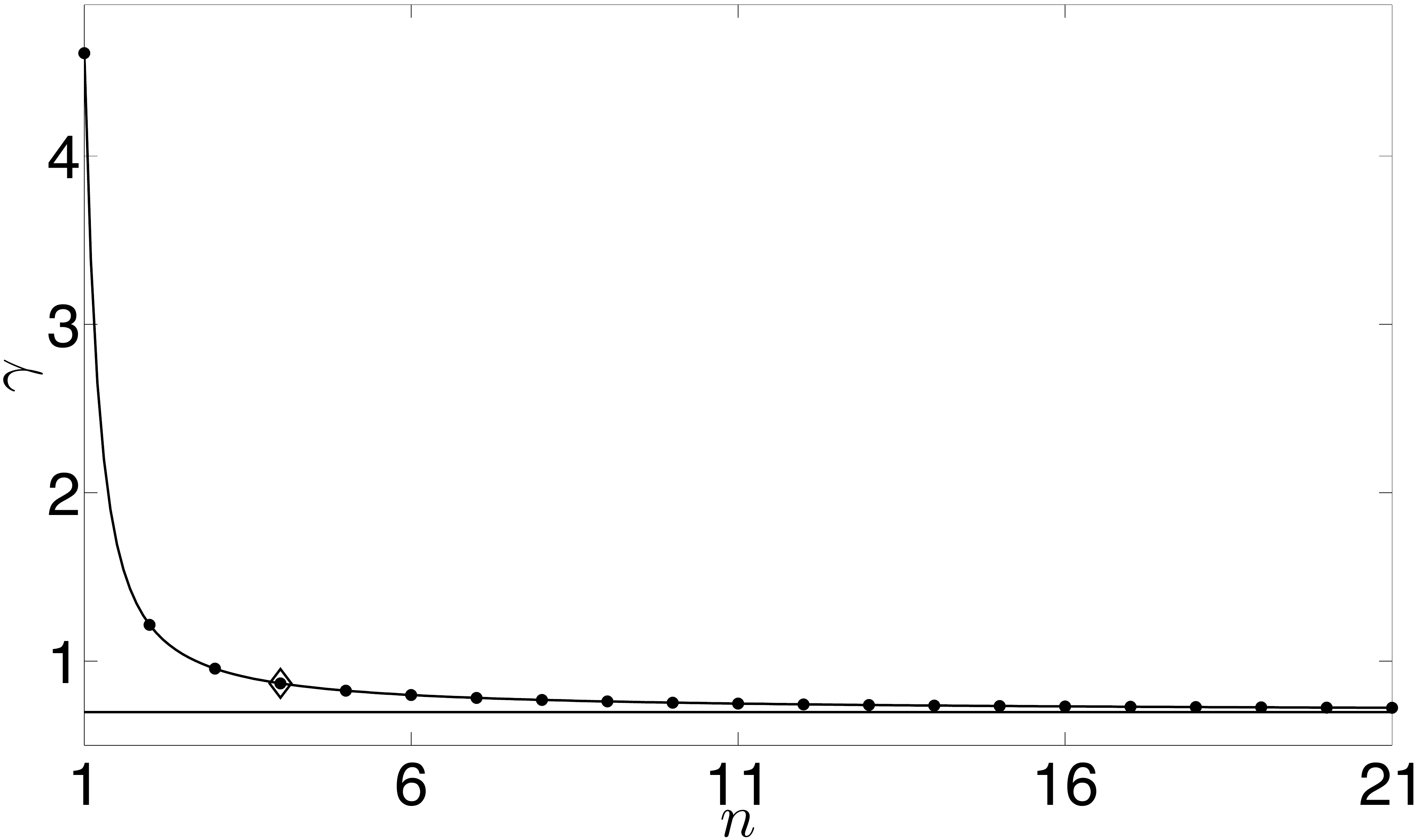}\includegraphics[width=0.5\textwidth,height=0.297\textwidth]{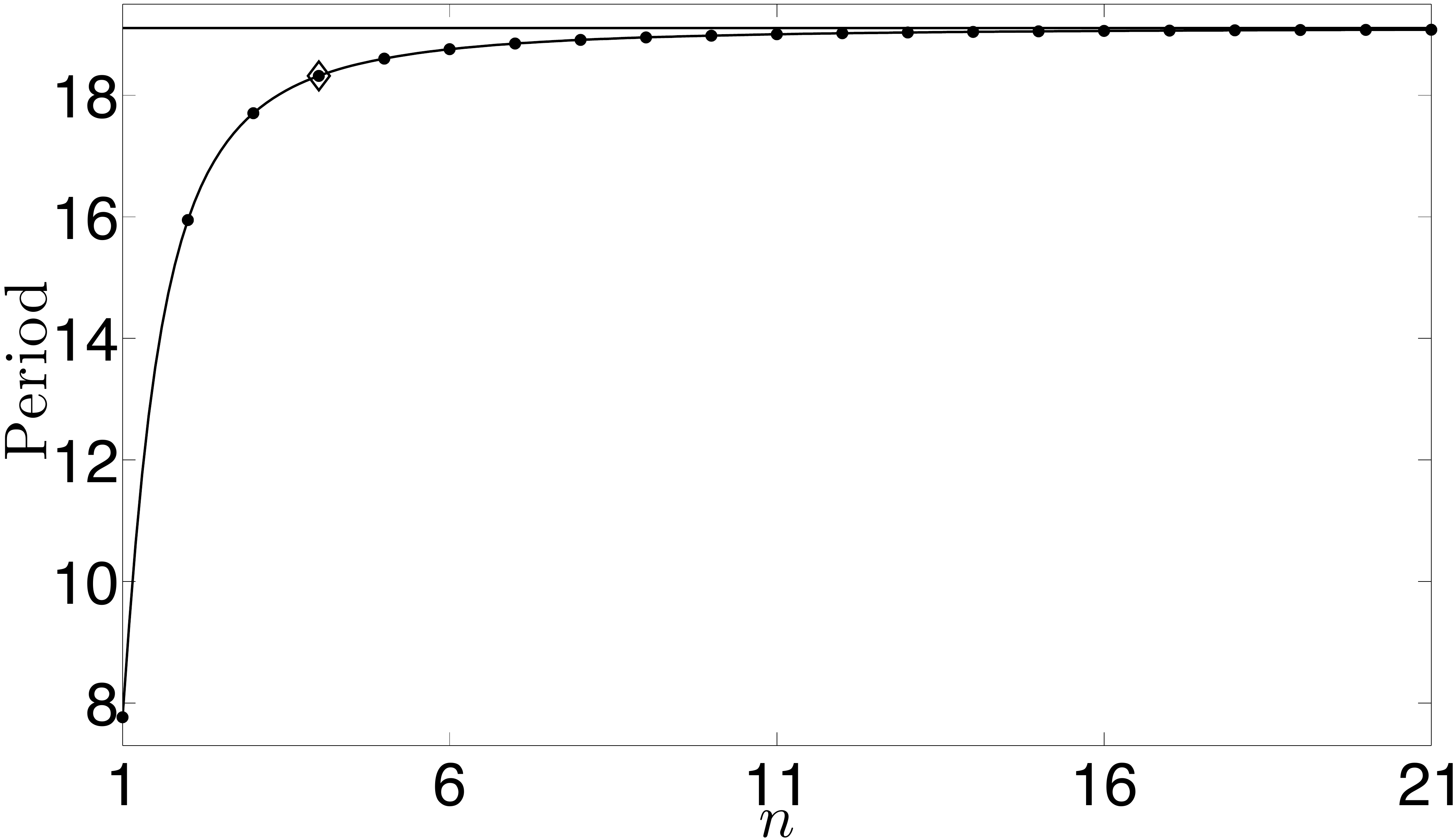}
\caption{Hopf bifurcation curve for $\mathbf{Y}^*_2$ on parameter space $(n,\gamma)$ (Left) and the respective period (in days) as function of $n$ (Right) for the distributed DDE model \eq{lctsd3}. The other parameters are as in Table~\ref{tab:QuartinoValues} and the diamond dot correspond to $n$ at its homeostasis value 4. The straight lines correspond to the bifurcation point $\gamma$ and its respective period for the discrete DDE model with other parameters at their homeostasis values.}
\label{fig.cont.dist}
\end{center}
\end{figure}

Table~\ref{tab.discrete} reports bifurcation points computed for the steady state $\mathbf{Y}^*_2$ of the discrete DDE model \eq{lctsd5} using the characteristic equation \eq{char.eq.disc}. Comparing the second rows of Tables \ref{tab.quartino} and \ref{tab.discrete}, we note that the region of stability for the discrete DDE model $\beta<\gamma<0.69754$ is smaller than that of the distributed DDE with $n=4$, and of the equivalent ODE Quartino model $\beta<\gamma<0.86766$. Furthermore, we verified that in the limit $n\to\infty$ with  $a=n/\tau$ or $a=(n+1)/\tau$ and holding $\tau$ fixed, the characteristic roots of the distributed DDE model \eq{char.eq.quart} converge to the roots of the discrete DDE model \eq{char.eq.disc} since \eq{char.eq.quart} approaches to \eq{char.eq.disc} when $n\to\infty$.
\begin{table}[!htbp]
\centering
\hspace*{-0mm}
\begin{tabular}{|c|c|c|}
\hline
$\gamma$ (-) & $\tau$ (hours)  & Period (days)\\
\hline
0.72631 & 99.8  & 18.00 \\
$\mathbf{0.69754}$ &$ \mathbf{106.41}$  &$\mathbf{19.11}$\\
0.65534 & 117.75  & 21.00 \\
\hline
\end{tabular}
\caption{Hopf bifurcation points for varying $\gamma$ and $\tau$ for $\mathbf{Y}^*_2$ computed as in Table~\ref{tab.hopf2} but considering the characteristic equation of the discrete DDE model \eq{char.eq.disc}. The other parameters were fixed at their homeostasis values given in Table~\ref{tab:QuartinoValues}.}
\label{tab.discrete}
\end{table}

In Figure \ref{fig.cont.disc} (left) we show the Hopf bifurcation curve for $\mathbf{Y}^*_2$ on parameter space $(\tau,\gamma)$ for the discrete DDE model \eq{lctsd5}. The steady state is stable in the region below the Hopf curve and unstable otherwise. In Figure \ref{fig.cont.disc} (right) we see that there is an approximate linear relation between the period of the limit cycles and the mean value of the distributed delay $\tau$ for along all values of $\gamma$.
\begin{figure}[!htbp]
\begin{center}
\includegraphics[width=0.5\textwidth]{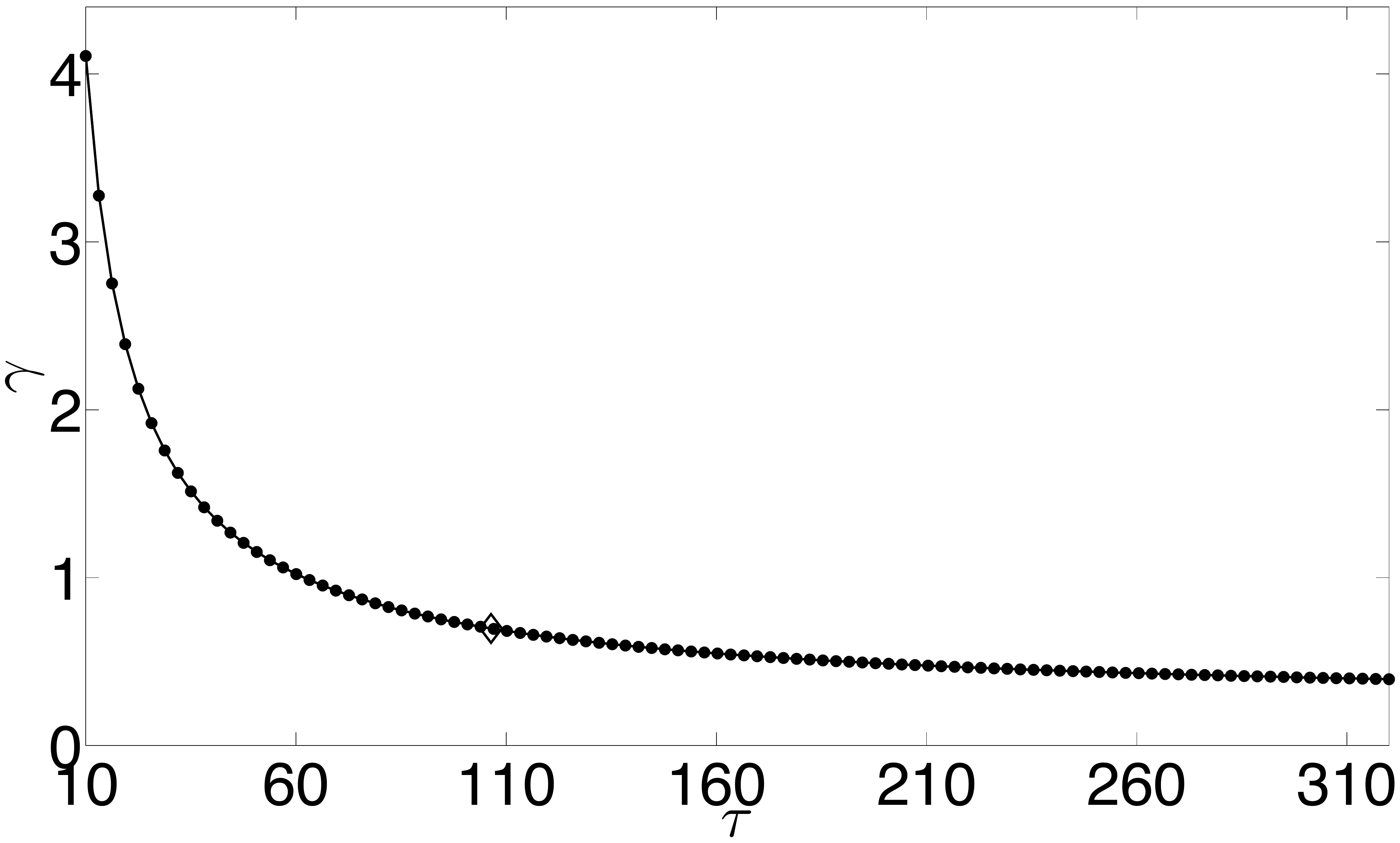}\includegraphics[width=0.5\textwidth,height=0.3\textwidth]{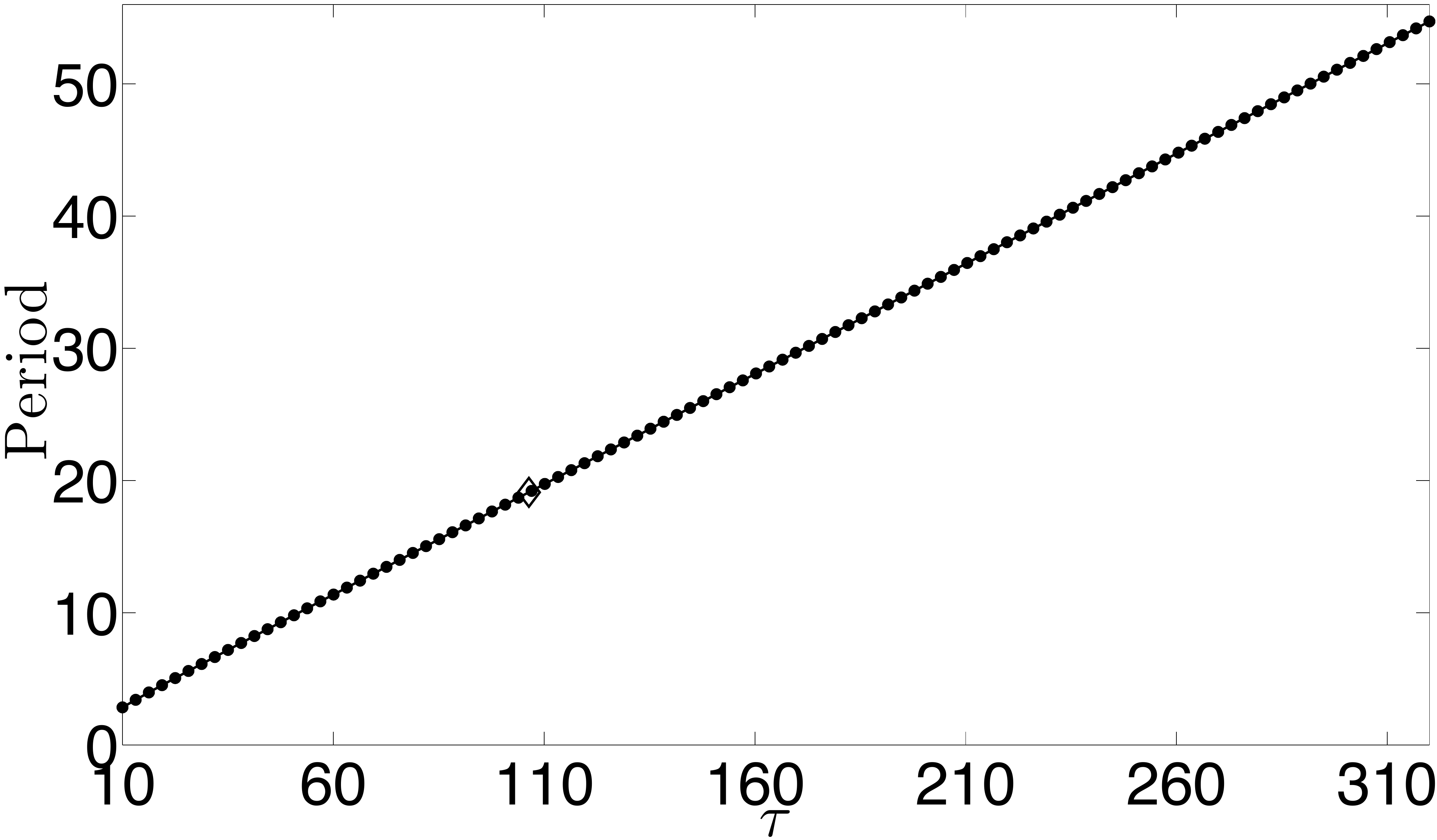}
\caption{Hopf bifurcation curve for $\mathbf{Y}^*_2$ on parameter space $(\tau,\gamma)$ (Left) and the respective period (in days) as function of $\tau$ (Right) for the discrete DDE model \eq{lctsd5}. Squares represent integer numbers $n$. The remaining parameters are as in Table~\ref{tab:QuartinoValues} and the diamond dot correspond to $\tau$ at its homeostasis value 106.41 hours.}
\label{fig.cont.disc}
\end{center}
\end{figure}

\subsection{Bifurcations in the QSP model of granulopoiesis}

We also studied whether parameter changes can lead the steady state $\mathbf{X}^{*}$ of system \eqref{eq:DDEmodel} to lose stability. We observed that changes in parameters related to proliferation and maturation lead to a loss of stability via a Hopf bifurcation, as reflected in Table~\ref{tab.hopf1}. Inversely, we further verified the stability of the steady state when the half-maximal neutrophil proliferation constant satisfies $\bNP \in [\bNP\times{10}^{-3},\bNP\times{10}^{3}]$, the rate of maturing neutrophil death satisfies $\gammaNM \in [\gammaNM\times{10}^{-2},\gammaNM\times{10}^{2}]$, and the neutrophil apoptosis rate in the bone marrow reservoir satisfies $\gammaNR \in [\gammaNR\times{10}^{-2},\gammaNR\times{10}^{2}]$.

We further investigated whether varying pairs of parameters in tandem could lead the steady state to lose stability via a Hopf bifurcation, as reflected in Table~\ref{tab.hopf2}. An additional Hopf bifurcation point leading to an orbit of period of $20.94$ days was observed by changing four parameters simultaneously: $\etaNPhomeo=1.7\,\text{days}^{-1}$; $\bNP=2.0$ ng/mL; $\etaNP^\textit{min}=1.3\,\text{days}^{-1}$; and $\tauNP=6.1\,\text{days}$.
\begin{table}[!htbp]
\centering
\hspace*{-0mm}
\begin{tabular}{|l|c|c|c|}
\hline
Parameter (units) & Homeostasis Value & Hopf Bifurcation  & Period (days)\\
\hline
$\gamma_{Q}$ (days$^{-1}$) & $0.1$ & $0.22791$ & $36.69$ \\
$\etaNPhomeo $ (days$^{-1}$)& $1.6647$ & $7.4$  & $60.31$\\
$\etaNP^{min}$ (days$^{-1}$) & $1.4060$ & $0.81 $  & $25.86$\\
$V_{max}$ (-)& $7.8669$ & $93$ & $5.20$\\
$b_{V}$ (ng/mL) & $0.24610$ & $0.0184$  & $5.20$\\
\hline
\end{tabular}
\caption{For each line the Hopf bifurcation point were computed changing the respective parameter and following the solution $\lambda$ of Eq. \eqref{char.dde} in the $(\sigma,\omega)$-plane numerically,  the period was estimated by $2\pi/\omega$ and the parameters values at homeostasis were obtained from \cite{Craig2016c}.}
\label{tab.hopf1}
\end{table}
\begin{table}[!htbp]
\centering
\hspace*{-0mm}
\begin{tabular}{|l|c|c|c|}
\hline
Parameters \;(units) & Hopf Bifurcation  & Period (days)\\
\hline
$(V_{max},b_{V})$\; (-,ng/mL)& $(30,7.8)$  & $5.20$ \\
$(V_{max},b_{V})$ \;(-,ng/mL)& $(61,16.1)$  & $5.20$ \\
$(\bNP,\etaNP^{min})$ \;(ng/mL,days$^{-1}$)& $(0.065,1.1)$  & $25.86$\\
$(\etaNPhomeo,\etaNP^\textit{min}) $ \;(days$^{-1}$,ng/mL) & $(2,1.1)$  & $32.58$\\
\hline
\end{tabular}
\caption{For each line the Hopf bifurcation point were computed changing the respective pair of parameters and following the solution $\lambda$ of Eq. \eqref{char.dde} in the $(\sigma,\omega)$-plane numerically,  the period was estimated by $2\pi/\omega$ and the parameters values at homeostasis were obtained from \cite{Craig2016c}.}
\label{tab.hopf2}
\end{table}

\section{The impact of stability and bifurcations on PK/PD considerations}
\label{sec:variability}

In the PK/PD context, sensitivity analysis is frequently applied to investigate the impact of parameters variability on the system's output. There the goal is to understand how predictions (outputs) change given changes to initial values (inputs). Still, when evaluating treatments, one may wonder how small changes to parameters affect the qualitative (e.g. existence and stability of equilibria, etc.) behaviour of the model. For example, if we have \textit{a priori} information about a PK parameter's variability and this parameter helps determine the model's stability, bifurcation analysis can help to assess whether small changes within the range of the measured variability can bring about serious unintended shifts in the physiological system. Put another way, how are parameters changing when the system shifts stability or becomes unstable? Bifurcation analysis is rarely used in conventional PK/PD analyses, however the study of qualitative model behaviour is becoming increasingly recognised as an important tool for drug development \cite{Bakshi2016,Ghosh2013}. The potential impact of variability in PK parameters on the dynamics of the governing equations, which correspond to the PD aspects, is multifaceted. Within pharmacometrics, various situations must therefore be considered when assessing which (and how) parameters are susceptible to generating bifurcations when their values change.

In the simplest case, for physiological or drug parameters not influenced by drug concentration, no bifurcation can be generated through any PK variability. Examples of such parameters could include the maximal achievable response in an Emax model or a zero-order endogenous production rate. In that vein, we observe that in the Quartino model, the bifurcation point for $\gamma$ is not likely to be reached by realistic variations in G-CSF concentrations.

More familiar to the pharmacometrician is the case where different parameters values associated with specific cohorts and/or patient subpopulations correspond to individual states in the dynamical system. Here bifurcation analysis is analogous to Population-PK (Pop-PK) covariate analysis techniques used to separate and determine Pop-PK models for each subgroup. From this covariate analysis, one infers that for each state, there is a particular PK/PD model for which the (between subject) variability has been explained. Accordingly, Pop-PK covariate studies examine the impact of between subject variability is readily accounted for during the process of building a Pop-PK model. In the same way, bifurcation analysis of mathematical models helps to ascertain how changes to (variability in) model parameters affect the stable state. Of note, in the most complex case, where the parameters of a dynamical system are affected by PKs (which can also vary during therapy), within subject variability in PK is an important factor in determining model stability and a case-by-case model analysis must be carried out.

The impact of interindividual variability (IIV) is also an open question when considering the time-rescaled Quartino \eqref{lctsd2} and discrete delay \eqref{lctsd5} models. Since the discrete delay model is the equivalent to taking $n$ to infinity in \eqref{lctsd2}, we investigated whether there would be altered behaviour due to IIV with increases in $n$. We began by generating 30 virtual patients using the docetaxel Pop-PK model of Bruno \cite{Bruno1996}
and set Edrug=Slope$C_{\text{doc}}$, as in \cite{Quartino2014}, where $C_{\text{doc}}$ is the concentration of docetaxel in the central compartment.
Using the same individual patient values for the generalised and time-rescaled Quartino models, we verified that solutions to \eqref{eq:QuartinoModelLimited} were identical as the solutions to \eqref{lctsd2} when the latter were rescaled according to \eqref{lctsd1} (not shown). We proceeded to compare the full Pop-PK predictions of \eqref{eq:QuartinoModelLimited} for these 30 virtual patients when $n=4$ and $n=10$. In Figure~\ref{fig:PopPK}, increases to $n$ are not significantly affected by the inclusion of IIV, though increases to $n$ decidedly impact on the distribution of solutions. Given these results and to better visualise the impact of increasing $n$, we then compared the predictions of \eqref{eq:QuartinoModelLimited} to \eqref{lctsd5} (the case of infinite $n$) using only the typical parameter estimates. As seen in Figure~\ref{fig:Increasing_n}, in the limit $n\to\infty$, solutions to the generalised Quartino model (and equivalently, the time-rescaled Quartino and distributed delay models) converge to that of the discrete delay.
Further, as $n$ increases, so too does the dimension of the resulting ODE system for \eqref{eq:QuartinoModelLimited}, which is not the case for the 3 equations of \eqref{lctsd5}, thereby encouraging the use of a delay model to speed up simulation time. It should be noted that in all simulations described in this section, due to the misspecification of MMT as MMT=$(n+1)/k_{tr}$, to compare directly with the results of \cite{Quartino2014} $k_{tr}=\kP$ were fixed to the value in Table~\ref{tab:QuartinoValues}, MMT was taken to be 133 hours, and $\tau$ was recalculate as $\tau=(4/5)\text{MMT}$, with $a=n/\tau$.

\begin{figure}[htbp!]
  \begin{subfigure}[b]{0.5\linewidth}
    \begin{center}
    \includegraphics[width=0.9\linewidth]{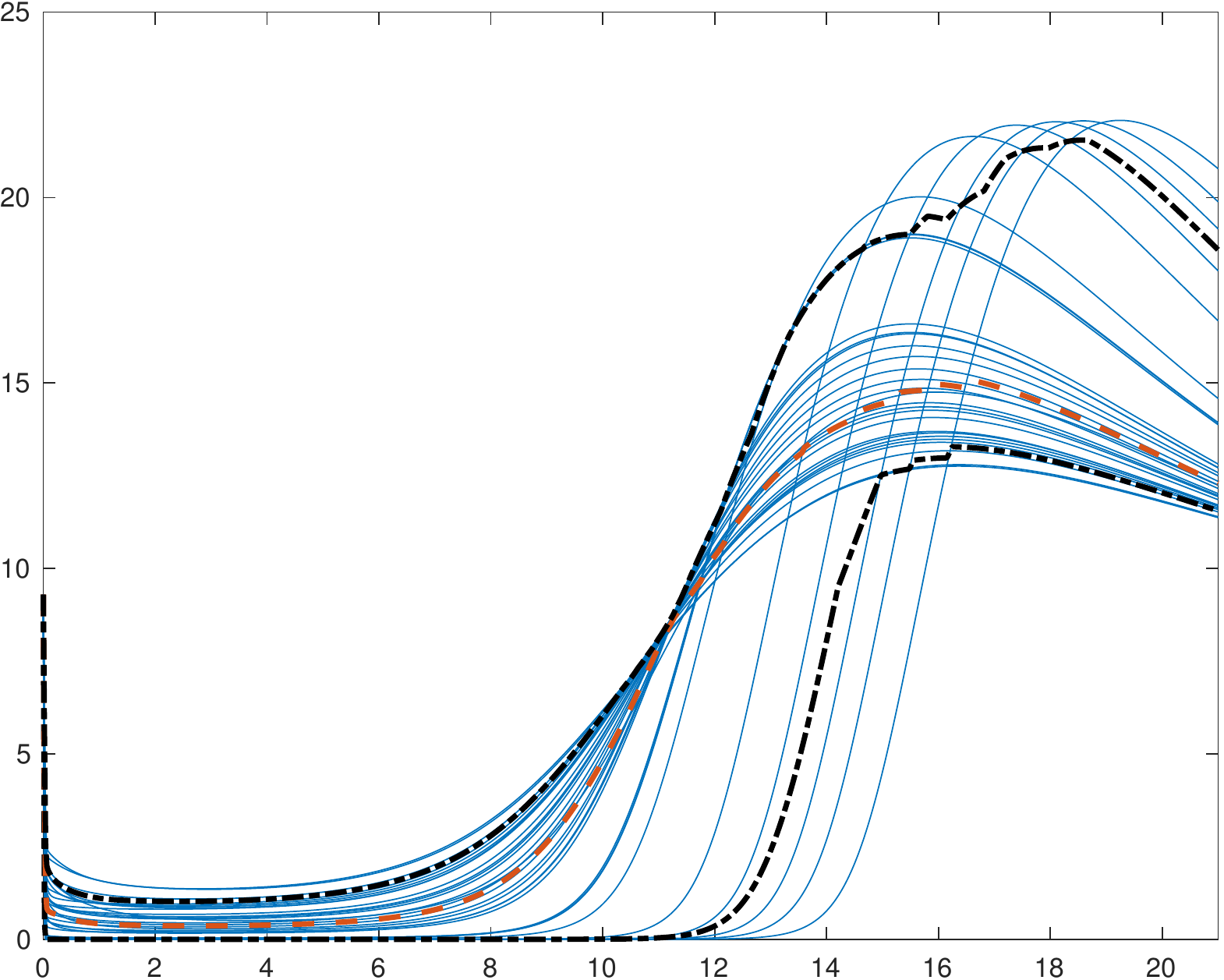}
      \put(-161,95){\rotatebox[origin=c]{90}{\scriptsize$P(t)$}}
     \put(-22,-6){\scriptsize Days}
     \end{center}
    \caption{$n=4$}
    \label{fig:Pn4}
    \vspace{4ex}
  \end{subfigure}
  \begin{subfigure}[b]{0.5\linewidth}
    \begin{center}
    \includegraphics[width=0.9\linewidth]{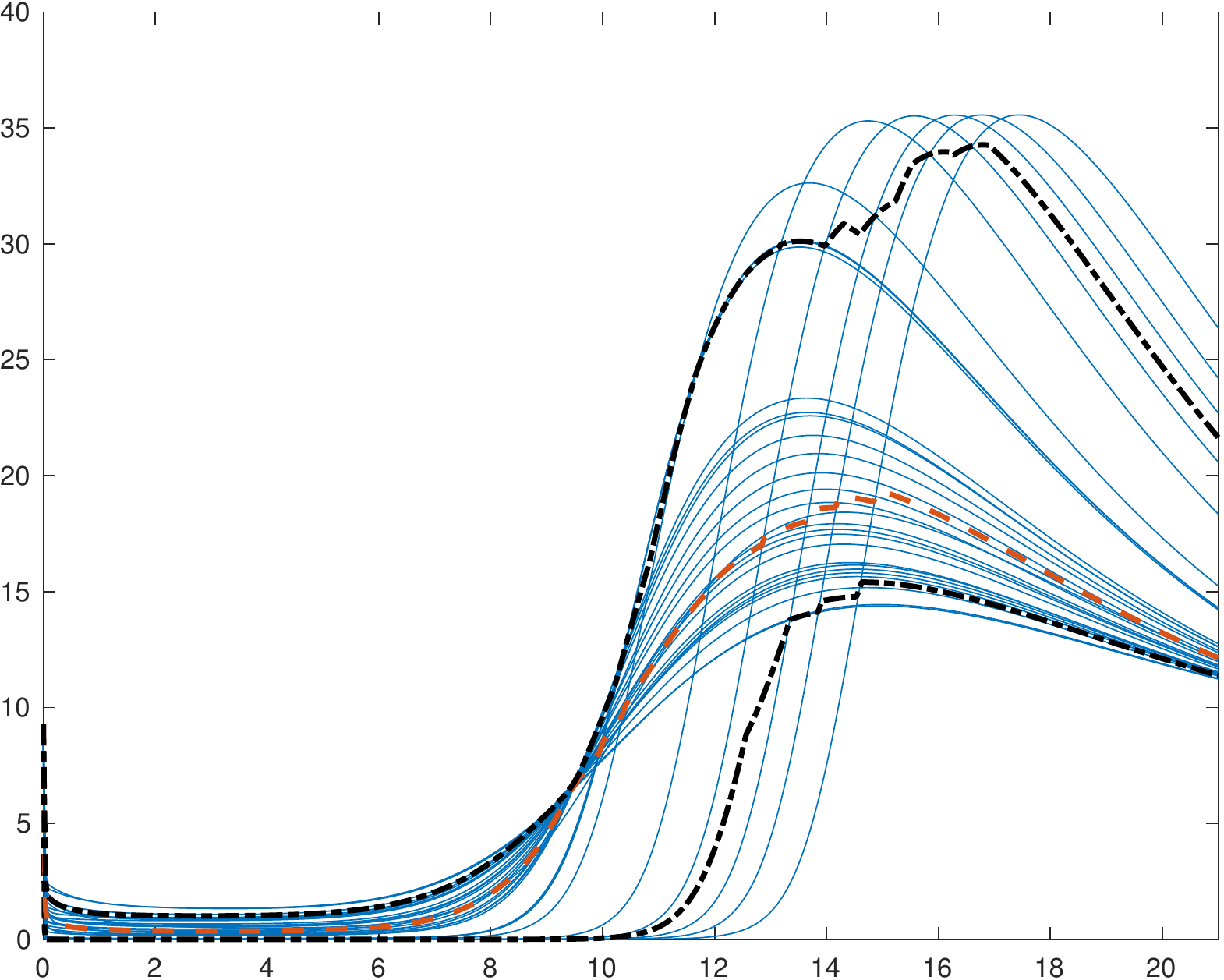}
         \put(-159,87){\rotatebox[origin=c]{90}{\scriptsize$P(t)$}}
     \put(-22,-6){\scriptsize Days}
     \end{center}
    \caption{$n=10$}
    \label{fig:Pn10}
    \vspace{4ex}
  \end{subfigure}
  \begin{subfigure}[b]{0.5\linewidth}
    \begin{center}
    \includegraphics[width=0.9\linewidth]{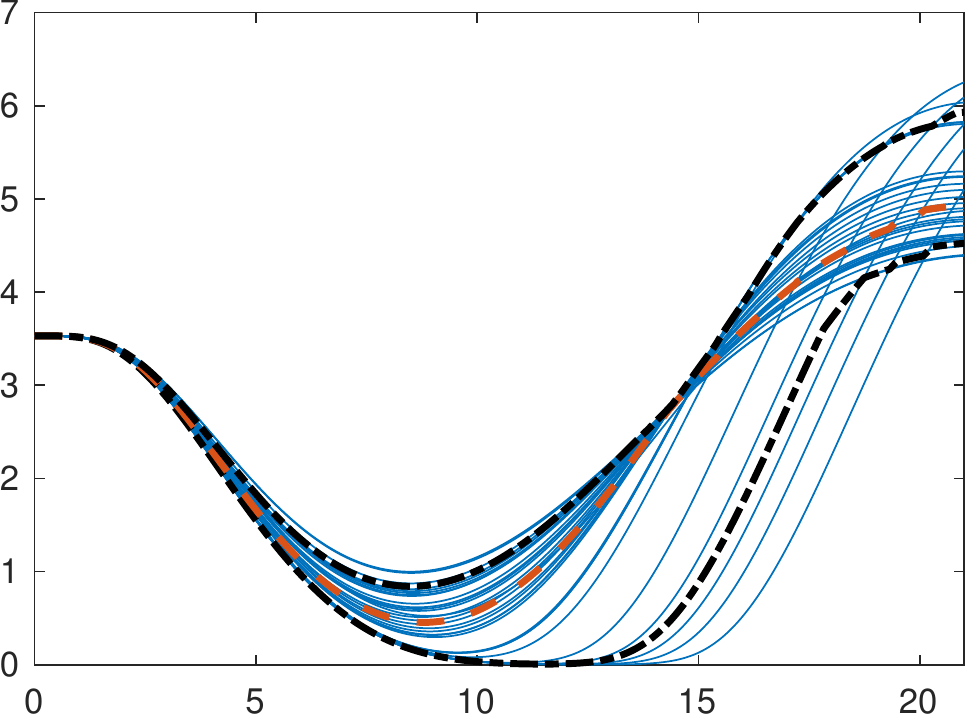}
          \put(-161,95){\rotatebox[origin=c]{90}{\scriptsize $Nt)$}}
     \put(-22,-6){\scriptsize Days}
     \end{center}
    \caption{$n=4$}
    \label{fig:Nn4}
  \end{subfigure}
  \begin{subfigure}[b]{0.5\linewidth}
    \begin{center}
    \includegraphics[width=0.9\linewidth]{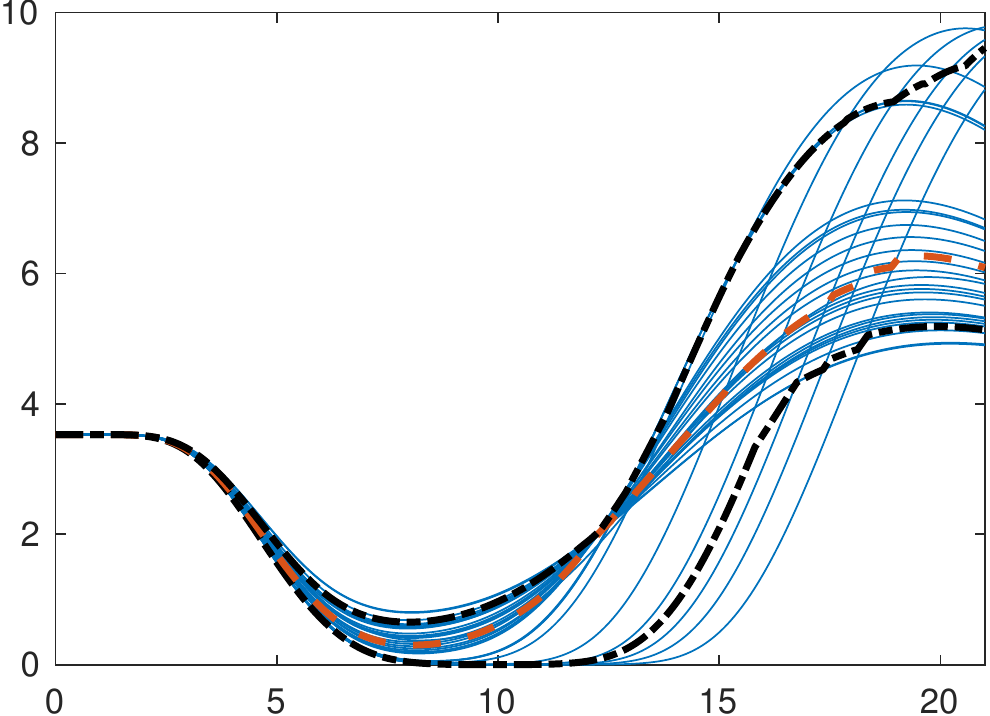}
          \put(-161,82){\rotatebox[origin=c]{90}{\scriptsize$N(t)$}}
     \put(-22,-6){\scriptsize Days}
     \end{center}
    \caption{$n=10$}
    \label{fig:Nn10}
  \end{subfigure}
  \vspace{4ex}
  \begin{subfigure}[b]{0.5\linewidth}
    \begin{center}
    \includegraphics[width=0.9\linewidth]{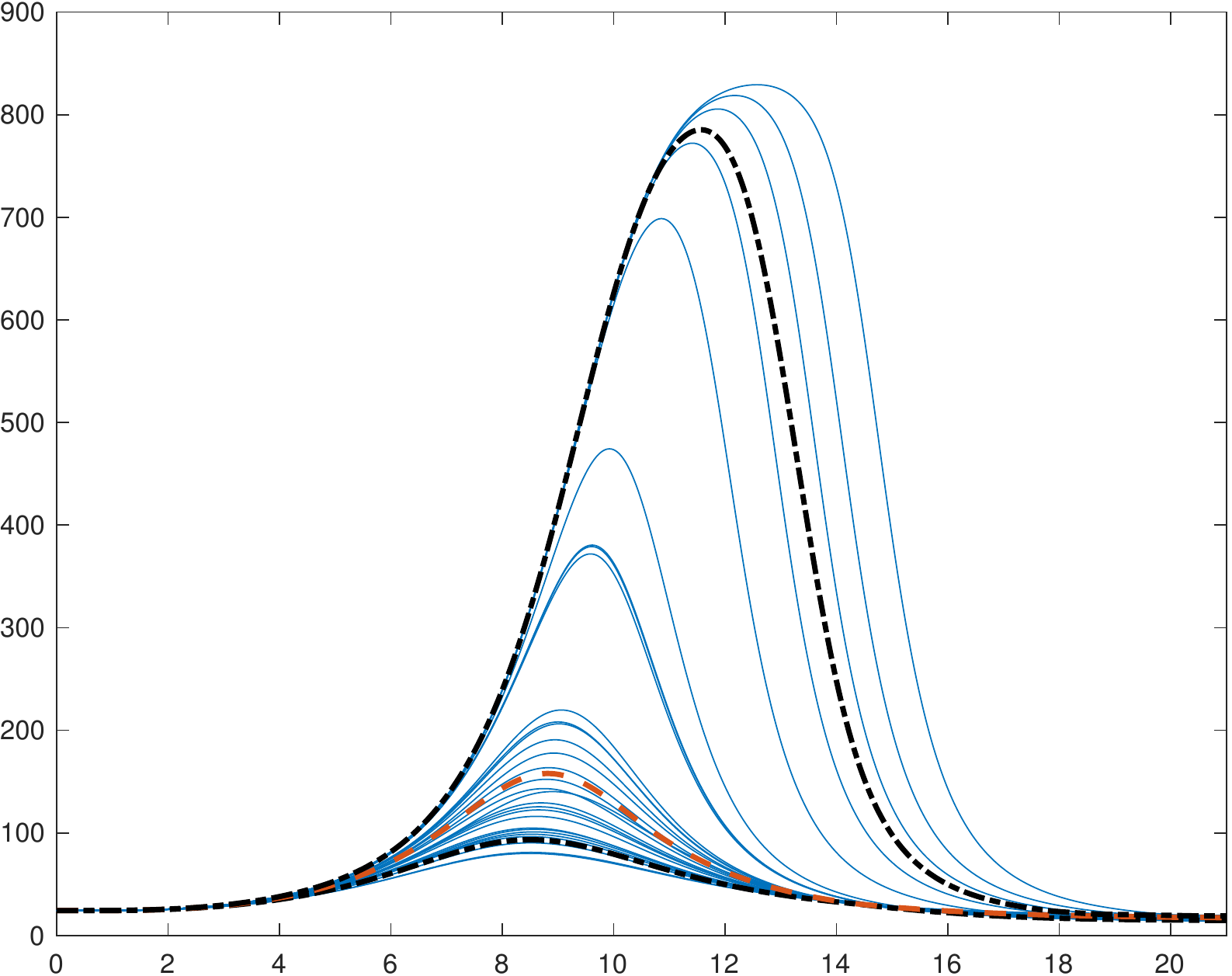}
          \put(-161,95){\rotatebox[origin=c]{90}{\scriptsize $G(t)$}}
     \put(-22,-6){\scriptsize Days}
     \end{center}
    \caption{$n=4$}
    \label{fig:Gn4}
  \end{subfigure}
  \begin{subfigure}[b]{0.5\linewidth}
    \begin{center}
    \includegraphics[width=0.9\linewidth]{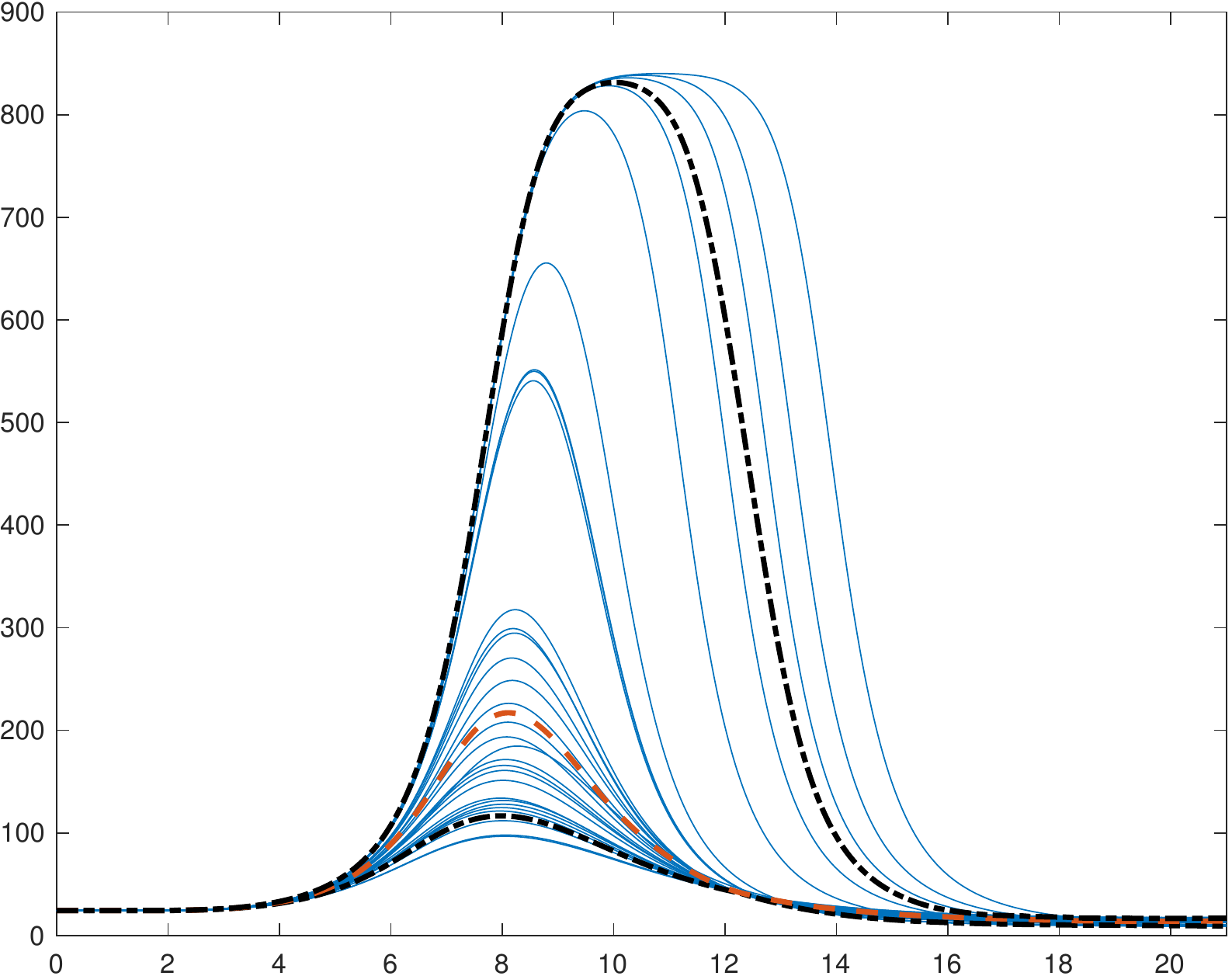}
          \put(-161,82){\rotatebox[origin=c]{90}{\scriptsize$G(t)$}}
     \put(-22,-6){\scriptsize Days}
     \end{center}
    \caption{$n=10$}
    \label{fig:Gn10}
  \end{subfigure}
   \caption{Impact of IIV on the Quartino model \eqref{eq:QuartinoModelLimited} and its equivalent forms \eqref{lctsd2} and \eqref{lctsd3} for two different values of $n$. PopPK parameters for 30 virtual patients were generated following \cite{Bruno1996} and used as inputs to \eqref{eq:QuartinoModelLimited}.
    Solid blue lines: individual predictions; black dotted-dashed lines: 10th and 90th percentiles of predictions; red dashed lines: median prediction.}
  \label{fig:PopPK}
\end{figure}

\begin{figure}[ht]
\begin{subfigure}[b]{0.5\linewidth}
    \begin{center}
    \includegraphics[width=0.9\linewidth]{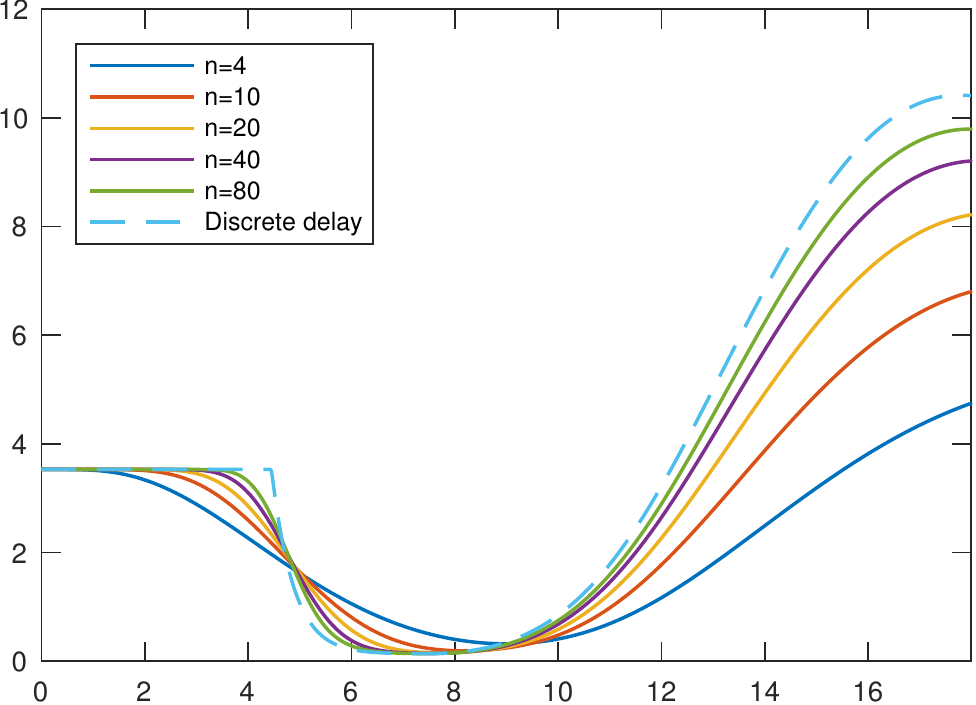}
     \put(-163,95){\rotatebox[origin=c]{90}{\scriptsize$N(t)$}}
     \put(-20,-6){\scriptsize Days}
     \end{center}
    \caption{Neutrophils}
    \vspace{4ex}
  \end{subfigure}
  \begin{subfigure}[b]{0.5\linewidth}
    \begin{center}
    \includegraphics[width=0.9\linewidth]{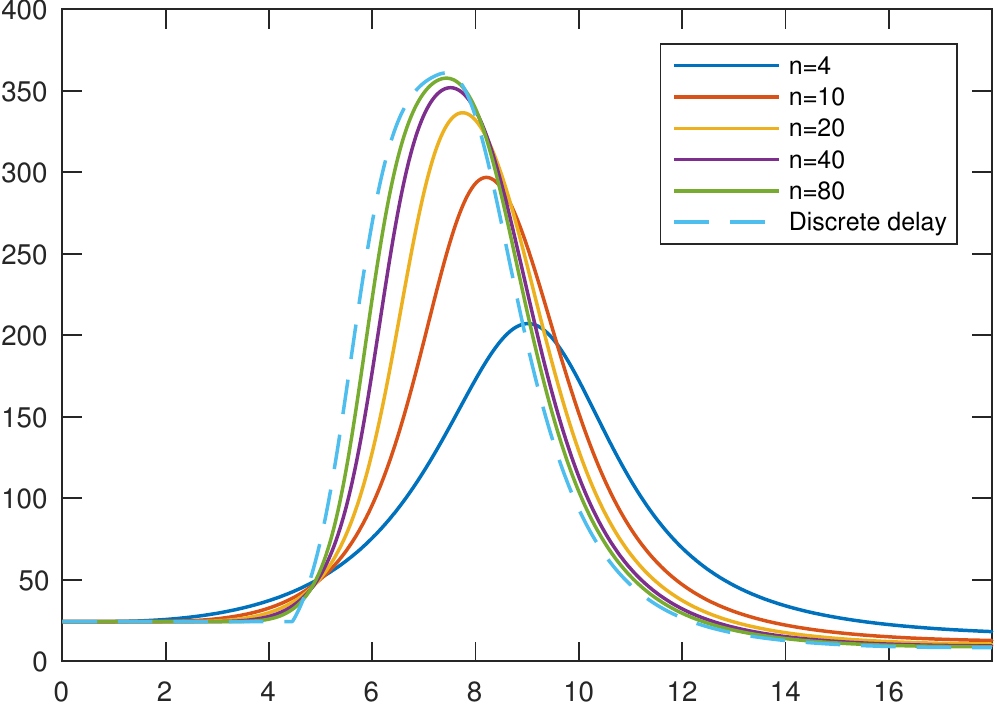}
          \put(-163,95){\rotatebox[origin=c]{90}{\scriptsize$G(t)$}}
     \put(-20,-6){\scriptsize Days}
     \end{center}
    \caption{G-CSF}
    \vspace{4ex}
  \end{subfigure}
\caption{Neutrophil and G-CSF concentrations from the generalised Quartino model \eqref{eq:QuartinoModelLimited} and the equivalent time-rescaled discrete delay model \eqref{lctsd5} for increasing values of $n$. Since the discrete delay model \eqref{lctsd5} is expressed in the time-rescaled $\hat{t}(t)$, we inverted $\eqref{lctsd1}$ and mapped the simulated solution back to $t$ to directly compare to \eqref{eq:QuartinoModelLimited}. Dashed line: solution of discrete delay.}
 \label{fig:Increasing_n}
\end{figure}

The analysis and results of Section~\ref{sec:DDEBifurcation} indicate the interest of performing stability and bifurcation studies in the pharmaceutical sciences setting.  Through sensitivity analysis, we have previously concluded that models constructed from first-principles, such as \eqref{eq:DDEmodel}, are robust to PKs by ``sufficiently'' accounting for the system's physiological mechanisms \cite{Craig2016}. Thus average PK parameters without reference to the full Pop-PK model are reliably predictive of the behaviour of such physiological models. Here we extend our previous work to study how PK variability (variations in G-CSF concentrations) affect the stable states of the PD (physiological) system.

The granulopoietic system is indeed very robust around homeostatic parameter values. An increase in G-CSF concentrations corresponds to a decrease in the value of $\gamma_Q$, and, from Table~\ref{tab.hopf1}, destabilisation of this equilibrium can only occur when $\gamma_q> \gamma^h_Q$, and thus such changes preserve homeostasis. With regard to changes in the proliferative  progenitor compartment, either by variations in the parameters $\eta^h_{N_P}$ or $b_v$, physiologically realistic scenarios preclude reaching the bifurcation values: for example, in the case of $\eta^h_{N_P}$, artificial removal of G-CSF from the body would be required.

\section{Discussion}
\label{sec:discussion}

Mathematical pharmacology, defined as the study of mathematical approaches to pharmacological processes, is increasingly recognised as a quantitative methodology critical to understanding pharmaceutical treatments  and their efficacy while simultaneously raising compelling mathematical problems \cite{vanderGraaf2016}. Using granulopoiesis as a backdrop, in the present paper we have examined the connections between the familiar PK/PD model formalism originally proposed by Friberg \cite{Friberg2002} and adapted by Quartino \cite{Quartino2014} and a discrete delay model of neutrophil production, connected via a distributed delay model. Crucially, we have shown how the stability of each model can be studied straightforwardly via this latter distributed delay model, underlining the advantage of being able to transfer between these equivalent expressions and motivating the present analysis.
We examined the impact of the inclusion of IIV on the solutions to each of the models, and determined that variations are driven through increases to $n$ rather than the presence of variability; as $n\to\infty$, solutions of the generalised (time-rescaled) Quartino transit compartment model converge to that of the time-rescaled discrete DDE model.
Last, using our previously published QSP model of the negative feedback relationship between granulopoiesis and G-CSF, we have identified several Hopf bifurcations through bifurcation analysis, a technique not commonly applied in the classical PK/PD analyses, and reviewed the impact the interpretation of such bifurcations can have on our understanding of pharmacological systems when used in concert with more common sensitivity and variability analyses. We would like to highlight two results in particular. First, the distributed delay model \eqref{lctsd3} exhibits wider regions of stability around the steady state $\mathbf{Y}_2^*$ as compared to the discrete DDE model \eqref{lctsd5}, consistent with the result that ``distributed delays are inherently more stable than the same system with discrete delays'' \cite{Campbell2009}. Second, we identified Hopf bifurcations in the distributed and discrete delay forms of the Quartino model with periods corresponding to those in cyclic neutropenic patients by varying the feedback parameter $\gamma$ and the delay $\tau$, demonstrating how bifurcation analyses can be applied in mathematical pharmacology to understand the pathogenesis towards diseases.

Perhaps the most immediately consequential conclusion drawn here is the incorrect definition of the mean transit/maturation time in the original and subsequent applications of the Friberg model. As previously mentioned, by setting $\kP=a=k_{tr}$ with $G(t)=G_0$, the MTT(MMT) was originally expressed as $(n+1)/k_{tr}$. However, we have shown that the mean delay of the distributed delay model is constant and instead given by $\tau=n/a$ (which, when $a=k_{tr}$ is then clearly given by $n/k_{tr}$). Thus it is mathematically incorrect to set MMT$=(n+1)/a$ as it treats the proliferative pool as an additional transit compartment and this formulation cannot be recovered via the linear chain technique.
The generalised Quartino model \eqref{eq:QuartinoModelLimited} explicitly decouples the maturation time and the production rate of cells to eliminate this problem.
Additionally, we highlighted the mathematical issue presented when $a$ is non-constant, as in \cite{Quartino2014}, to the derivation \eqref{lct3}, which is essential to the linear chain technique to recover the correct ODE formulation from the distributed delay model \eqref{lctq1}.

Further, since, in \cite{Friberg2002} and its various extensions and applications, the parameter $k_{tr}$ is determined via the MMT, the mean maturation time is fit and then the rate of transit through each compartment is determined via the equation MMT$=(n+1)/k_{tr}$. This leads to disparate estimates for the maturation process, ranging, for example, from 102 hours ($n=6$) in \cite{Quartino2012} to 210 hours ($n=4$) in \cite{Quartino2014}. Physiological labelling studies report a much narrower range of maturation times (6.4 days in \cite{Price} and 6.9 days in \cite{Dancey1976}, for example). Thus, allowing the MMT to vary widely is not physiologically consistent and further introduces additional mathematical difficulties since, in general, $\tau=n/a$, where $a$ is not necessarily equal to $k_{tr}$ nor $\kP$.

We therefore emphasise that this work provides further motivation to systematically incorporate,
from first principles, the physiological architecture yielding the
proper mathematical formulation of pharmacological models.

\section*{Acknowledgements}

DCS was supported by National Council for Scientific and Technological
Development of Brazil (CNPq) postdoctoral fellowship 201105/2014-4.
MC was supported by an Natural Sciences and Engineering Research Council of Canada (NSERC) postdoctoral fellowship and grant DP5OD019851 from the Office of the Director at the National Institutes of Health to her PI. 
TC was supported by the Alberta government via the Sir James Lougheed Award of Distinction as well as the Centre de Recherche Math\'ematiques, Montr\'eal.
FN and JL are funded by FN's NSERC Industrial Chair in Pharmacometrics, supported by Novartis, Pfizer, and inVentiv Health Clinics, and an FQRNT projet d'\'equipe. JB and ARH~are grateful to NSERC
for funding through the Discovery Grant program.
We are appreciative for our many very useful discussions with Michael C. Mackey.



\newpage
\appendix

\noindent
\textbf{\large Appendices}

\section{Time Rescaling of Quartino Model}
\label{app.timerescale}

Here we show how the time rescaling \eq{lctsd1} we applied to the generalised Quartino model \eq{eq:QuartinoModelLimited} relates to the state-dependent delays that are used in the QSP granulopoiesis model \eq{eq:DDEmodel}.

For the Friberg model \eq{gdd2} we have average maturation delay $\tau$ given by $\tau=n/a$. Hence
$$n=a\tau=\int_{t-\tau}^t a ds.$$
For the Quartino model \eq{eq:QuartinoModelLimited} the maturation rate $a$ is replaced by $a(G/G_0)^\beta$
and hence the time-dependent maturation delay $\alpha(t)$ for this model is given by
$$n=\int_{t-\alpha(t)}^t a \left(\frac{G(s)}{G_0}\right)^\beta ds.$$
Thus the mean maturation time $\tau$ for the time-rescaled Quartino model \eq{lctsd2} is related to $\alpha(t)$ by
\be \label{alphathres}
\tau=\frac{n}{a}=\int_{t-\alpha(t)}^t \left(\frac{G(s)}{G_0}\right)^\beta ds.
\ee
Equation \eq{alphathres} defines $\alpha(t)$ by a threshold condition. This is completely analogous to the threshold condition \eq{tauNMthres} used to define the state-dependent delay $\tauNM(t)$ in the QSP model \eq{eq:DDEmodel}.

Differentiating \eq{alphathres} using Leibniz rule we obtain an expression for the evolution of $\alpha(t)$ as
$$0=\left(\frac{G(t)}{G_0}\right)^\beta - \Bigl(1-\frac{\textrm{d}\alpha}{\textrm{d}t}\Bigr)\left(\frac{G(t-\alpha(t))}{G_0}\right)^\beta,$$
which can be rewritten as
\be \label{alphaev}
\frac{\textrm{d}\alpha}{\textrm{d}t}=1-\left(\frac{G(t)}{G(t-\alpha(t))}\right)^\beta,
\ee
and determines the evolution of $\alpha(t)$. An analogous expression was derived in Craig~\cite{Craig2016c} for the evolution of $\tauNM(t)$ in the QSP model \eq{eq:DDEmodel}.

\section{Positivity of Solutions}

We show positivity of the solutions to the models considered in the paper.

\subsection{Positivity of solutions to the Quartino endogenous G-CSF model}
\label{sec:QuartinoPositivity}

Consider the generalised Quartino model \eqref{eq:QuartinoModelLimited}.

\begin{lemma}\label{GLemma1}
Assume that the parameters in \eqref{eq:QuartinoModelLimited} are strictly positive and that $G(0)\geq 0$.
\begin{enumerate}[i)]
\item
Then $G(t)>0$ for all time $t>0$.
\item
If $N(t)\geq0$ for $t\in[0,\TN]$ then $G(t)\leq\max\{k_{in}/k_{e},G(0)\}$ for all $t\in[0,\TN]$.
\end{enumerate}
\end{lemma}

\begin{proof}
i) Either $G(0)>0$ or $G(0)=0$ and $\frac{\textrm{d}G}{\textrm{d}t}(0)=k_{in}>0$. In both cases $G(t)>0$ for $t\in(0,\epsilon)$ for some $\epsilon>0$. Assume, for contradiction, that there exists $\TG>0$ such that $G(\TG)=0$ but
$G(t)>0$ for $t\in(0,\TG)$. Since $G$ is decreasing at $t=\TG$ this implies that $\frac{\textrm{d}G}{\textrm{d}t}(\TG)\leq0$.
But this is contradicted by \eq{eq:QuartinoModelLimited} which implies that $\frac{\textrm{d}G}{\textrm{d}t}(\TG)=k_{in}>0$ if $G(\TG)=0$. Thus there exists no such that time $\TG$, and hence $G(t)>0$ for all $t>0$. \qed

\noindent
ii) For $t\in[0,\TN]$ we have $\frac{\textrm{d}G}{\textrm{d}t}\leq k_{in}-k_{e}G$, and hence $\frac{\textrm{d}G}{\textrm{d}t}<0$ if
$G(t)>k_{in}/k_{e}$. The result follows. \qed
\end{proof}

\begin{lemma}\label{PLemma}
Assume that the parameters in \eqref{eq:QuartinoModelLimited} are strictly positive, and that the initial conditions satisfy $P(0)>0$ and $G(0)\geq 0$. Furthermore, assume that there exists $\TN>0$ such that
$N(t)\geq0$ for $t\in[0,\TN]$. Then $P(t)>0$ for all time $t\in[0,\TN]$.
\end{lemma}

\begin{proof}
By Lemma~\ref{GLemma1} for $t\in[0,\TN]$ we have $0 \leq G(t) \leq M=\max\{k_{in}/k_{e},G(0)\}$.
Thus from \eq{eq:QuartinoModelLimited} we have
$$\frac{\textrm{d}P}{\textrm{d}t} \geq -k_{tr} \left( \frac{M}{G_0}\right)^{\beta}P(t),$$
which (using the continuous Gronwall lemma) implies  that
$$P(t) \geq P(0) \exp \left[ - k_{tr}t\left(\frac{M}{G_0} \right)^{\beta}\right] > 0, \qquad t\in[0,\TN]. \eqno{\qed}$$
\end{proof}

\begin{lemma}\label{TransitLemma}
Assume that the parameters in \eqref{eq:QuartinoModelLimited} are strictly positive, and that the initial conditions satisfy $P(0)>0$, $G(0)\geq 0$ and $T_j(0)\geq0$ for $j=1,2,\ldots,n$.
Then there exists $\epsilon>0$ such that $T_j(t)>0$ for all time $t\in(0,\epsilon)$ for all $j=1,2,\ldots,n$.
Furthermore, if there exists $\TN>0$ such that
$N(t)\geq0$ for $t\in[0,\TN]$ then $T_j(t)>0$ for all time $t\in(0,\TN]$ for all $j=1,2,\ldots,n$.
\end{lemma}

\begin{proof}
We proceed by induction on $j$. For $j=1$, either (i) $T_1(0)>0$ and $G(0) \geq 0$, or (ii) $T_1(0)=0$ and $G(0)>0$, or (iii)
$T_1(0)=G(0)=0$. In case (ii)
$\frac{\textrm{d}T_1}{\textrm{d}t}(0)=k_{tr}\left(\frac{G(0)}{G_0}\right)^{\beta}P(0)>0$,
while in case (iii) $T_1(0)=\frac{\textrm{d}T_1}{\textrm{d}t}(0)=0$.

For all three cases, $P(0)>0$ implies that that there exists $\epsilon>0$ such that $P(t)>0$ for $t\in(0,\epsilon)$. Lemma~\ref{GLemma1} ensures also that $G(t)>0$ for $t\in(0,\epsilon)$. But now, for $t\in(0,\epsilon)$ if $T_1(t)\leq0$ we have the strict inequality
$$\frac{\textrm{d}T_1}{\textrm{d}t}(t) = \left(\frac{G(t)}{G_0}\right)^\beta(k_{tr}P(t)-aT_1(t))>0.$$
Thus if there exists $t_1\in(0,\epsilon)$ such that $T_1(t_1)\leq0$ then $\TimeDeriv T_1(t)>0$ for all $t<t_1$ and therefore
$T_1(0)<0$, which contradicts the initial condition. Hence there exists no such time $t_1$, and so $T_1(t)>0$ for $t\in(0,\epsilon)$.

For general $j$, it is shown that $T_j(t)>0$ for $t\in(0,\epsilon)$ similarly.
The positivity of $G(t)>0$ and $T_{j-1}(t)>0$ for $t\in(0,\epsilon)$ along with the strict inequality
$$\frac{\textrm{d}T_j}{\textrm{d}t}(t) = a\left(\frac{G(t)}{G_0}\right)^\beta(T_{j-1}(t)-T_j(t))>0,$$
for all $t\in(0,\epsilon)$ if $T_j(t)\leq0$ similarly ensures that actually $T_j(t)>0$ for $t\in(0,\epsilon)$, which establishes the result for any finite $n$.

Now, from Lemmas~\ref{GLemma1} and~\ref{PLemma}, if
there exists $\TN>0$ such that
$N(t)\geq0$ for $t\in[0,\TN]$ then $P(t)>0$ for all time $t\in[0,\TN]$ and
$0 \leq G(t) \leq M=\max\{k_{in}/k_{e},G(0)\}$. Now we have $T_1(t)>0$ for $t\in(0,\epsilon)$ and for $j=1$
\be \label{Tjbd}
\frac{\textrm{d}T_j}{\textrm{d}t} \geq -a \left( \frac{M}{G_0}\right)^{\beta}T_j(t),
\ee
which, similar to the proof of Lemma~\ref{PLemma}, implies $T_1>0$ for all $t\in(0,\TN]$. But if
$T_{j-1}(t)>0$ for all $t\in(0,\TN]$ then $T_{j}(t)>0$ satisfies \eq{Tjbd} and it follows by induction that
$T_{j}(t)>0$ for all $t\in(0,\TN]$ for $j=1,\ldots,n$ if $N(t)\geq0$ for $t\in[0,\TN]$. \qed
\end{proof}


\begin{theorem}\label{ODETheorem}
Assume that the parameters in \eqref{eq:QuartinoModelLimited} are strictly positive, and that the initial conditions satisfy $P(0)>0$, $G(0)\geq 0$, $N(0)\geq 0$ and $T_j(0)\geq0$ for $j=1,2,\ldots,n$.
Then $P(t)>0$, $0<G(t)\leq M=\max\{k_{in}/k_{e},G(0)\}$, $N(t)>0$ and $T_j(t)>0$ for $j=1,2,\ldots,n$ for all $t>0$.
\end{theorem}

\begin{proof}
Lemma~\ref{GLemma1} implies that $G(t)>0$ for all $t>0$. While $P(0)>0$ and Lemma~\ref{TransitLemma}
imply that $P(t)>0$ and $T_j(t)>0$ for $j=1,2,\ldots,n$ for $t\in(0,\epsilon)$.

Now similar to the proof of Lemma~\ref{TransitLemma},
$$\frac{\textrm{d}N}{\textrm{d}t}(t) = a\left(\frac{G(t)}{G_0}\right)^\beta T_n(t)-\kcirc N(t)>0, \qquad\text{if}\qquad N(t)\leq0,$$
implies that $N(t)>0$ for $t\in(0,\epsilon)$. Then while $N(t)>0$, we have $G(t)\leq M=\max\{k_{in}/k_{e},G(0)\}$ and
$$\frac{\textrm{d}N}{\textrm{d}t}(t) \geq -\kcirc N(t).$$
But this last inequality implies that $N(t)>0$ for all $t>0$.
It follows from Lemmas~\ref{GLemma1}, \ref{PLemma} and~\ref{TransitLemma}
that $G(t)\leq M=\max\{k_{in}/k_{e},G(0)\}$, $P(t)>0$, and $T_j(t)>0$ for $j=1,\ldots,n$ for all $t>0$. \qed
\end{proof}

Although an essential part of the solution positivity proofs in this section was to show that solutions decay with at most a bounded linear rate for general parameters, these are not the dynamics that we actually expect to observe. For normal individuals/subjects we should have $\gamma>\beta$ (as is the case for the parameters in Table~\ref{tab:QuartinoValues}), in which case when $G(t)$ is sufficiently large $\frac{\textrm{d}P}{\textrm{d}t}$ is positive, as is required for the feedback loops to function effectively.


Since the Quartino model \eq{eq:QuartinoModelLimited} is equivalent to the time-rescaled Quartino model \eqref{lctsd2}
and the distributed DDE \eqref{lctsd3}, positivity of solutions to those models follows directly from
Theorem~\ref{ODETheorem} (it is important to note that $\that(t)>0$ for all $t>0$ follows from \eq{lctsd1} and the positivity of $G(t)$). However, this only establishes positivity of solutions for the distributed DDE \eqref{lctsd3}
when $n$ is an integer. When can show positivity directly for both the distributed DDE \eqref{lctsd3} for general real $n$ and also for the discrete DDE \eqref{lctsd5}.

\begin{theorem}\label{DDETheorem}
Assume that the parameters in distributed delay DDE \eqref{lctsd3} or the discrete delay DDE \eqref{lctsd5}
are strictly positive, and that the initial conditions satisfy $G(0)>0$, $N(0)\geq 0$, and $P\brackthat=\phi\brackthat$ for $\hat{t}\leq0$ where $\phi$ is continuous and $P(0)=\phi(0)>0$. Then
Then $P\brackthat>0$ and $0<G\brackthat\leq M=\max\{k_{in}/k_{e},G(0)\}$ for all $\hat{t}>0$. Finally $N\brackthat>0$ for all $\hat{t}>0$ for the system \eq{lctsd3}, and for all $\hat{t}\>\tau$ for the system \eq{lctsd5}.
\end{theorem}

\begin{proof}
The proof uses similar ideas to the proof of Theorem~\ref{ODETheorem}, so we just outline the details here. We have $P(0)>0$, and similar to the proof of Lemma~\ref{PLemma} the rate of decrease of $P\brackthat$ is bounded so $P\brackthat>0$ for all $\hat{t}>0$. But now for \eqref{lctsd3}, the positivity of $P$ ensures that $\TimeDeriv N>0$ if $N=0$, which ensures that $N\brackthat>0$ for all $\hat{t}>0$. For the model \eqref{lctsd5}, $P(\hat{t}-\tau)>0$ for $\hat{t}\geq\tau$ leads to the positivity of $N\brackthat>0$ for all $\hat{t}>\tau$.
Finally the bounds on $G$ are derived similarly to Lemma~\ref{GLemma1}. \qed
\end{proof}

\subsection{Positivity of solutions to the QSP granulopoiesis model}
\label{sec:PositivitySolutionsDDE}

Consider the QSP granulopoiesis model \eq{eq:DDEmodel}.
Using the constraints listed in \eqref{Constraints} and setting
\begin{equation*}
G_{BF} = \frac{G_2(t)}{V[\NR(t)+N(t)]},
\end{equation*}
yields
\begin{equation}
\begin{array}{lll}
V(G_1(t))  >  0 & \textrm{if} & G_1(t) > 0 \\
\ftrans(G_{BF}(t))  >  0 & \textrm{if} & G_{BF}(t) >  0 \\
\kappa(G_1(t))  >  0 & \textrm{if} & G_1(t) > 0 \\
\beta(Q(t))  > 0 & \textrm{if} & Q(t) > 0. \\
\end{array}
\label{ConstraintInequality}
\end{equation}

\begin{lemma}\label{G1G2Lemma}
Consider the initial value problem $\mathcal{P}$, given by equations~\eqref{eq:HSCs}--~\eqref{eq:BoundGCSF} and the ICs and histories given in \eqref{FullModel}, with $G_{1,0}+G_{2,0} >0$. Assume that $ \Gprod, k_{ren}$ and $ k_{int}$ are positive constants and set $\beta = \max\{k_{ren},k_{int}\}$. Then
\begin{equation*}
G_1(t)+G_2(t) \geq K e^{-\beta t} \quad \textrm{with} \quad K > 0 \quad \forall t >0.
\end{equation*}
\end{lemma}
\begin{proof}
A simple calculation shows
\begin{equation*}
\begin{array}{lll}
\TimeDerivD [G_1(t)+G_2(t)] & = & \Gprod -(k_{ren}G_1(t)+k_{int}G_2(t)) \\
& \geq & \Gprod -\beta (G_1(t)+G_2(t)). \\
\end{array}
\end{equation*}

Multiplying by $ e^{\beta t}$ and rearranging gives
\begin{equation*}
\TimeDeriv [(G_1(t)+G_2(t))e^{\beta t}] \geq \Gprod e^{\beta t}.
\end{equation*}

Integrating the inequality from $0$ to $t$ yields
\begin{equation*}
G_1(t)+G_2(t) \geq \frac{\Gprod}{\beta}(1-\frac{1}{e^{\beta t}}) +   [G_{1,0}+G_{2,0}] e^{-\beta t} > [G_{1,0}+G_{2,0}] e^{-\beta t}.
\end{equation*}
Taking $ K =[G_{1,0}+G_{2,0}]$ yields the claim. \qed
\end{proof}

Rearranging the bound of Lemma \ref{G1G2Lemma} gives:
\begin{equation}
G_2(t) \geq K e^{-\beta t} - G_1(t) .
\label{G2Identity}
\end{equation}

\begin{lemma}\label{G1Lemma}
Consider the initial value problem, $\mathcal{P}$, given equations~\eqref{eq:HSCs}--~\eqref{eq:BoundGCSF} and the ICs and histories given in \eqref{FullModel}, with $G_{1,0}+G_{2,0} >0$ and $G_{1,0} \geq 0$. Assume that $ \Gprod, k_{ren}, k_{21}$ and $ k_{int}$ are strictly positive constants. Then, $G_1(t) > 0 $ for all $t >0$.
\end{lemma}
\begin{proof}
Assume, for contradiction, that there exists a time $s >0$ such that $G_1(s) < 0$. As $G_1$ is a solution of the differential equation, it is continuously differentiable. By the IVT, there must exist a time $T^*$ such that $G_1(T^*) =0$ with $ s > T^*$. Using \eqref{G2Identity}, at $t= T^*$:
\begin{equation*}
\begin{array}{lll}
\TimeDerivD G_1(t)|_{t = T^*} & = & \Gprod + k_{21}G_2(T^*)   \\
 & \geq & \Gprod + k_{21}(K e^{-2\beta T^*} - G_1(T^*))  \\
 & > & \Gprod \quad  >\quad 0. \\
\end{array}
\end{equation*}
The Mean Value Theorem (MVT) yields a contradiction to $G_1(s) < 0$. \qed
\end{proof}

\begin{lemma}\label{QLemma}
Consider the initial value problem $\mathcal{P}$, equations~\eqref{eq:HSCs}--~\eqref{eq:BoundGCSF} and the ICs and histories given in \eqref{FullModel}, with $G_{1,0}+G_{2,0} >0$, $G_{1,0} \geq 0$ and $\phi_1(s) \geq 0 $ for $ s \in [-\tauQ,0]$ with $\phi_1(s^*) > 0$ for at least one $s^*$. Assume that $ \Gprod, k_{ren}, k_{21}$ and $ k_{int}$ are strictly positive constants. Then, $Q(t) > 0$ for all time  $t >0$.
\end{lemma}
\begin{proof}
Assume, for contradiction, that there exists a $s^*$ such that $Q(s^*) = 0$. The assumption on $\phi_1$ and the continuity of $Q$ allows us to assume that $Q(s) > 0$ for all $s<s^*$.

The constraints in \eqref{ConstraintInequality} ensures that $\beta(Q(t-\tauQ))>0$.

Then:

\begin{equation*}
\begin{array}{lll}
\TimeDerivD Q(t) & = & -(\kappa(G_{1}(t))+\kappa_{\delta}+\beta(Q(t)))Q(t)+A_{Q}(t)\beta(Q(t-\tauQ))Q(t-\tauQ) \\
 & \geq & -(\kappa(G_{1}(t))+\kappa_{\delta}+\beta(Q(t)))Q(t). \\
\end{array}
\end{equation*}
Using the result of Lemma \ref{G1Lemma} and that $\kappa(G_1(t))$ is a monotonically increasing function for $G_1(t) >0$ gives
\begin{equation*}
\kappa(G_1(t)) \leq \lim \limits_{G_1(t) \to \infty} \kappa(G_1(t)) = 2\kappa^*-\kappa^{min} = \kappa^{max} < \infty.
\end{equation*}
Setting $\xi = \kappa^{max}+ \kappa_{\delta} + f_Q$ gives:
\begin{equation*}
\begin{array}{lll}
\TimeDeriv Q(t) & \geq & -\xi Q(t). \\
\end{array}
\end{equation*}
An argument similar to that in Lemma \ref{G1G2Lemma} ensures that $Q(t) >0$  for all $t >0$. This contradicts the existence of $s^*$ and establishes the positivity of $Q(t)$ for all $t>0$. \qed

\end{proof}

\begin{lemma}\label{NRNLemma}
Consider the initial value problem $\mathcal{P}$, given by equations~\eqref{eq:HSCs}--~\eqref{eq:BoundGCSF} and the ICs and histories given in \eqref{FullModel}, with $G_{1,0}+G_{2,0} >0$, $\phi_1(s) \geq 0 $ for $ s \in [-\tau,0]$ with $\phi_1(s^*) > 0$ for at least one $s^*$, $\phi_2 \geq 0$ for $ s \in [-\tau,0]$ and  $N_0+ N_{R,0} > 0$. Assume that $ \Gprod, k_{ren}, k_{21}$ and $ k_{int}$ are strictly positive constants. Then  $\NR(t)+N(t) > 0$  for all $t >0$.
\end{lemma}
\begin{proof}
A simple calculation gives:
\begin{equation*}
\begin{array}{lll}
\TimeDerivD (\NR(t)+N(t)) & = &  A_{N}(t)\kappa(G_{1}(t-\tauN(t)))Q(t-\tauN(t))\frac{\VN(G_{1}(t))}{\VN(G_{1}(t-\tauNM(t)))} \\
 & & - \gammaNR \NR(t) - \gamma_N N(t).
\end{array}
\end{equation*}
Applying Lemmas  \ref{G1Lemma}, \ref{QLemma} and \eqref{ConstraintInequality}, the following holds
\begin{equation*}
 A_{N}(t)\kappa(G_{1}(t-\tauN(t)))Q(t-\tauN(t))\frac{\VN(G_{1}(t))}{\VN(G_{1}(t-\tauNM(t)))} \geq 0.
\end{equation*}
Define $\alpha = \textrm{max}[\gammaNR,\gamma_N]$ and calculate:
\begin{equation*}
\begin{array}{lll}
\TimeDerivD (\NR(t)+N(t)) & \geq &  -\alpha (\NR(t)+N(t)). \\
\end{array}
\end{equation*}
An argument similar to that of Lemma \ref{G1G2Lemma} gives the positivity of the sum. \qed
\end{proof}

\begin{lemma}\label{G2Bound}
Consider the initial value problem $\mathcal{P}$ given by equations~\eqref{eq:HSCs}--~\eqref{eq:BoundGCSF} and the ICs and histories given in \eqref{FullModel} such that $\phi_1(s) \geq 0 $ for $ s \in [-\tau,0]$ with $\phi_1(s^*) > 0$ for at least one $s^*$, $\phi_2 \geq 0$ for $ s \in [-\tau,0]$ and
\begin{equation*}
\begin{array}{lll}
G_{2,0} & \leq & V[N_{R,0}+N_{0}] \\
G_{1,0}+G_{2,0}&  >& 0 \\
N_0+ N_{R,0} &  > & 0 \\
\end{array}
\end{equation*}
Assume that $ k_{21}$ and $ k_{int}$ are strictly positive constants. Then, $G_2(t) \leq V[\NR(t)+N(t)]$ for all $t >0$.
\end{lemma}
\begin{proof}
Assume, for contradiction, that there is a time $s > 0$ with $G_2(s) > V[\NR(t)+N(t)]$. Then, there must exist a time $T^*$ such that $G_2(T^*) = V[\NR(T^*)+N(T^*)]$. At $t = T^*$:
\begin{align*}
\TimeDerivD G_{2}(t)|_{t=T^*} &= -k_{int}G_{2}(T^*)+k_{12}[(\NR(T^*)+N(T^*))V-G_{2}(T^*)](G_{1}(T^*))^{s_{G}}-k_{21}G_{2}(T^*) \nonumber\\
 & \leq -2\textrm{max}[k_{int},k_{21}]G_{2}(T^*)
\end{align*}
\sloppy{Lemma \ref{NRNLemma} ensures that $G_2(T^*) = V[\NR(T^*)+N(T^*)] > 0$. Therefore $ \TimeDeriv G_{2}(t)|_{t=T^*} < 0$, which contradicts the MVT and there can be no $s$.} \qed
\end{proof}

\begin{lemma}\label{G2Lemma}
Consider the initial value problem $\mathcal{P}$ given by equations~\eqref{eq:HSCs}--~\eqref{eq:BoundGCSF} and the ICs and histories given in \eqref{FullModel} such that $\phi_1(s) \geq 0 $ for $ s \in [-\tau,0]$ with $\phi_1(s^*) > 0$ for at least one $s^*$, $\phi_2 \geq 0$ for $ s \in [-\tau,0]$ and
\begin{equation*}
\begin{array}{lll}
G_{2,0} & \leq & V[N_{R,0}+N_{0}] \\
G_{1,0}+G_{2,0}&  >& 0 \\
G_{2,0}&  \geq & 0 \\
N_0+ N_{R,0} &  > & 0 \\
\end{array}
\end{equation*}
 Assume that $ \Gprod, k_{ren}, k_{21}$ and $ k_{int}$ are strictly positive constants. Then, $G_2(t)>0 $ for all $t>0$.
\end{lemma}
\begin{proof}
Using Lemma \ref{G1Lemma} and the bound from Lemma \ref{G2Bound} gives
\begin{equation*}
k_{12}[(\NR(t)+N(t))V-G_{2}(t)](G_{1}(t))^{s_{G}} >0,
\end{equation*}
and
\begin{align*}
\TimeDerivD G_{2}(t) &= -k_{int}G_{2}(t)+k_{12}[(\NR(t)+N(t))V-G_{2}(t)](G_{1}(t))^{s_{G}}-k_{21}G_{2}(t)\nonumber \\
 & \geq  -2\textrm{max}[k_{int},k_{21}] G_2(t).
\end{align*}
An argument similar to Lemma \ref{G1G2Lemma} gives the result. \qed
\end{proof}

\begin{lemma}\label{NRLemma}
Consider the initial value problem $\mathcal{P}$ given by equations~\eqref{eq:HSCs}--~\eqref{eq:BoundGCSF} and the ICs and histories given in \eqref{FullModel} such that $\phi_1(s) \geq 0 $ for $ s \in [-\tau,0]$ with $\phi_1(s^*) > 0$ for at least one $s^*$, $\phi_2 \geq 0$ for $ s \in [-\tau,0]$ and
\begin{equation*}
\begin{array}{lll}
G_{2,0} & \leq & V[N_{R,0}+N_{0}] \\
G_{1,0}+G_{2,0}&  >& 0 \\
G_{1,0}&  \geq & 0 \\
G_{2,0}&  \geq & 0 \\
N_0+ N_{R,0} &  > & 0 \\
\end{array}
\end{equation*}
Assume that $ k_{21}$ and $ k_{int}$ are positive constants. Then, $\NR(t) > 0$  for all $t > 0$.
\end{lemma}
\begin{proof}
The function $A_N(t)$ is positive for all time, the constraints in \eqref{ConstraintInequality}, the assumptions on $\phi_{1,2}$ and Lemmas \ref{QLemma} and \ref{G1Lemma} ensure that $ A_{N}(t)\kappa(G_{1}(t-\tauN(t)))Q(t-\tauN(t))\frac{\VN(G_{1}(t))}{\VN(G_{1}(t-\tauNM(t)))} \geq 0$. Then
\begin{equation*}
\begin{array}{lll}
\TimeDerivD \NR(t) &=&  A_{N}(t)\kappa(G_{1}(t-\tauN(t)))Q(t-\tauN(t))\frac{\VN(G_{1}(t))}{\VN(G_{1}(t-\tauNM(t)))}\notag \\[2mm]
&&-(\gammaNR+\ftrans(G_{BF}(t)))\NR(t),\\[2mm]
 & \geq & -(\gammaNR+\ftrans(G_{BF}(t)))\NR(t) \\
\end{array}
\end{equation*}
Lemma \ref{G2Bound} and the positivity of $N(t)+\NR(t)$ and $G_2(t)$ bounds $G_{BF} \in [0,1]$. Therefore, $\ftrans(G_{BF})$ is a continuous function on a compact domain and is therefore bounded below by 0 and above by $\ftrans^{max}$. Then
\begin{equation*}
\begin{array}{lll}
\TimeDerivD \NR(t) &\geq& -(\gammaNR+\ftrans^{max})\NR(t),
\end{array}
\end{equation*}
and an argument similar to Lemma \ref{G1G2Lemma} yields the result. \qed
\end{proof}

\begin{lemma}\label{NLemma}
Consider the initial value problem $\mathcal{P}$ given by equations~\eqref{eq:HSCs}--~\eqref{eq:BoundGCSF} and the ICs and histories given in \eqref{FullModel} such that $\phi_1(s) \geq 0 $ for $ s \in [-\tau,0]$ with $\phi_1(s^*) > 0$ for at least one $s^*$, $\phi_2 \geq 0$ for $ s \in [-\tau,0]$  and
\begin{equation*}
\begin{array}{lll}
G_{2,0} & \leq & V[N_{R,0}+N_{0}] \\
G_{1,0}+G_{2,0}&  >& 0 \\
G_{1,0}&  \geq & 0 \\
G_{2,0}&  \geq & 0 \\
N_0+ N_{R,0} &  > & 0 \\
\end{array}
\end{equation*}
 Assume that $ k_{21}$ and $ k_{int}$ are strictly positive constants. Then, $N(t)> 0$ for all $t < 0$.
\end{lemma}
\begin{proof}
Using \eqref{ConstraintInequality} and the Lemmas \ref{G1Lemma} and \ref{NRLemma} to ensure the positivity of $\NR(t)$ and $\ftrans$ yields:
\begin{equation*}
\begin{array}{lll}
\TimeDerivD N(t) &=& \ftrans(G_{BF}(t))\NR(t)-\gamma_{N}N(t),\\[2mm]
 & \geq &  -\gamma_{N}N(t).\\
\end{array}
\end{equation*}
A similar argument to that used in Lemma \ref{G1G2Lemma} yields the result. \qed
\end{proof}

Together, these results lead to the following theorem.
  \begin{theorem}\label{PositivityTheorem}
 Consider the initial value problem $\mathcal{P}$ given by equations~\eqref{eq:HSCs}--~\eqref{eq:BoundGCSF} and the ICs and histories given in \eqref{FullModel} such that $\phi_1(s) \geq 0 $ for $ s \in [-\tau,0]$ with $\phi_1(s^*) > 0$ for at least one $s^*$, $\phi_2 \geq 0$ for $ s \in [-\tau,0]$  and
\begin{equation*}
\begin{array}{lll}
G_{2,0} & \leq & V[N_{R,0}+N_{0}] \\
G_{1,0}+G_{2,0}&  >& 0 \\
G_{1,0}&  \geq & 0 \\
G_{2,0}&  \geq & 0 \\
N_0+ N_{R,0} &  > & 0 \\
\end{array}
\end{equation*}
with $I_G(t) = 0 $. Moreover, assume that \eqref{Constraints} is satisfied along with strictly positive model parameters. Finally, assume that the history functions $\phi_{1,2}$ are positive at least once and are non-negative in their domain.
Then
\begin{enumerate}
\item the solutions of $\mathcal{P}$, $\mathbf{x}(t) = \left( Q(t),\NR(t),N(t),G_1(t),G_2(t) \right)$ remain component wise positive for all time.
\item the solutions of $\mathcal{P}$, $\mathbf{x}(t) = \left( Q(t),\NR(t),N(t),G_1(t),G_2(t) \right)$ remain component wise bounded for all time.
\end{enumerate}
 \end{theorem}
\begin{proof}
1. The results of Lemmas \ref{G1Lemma}; \ref{QLemma}; \ref{G2Lemma}; \ref{NRLemma}; and \ref{NLemma} give the result.

2. In \cite{Mackey1994}, the authors prove that solutions of the equation:
\begin{equation}
\TimeDeriv x(t) = -(\delta + \beta(x(t)))x(t) + A_Q\beta(x(t-\tauQ))x(t-\tauQ)
\end{equation}
are bounded above by a finite $x_1$. Setting $\tilde{\delta} = \kappa^{min}+\kappa_{\delta}$ gives:
\begin{equation*}
\begin{array}{lll}
\TimeDeriv Q(t) &=& -(\kappa(G_{1}(t))+\kappa_{\delta}+\beta(Q(t)))Q(t)+A_{Q}(t)\beta(Q(t-\tau_{Q}))Q(t-\tau_{Q}) \\[2mm]
 & \leq & (\tilde{\delta}+\beta(Q(t)))Q(t)+A_{Q}(t)\beta(Q(t-\tau_{Q}))Q(t-\tau_{Q}),
\end{array}
\end{equation*}
which implies that $Q(t)$ is bounded by some $Q_1$.

$V_{N_M}$ is an increasing function bounded above by $V_{max}$ and bounded below by $V(0) > 0$.  Finally,
\begin{equation*}
A_N(t) = \exp \left[ \int_{t-\tauNM}^{t-\tauN} \etaNP(G_1(s))\textrm{ds} - \gammaNM \tauNM(t) \right] \leq \exp [\etaNP^{max}\tauNP].
\end{equation*}
This gives
\begin{equation*}
\begin{array}{lll}
\TimeDeriv (N_R(t)+N(t)) & = & A_N(t)\beta(Q(t-\tau_N(t)))Q(t-\tau_N(t)) \frac{V_{N_M}(G_1(t))}{V_{N_M}(G_1(t-\tau_N(t)))} - \gammaNR N_R(t)-\gamma_N N(t) \\
& \leq & \exp [\etaNP^{max}\tauNP] f_0Q_1\frac{V_{max}}{V_{N_M}(0)}
\end{array}
\end{equation*}
Then, setting $K = \exp [\etaNP^{max}\tauNP] f_0Q_1\frac{V_{max}}{V_{N_M}(0)} $ and $\beta = \textrm{min}[\gamma_N, \gammaNR]$  gives
\begin{equation}
\TimeDeriv (N_R(t)+N(t)) \leq K - \beta( N_R(t) + N(t)).
\label{EQ:NDifferentialInequality}
\end{equation}

Equation \eqref{EQ:NDifferentialInequality} is equivalent to
\begin{equation*}
\TimeDeriv (N_R(t)+N(t)e^{\beta t}) \leq Ke^{\beta t},
\end{equation*}
Integrating this inequality gives
\begin{equation*}
N_R(t)+N(t) \leq \frac{K}{\beta} + ( N_{R,0}+N_0 - \frac{K}{\beta}) e^{-\beta t} \leq \frac{K}{\beta}+ N_{R,0}+N_0.
\end{equation*}

As $N_R,N$ are both positive, they are individually bounded by this constant.

Now,
\begin{equation*}
\TimeDeriv (G_1(t)+G_2(t))  =  G_{prod} - k_{int}G_2(t) - k_{ren} G_1(t)
\end{equation*}
Setting $\alpha = \textrm{min}[k_{int},k_{ren}]$ gives:
\begin{equation*}
\TimeDeriv (G_1(t)+G_2(t)) \leq \Gprod - \alpha(G_1(t)+G_2(t))
\end{equation*}

A similar argument to above gives
 \begin{equation}
G_1(t)+G_2(t) \leq \frac{\Gprod}{\alpha} + ( G_{1,0}+G_{2,0}-\frac{\Gprod}{\alpha}) e^{-\alpha t} \leq  G_{1,0}+G_{2,0}+ \frac{\Gprod}{\alpha}.
\end{equation}
Again, as both $G_1(t)$ and $G_2(t)$ are positive, they are individually bounded by this constant.

Then, solutions of the mathematical model $\mathcal{P}$ with positive initial conditions remain positive and bounded. \qed
\end{proof}


\section{Existence and uniqueness of the Homeostatic Steady State of the QSP granulopoiesis model}
\label{sec:ExistUniqueDDE}

Consider the QSP granulopoiesis model defined by the system of DDEs \eqref{eq:DDEmodel},
and associated initial conditions and histories, with strictly positive parameters satisfying the constraints
\begin{equation}
\ftrans^{max} > \ftranshomeo, \quad V_{max} >1, \quad b_G > \frac{\ftrans^{max}}{\ftranshomeo} G_{BF}^*, \quad  \textrm{and} \quad b_V > \Gonehomeo V_{max}.
\label{Constraints}
\end{equation}
The following results will demonstrate the existence and uniqueness of a homeostasic steady state of this model.

\begin{proposition}\label{SteadyStateProp}
Assume that
\begin{equation*}
I_G(t) = 0,\quad s_G >0,\quad f_Q(A_Q-1) > \kappa^h+\kappa_{\delta}
\end{equation*}
and define
\begin{equation*}
\overline{G}_1 = \frac{\kappa_{int}A_N^h\kappa^h\theta_Q}{\gammaNR+\phi^h_{\NR}}\left[ 1+\frac{\phi^h_{\NR}}{\gamma_N}\right] \left( f_Q\frac{A_Q-1}{\kappa^h+\kappa_{\delta}}-1\right)^{\frac{1}{s_Q}}.
\end{equation*}

If $\overline{G}_1 > \Gprod$, then the system of DDEs given by \eqref{FullModel} has a unique positive homeostatic steady state.
\end{proposition}

\begin{proof}
Consider the differential equation given in \eqref{FullModel}. Define
\begin{equation*}
\mathbf{X} = \left( Q,\NR,N,G_1,G_2\right)\quad \textrm{and} \quad \mathbf{f} = \frac{\textrm{d}\mathbf{X}}{\textrm{d}t}.
\end{equation*}

The steady states are given implicitly by the solutions of
\begin{equation*}
\TimeDeriv\mathbf{X}(t) = 0,
\end{equation*}
and are denoted by $Q^h,\NR^h,N^h,G_1^h,G_2^h$ respectively. Moreover, denote the homeostatic values of each function with the superscript $h$.

A simple calculation to find the non-zero solution of
\begin{equation*}
\TimeDeriv Q(t) = -(\kappa(G_{1}(t))+\kappa_{\delta}+\beta(Q(t)))Q(t)+A_{Q}(t)\beta(Q(t-\tauQ))Q(t-\tauQ) = 0
\end{equation*}
gives the expression for the homeostatic steady state of $Q(t)$
\begin{equation*}
Q^h = \theta_Q\left[ f_Q\frac{A_Q-1}{\kappa^h+\kappa_{\delta}}-1\right]^{1/s_Q}.
\end{equation*}

The conditions in the statement of Proposition \ref{SteadyStateProp} ensure that the steady state is positive.
Using $Q^h$, we calculate the solution $\NR^h$ of
\begin{equation*}
\begin{array}{lll}
\dfrac{d}{\textrm{d}t}\NR(t) &=&  A_{N}(t)\kappa(G_{1}(t-\tauN(t)))Q(t-\tauN(t))\dfrac{\VN(G_{1}(t))}{\VN(G_{1}(t-\tauNM(t)))}\notag \\[2mm]
&&-(\gammaNR+\ftrans(G_{BF}(t)))\NR(t) = 0.\\
\end{array}
\end{equation*}

The homeostatic value is
\begin{equation*}
\NR^h = \frac{A^h_N\kappa^hQ^h}{\gammaNR+ \ftrans^h} > 0.
\end{equation*}

The final homeostasis value that can be easily expressed is the solution of
\begin{equation*}
\TimeDeriv N(t) = \ftrans(G_{BF}(t))\NR(t)-\gamma_{N}N(t) = 0.
\end{equation*}

Using the homeostatic value of $\NR^h$ and solving the steady state equation gives
\begin{equation*}
N^h = \frac{\ftrans^hN^h_R}{\gamma_N}.
\end{equation*}

Now, consider the coupled G-CSF kinetics given by equations~\eqref{eq:FreeGCSF} and \eqref{eq:BoundGCSF} whose homeostatic sum is expressed as
\begin{equation}\label{G1hom}
\TimeDeriv (G_{1}(t)+G_2(t)) = \Gprod -k_{ren}G_1^h- k_{int}G_2^h =0.
\end{equation}

Isolating $G_2^h$ as a function of $G_1^h$ and substituting into \eqref{eq:BoundGCSF}
\begin{equation}\label{G1G2Identity}
0 = -k_{int}G_{2}^h+k_{12}[(\NR^h+N^h)V-G_{2}^h](G_{1}^h)^{s_{G}}-k_{21}G_{2}^h
\end{equation}
gives the following exponential polynomial in $G_1^h$.

\begin{equation}
F(G_1^h) = \alpha [G_1^h]^{s_G+1} + \beta [G_1^h]^{s_G} + \delta G_1^h -\zeta,
\label{PolynominalG1}
\end{equation}
where
\begin{align*}
\alpha = \frac{k_{12} k_{ren}}{k_{int}},\qquad \beta = \frac{k_{21}}{k_{int}}\left[V(N^h+N^h_R)k_{int}-\Gprod\right], \\
 \delta = k_{ren}\frac{k_{int}+k_{21}}{k_{int}},\qquad \zeta = \Gprod\frac{k_{int}+k_{21}}{k_{int}}.
\end{align*}

Then, any root of $F$ and the corresponding $G_2$ value from \eqref{G1hom}  will be a solution of \eqref{G1G2Identity}. The condition
\begin{equation*}
\Gprod < V[\NR^h+N^h]k_{int} = \overline{G}_1
\end{equation*}
guarantees that the coefficients are positive.

A simple calculation shows that
\begin{equation*}
F(0) < 0, \quad F(\zeta/\delta) > 0.
\end{equation*}

By the IVT, there must exist a solution $G_1^h\in (0,\zeta/\delta)$ such that
\begin{equation}
0 = \alpha (G_1^h)^{s_G+1}+ \beta (G_1^h)^{s_G}+ \delta (G_1^h) - \zeta.
\label{G1Homeo}
\end{equation}
Since for all $x \in \mathds{R}^+$
\begin{equation*}
F'(x) = (s_G+1)\alpha x^{s_G} + \beta s_G x^{s_G-1} + \delta > 0,
\end{equation*}
thus $F$ is a strictly increasing function and the zero found in \eqref{G1Homeo} is unique.

Therefore, the corresponding $G_2^h$ given by \eqref{G1hom} is unique. The condition $0<G_{1}^{h}<\zeta/\delta$ together with \eqref{G1hom} and $\zeta/\delta=\Gprod/k_{ren}$ ensures the positivity of $G_2^h$. Thus the homeostatic steady state $(Q^{h},N_{R}^{h},N^{h},G_{1}^{h},G_{2}^{h})$ exists, is positive and unique. \qed
\end{proof}

\section{Stability of the equilibria}
\label{sec:Stability}

\subsection{Characteristic equations of the Quartino model}
\label{sec:CharEqns}

The characteristic equations for the time-rescaled Quartino model \eq{lctsd2}, for the equivalent distributed delay DDE \eq{lctsd3}, and for the discrete delay DDE model \eq{lctsd5} have relatively simple form and are stated explicitly in the main text as equations \eq{char.eq.distr} and \eq{char.eq.disc}. In contrast the Quartino model \eq{eq:QuartinoModelLimited} has a much more complicated form both for general $n$ and when $n=4$ as in Quartino \cite{Quartino2014}, which we derive here.

The Quartino model \eq{eq:QuartinoModelLimited} has characteristic equation given by \eq{char.eq.quart} where $\mathds{J}$ is the Jacobian of $\mathbf{F}$ in \eq{eq:VectorQuartino}.
Differentiating the terms in \eqref{eq:QuartinoModelLimited} for general $\mathbf{X}$ we obtain the Jacobian
\begin{equation}
\label{eq:Quart.Jac}
\mathds{J}(\mathbf{X})=\begin{bmatrix}
    J_{11} & 0 & 0 & \cdots & \cdots & 0 & J_{1,n+3} \\
    J_{21} & -J_{32} & 0 & \cdots & \cdots & 0 & J_{2,n+3} \\
    0  & J_{32} & -J_{32} & 0 & \cdots & 0 & J_{3,n+3} \\
     \vdots  & \ddots & \ddots & \ddots &\ddots & \vdots & \vdots \\
     \vdots & & \ddots & J_{32} & -J_{32} & 0 & J_{n+1,n+3} \\
    0 & 0 & \cdots & 0 & J_{32} & -k_{circ} &  J_{n+2,n+3} \\
    0 & 0 & \cdots & \cdots & 0 &-\kANC G & J_{n+3,n+3}
\end{bmatrix}
\end{equation}
where
\begin{gather*}
J_{11}=k_{tr}\left(\left(\frac{G}{G_0}\right)^\gamma -\left(\frac{G}{G_0}\right)^\beta\right),\qquad  J_{21}=k_{tr}\left(\frac{G}{G_0}\right)^\beta,\qquad J_{32}=a\left(\frac{G}{G_0}\right)^\beta,\\[2mm]
J_{1,n+3}=-\frac{Pk_{tr}}{G_0}\left(\gamma\left(\frac{G}{G_0}\right)^{\!\gamma-1}\!-\beta \left(\frac{G}{G_0}\right)^{\!\beta-1}\right),\quad
J_{2,n+3}=\frac{\beta}{G_0}(k_{tr}P - aT_1) \!\left(\frac{G}{G_0}\right)^{\!\beta - 1}\!\!,
\\[2mm]
J_{n+2,n+3}=\frac{T_n\:\beta\: a}{G_0}\left(\frac{G}{G_0}\right)^{\beta - 1},\qquad
J_{n+3,n+3}=-k_e-N\kANC,\\
J_{j,n+3}=\frac{\beta a}{G_0}(T_{j-2}-T_{j-1})\left(\frac{G}{G_0}\right)^{\beta - 1},\quad j=3,\ldots,n+1.
\end{gather*}
Then the characteristic function for the Quartino model \eq{eq:QuartinoModelLimited} is the $(n+3)$-degree polynomial in $\lambda$:
\begin{align} \notag
\Delta(\lambda)&=\det(\lambda\mathds{I}-\mathds{J}(\mathbf{X}))\\ \notag
&=(\lambda-J_{11})(\lambda+J_{32})^n\Bigl((\lambda-k_{circ})(\lambda+J_{n+3,n+3})+\kANC G J_{n+2,n+3}\Bigr)\\ \label{char.eq.quart.un.X}
& \qquad -J_{1,n+3}J_{21}(J_{32})^{n}\kANC G\\ \notag
& \qquad +\sum_{j=2}^{n+1}(-1)^j(\lambda-J_{11})(\lambda+J_{32})^{j-2}J_{j,n+3}(J_{32})^{n+2-j}\kANC G.
\end{align}
This is considerably more complicated than the characteristic function of the time-rescaled Quartino model which we also stated for general $\mathbf{X}$ in equation \eq{char.eq.quart}. Fortunately \eq{char.eq.quart.un.X} does simplify somewhat at the steady states. At both $\mathbf{X}_1^*$ and $\mathbf{X}_2^*$
we have
$$k_{tr}P=aT_1,\qquad T_{j}=T_{j-1},$$
and hence
$$T_{j,n+3}=0, \quad j=2,\ldots,n+1.$$
Additionally at $\mathbf{X}_2^*$ only, $G=G_0$ implies $J_{11}=1$, while at
$\mathbf{X}_1^*$ only, $P=T_n=0$ implies that $T_{1,n+3}=T_{n+2,n+3}=0$. So
%
\begin{equation}
\label{eq:QuartinoJacobian1}
\mathds{J}(\mathbf{X}^*_1)=\begin{bmatrix}
    J_{11} & 0 & 0 & \cdots & \cdots & 0 & 0\\
    J_{21} & -J_{32} & 0 & \cdots & \cdots & 0 & 0 \\
    0  & J_{32} & -J_{32} & 0 & \cdots & 0 & 0 \\
     \vdots  & \ddots & \ddots & \ddots &\ddots & \vdots & \vdots \\
     \vdots & & \ddots & J_{32} & -J_{32} & 0 & 0 \\
    0 & 0 & \cdots & 0 & J_{32} & -k_{circ} &  0 \\
    0 & 0 & \cdots & \cdots & 0 &-\frac{\kANC k_{in}}{k_e} & -k_e
\end{bmatrix},
\end{equation}
with
\be \label{JvalsX1}
J_{11}=k_{tr}\!\left(\!\left(\frac{k_{in}}{G_0k_e}\right)^{\!\gamma}\! - \left(\frac{k_{in}}{G_0k_e}\right)^{\!\beta}\right)\!, \quad
J_{21}=k_{tr}\left(\frac{k_{in}}{G_0k_e}\right)^{\!\beta}\!\!, \quad
J_{32}=a\left(\frac{k_{in}}{G_0k_e}\right)^{\!\beta}\!,
\ee
and
\begin{align} \notag
&\Delta_1(\lambda)=\det(\lambda\mathds{I}-\mathds{J}(\mathbf{X}^*_1))\\
& = \left(\lambda-k_{tr}\!\left(\!\left(\frac{k_{in}}{G_0k_e}\right)^{\!\gamma}\! - \left(\frac{k_{in}}{G_0k_e}\right)^{\!\beta}\right)\!\right)\!\left(\lambda+a\left(\frac{k_{in}}{G_0k_e}\right)^{\!\beta}\!\right)^{\!n}
(\lambda+k_{circ})(\lambda+k_e) \label{char.eq.quart.un.X1}
\end{align}
At $\mathbf{X}_2^*$ from \eq{char.eq.quart.un.X} we obtain
\begin{align} \notag
\Delta_2(\lambda)&=\det(\lambda\mathds{I}-\mathds{J}(\mathbf{X}_2^*))\\ \notag
&=\lambda(\lambda+a)^n\Bigl((\lambda-k_{circ})(\lambda+J_{n+3,n+3})+\kANC G J_{n+2,n+3}\Bigr)\\ \label{char.eq.quart.un.X2}
& \qquad -J_{1,n+3}a^{n+1}\kANC G.
\end{align}
Equations \eq{char.eq.quart.un.X1} and \eq{char.eq.quart.un.X2} give the characteristic function as a polynomial of degree $n+3$ at the two steady states $\mathbf{X}_1^*$ and $\mathbf{X}_2^*$ of the generalised Quartino model
\eq{eq:QuartinoModelLimited} for general values of the parameters, including for the parameters used in Quartino \cite{Quartino2014} when $n=4$ and the polynomials have degree $7$.

\subsection{Stability of the equilibria of the Quartino ODE model}
\label{sec:QuartinoStability}

The stability of the steady states of the Quartino model \eq{eq:QuartinoModelLimited} was already considered in Theorem~\ref{stab.prop.disc}, which was proved using the alternate time-rescaled forms of the model. Here we present
proofs of stability results directly in the original formulation, mainly to demonstrate how much more difficult they are.

\begin{proposition}\label{thm:StabilityEq1}
The equilibrium point $\mathbf{X}^*_1$ of the Quartino model \eq{eq:QuartinoModelLimited}
is locally asymptotically stable when $\gamma<\beta$ and unstable when $\gamma>\beta$.
\end{proposition}
\begin{proof}
For the case $\beta>\gamma$, we show asymptotic stability of $\mathbf{X}^*_1$ using a Lyapunov function
\cite{Meiss2007} applied to the first $(n+2)$ coordinates of the system.

First consider the dynamics of $G(t)$. Let $\lambda\in(0,1)$ and suppose $N(t)\in[0,\lambda N_0]$ for $t\in[0,\TN]$, and let
$$G_\textit{ratio}=\frac{k_e+\kANC N_0}{k_e+\lambda\kANC N_0}>1.$$
Then from \eq{eq:QuartinoModelLimitede} we have $\TimeDeriv G\geq0$ when $G/G_0<G_\textit{ratio}$, and the positivity of $N(t)$ implies that $\TimeDeriv G\leq0$ when $G>k_{in}/k_e$.
Hence choosing $G(0)\in[G_0 G_\textit{ratio},(1+\lambda)k_{in}/k_e]$, ensures that
$G(t)\in[G_0 G_\textit{ratio},(1+\lambda)k_{in}/k_e]$ for all $t\in[0,\TN]$, under the assumption that
$N(t)\in[0,\lambda N_0]$ for $t\in[0,\TN]$. Now let
\be \label{eq:Lyap1}
V(\mathbf{X}(t))=\frac12\mu P(t)^2+\frac12\sum_{j=1}^nT_j(t)^2+\frac12\alpha N(t)^2,
\ee
where the parameters satisfy $\mu>0$ and $\alpha>0$, 
with values to be specified below.
Note that
$V(\mathbf{X}^*_1)=0$, while $V(\mathbf{X})>0$ if any of $P$, $T_j$, $N$ is non-zero.
Then differentiating and using
\eq{eq:QuartinoModelLimited} we obtain
\begin{align*}
\TimeDeriv & V(\mathbf{X}(t))
=\mu P(t)\frac{\textrm{d}P}{\textrm{d}t}+\sum_{j=1}^nT_j(t)\frac{\textrm{d}T_j}{\textrm{d}t}
+\alpha N(t)\frac{\textrm{d}N}{\textrm{d}t}
\\
&=-\left(\frac{G(t)}{G_{0}}\right)^\beta\left(\mu\left(k_{tr}- \kP\left(\frac{G(t)}{G_{0}}\right)^{\gamma-\beta}\right)
-\frac{k_{tr}^2}{2a}\right)P(t)^2\\
&\quad-\frac{a}{2}\left(\frac{G(t)}{G_{0}}\right)^\beta \!\left(\frac{k_{tr}}{a}P(t)-T_1(t)\right)^2
-\sum_{j=2}^n \frac{a}{2}\left(\frac{G(t)}{G_{0}}\right)^\beta \!(T_{j-1}(t)-T_j(t))^2\\
&\quad-\frac{a}{2}\left(\frac{G(t)}{G_{0}}\right)^\beta (T_n(t)-\alpha N(t))^2
-\alpha\biggl(\kcirc-\frac{a}{2}\biggl(\frac{G(t)}{G_{0}}\biggr)^{\!\beta}\alpha\biggr)N(t)^2.
\end{align*}
Notice that all the terms on the right-hand side are non-positive, except possibly for the first and last term.
But with $k_{tr}=\kP$, $\beta>\gamma$ and $G(t)/G_0\geq G_\textit{ratio}>1$ it follows that
$$k_{tr}- \kP\left(\frac{G(t)}{G_{0}}\right)^{\gamma-\beta}>0,$$
and hence for $\mu>0$ sufficiently large the coefficient of $P(t)^2$ is strictly negative. Similarly, since $G(t)$ is bounded above for $t\in[0,\TN]$ for $\alpha>0$ sufficiently small the coefficient of $N(t)^2$ is also strictly negative.
It follows that $\TimeDeriv V(\mathbf{X}(t))<0$ unless $P=T_j=N=0$.

If $V(\mathbf{X}(t))<\frac{\alpha\lambda^2N_0^2}{2}$, then \eq{eq:Lyap1} implies that $N(t)<\lambda N_0$. Therefore, we choose initial conditions such that $G(0)\in[G_0 G_\textit{ratio},(1+\lambda)k_{in}/k_e]$ and
$$V(\mathbf{X}(0))<\frac{\alpha\lambda^2N_0^2}{2}.$$
Then since $V(\mathbf{X}(t))$ is nonincreasing, we have $N(t)<\lambda N_0$ for all $t>0$, then
$G(t)\in[G_0 G_\textit{ratio},(1+\lambda)k_{in}/k_e]$ for all $t>0$ and $V(\mathbf{X}(t))\to0$ as $t\to\infty$ which implies that $[P(t),T_1(t),\ldots,T_n(t),N(t)]\to[0,0,\ldots,0,0]$ as $t\to\infty$.

Finally considering \eq{eq:QuartinoModelLimitede}, we have $\TimeDeriv G<0$ if $G>k_{in}/k_e$ and
$\TimeDeriv G>0$ if $G<k_{in}/(k_e+\kANC N(t))$ with $N(t)\to0$ as $t\to\infty$, hence $G(t)\to k_{in}/k_e$ as $t\to\infty$. This completes the proof of local asymptotic stability of $\mathbf{X}^*_1$ when $\gamma<\beta$.

For the case $\gamma>\beta$, we show that $\mathbf{X}^*_1$ is unstable, using linearization theory, by showing
there is a real positive characteristic root in this case.
The characteristic function evaluated at $\mathbf{X}^*_1$ is given by \eq{char.eq.quart.un.X1} and statisfies
$\Delta_1(\lambda)\to+\infty$ as $\lambda\to+\infty$.
Hence, by the IVT, to show that there exists
a positive eigenvalue, $\lambda>0$ such that $\Delta_1(\lambda)=0$, it is sufficient to show that
$\Delta_1(0)<0$. But from \eq{char.eq.quart.un.X1}
\begin{align*}
\Delta_1(0)&=\det(-\mathds{J}(\mathbf{X}^*_1))=(-1)^{n+3}\det(\mathds{J}(\mathbf{X}^*_1))
=(-1)^{2n+5}k_{circ}k_eJ_{11}J_{32}^n\\
&=k_{circ}k_ea^nk_{tr}
\left(\left(\frac{k_{in}}{G_0k_e}\right)^{\!\beta} - \left(\frac{k_{in}}{G_0k_e}\right)^{\!\gamma}\right)
\left(\frac{k_{in}}{G_0k_e}\right)^{\!n\beta}.
\end{align*}
From the constraint in equation~\eqref{eq:kinb}, we have that $k_{in}=G_0(k_e+\kANC N_0)>G_0k_e$, then
$\gamma>\beta$ and the positivity of all the parameters imply that $\Delta_1(0)<0$ as required. Thus
$\Delta_1(\mathbf{X}^*_1)$ has a characteristic value $\lambda>0$ and $\mathbf{X}^*_1$ is unstable.
\qed
\end{proof}

\begin{proposition}\label{thm:StabilityEq2}
The equilibrium point $\mathbf{X}^*_2$ of the Quartino model \eq{eq:QuartinoModelLimited}
is unstable when $\gamma<\beta$.
\end{proposition}

\begin{proof}
Consider the characteristic function $\Delta_2(\lambda)$ evaluated at $\mathbf{X}_2^*$. From
\eq{char.eq.quart.un.X2} we have $\Delta_2(\lambda)\to+\infty$ as $\lambda\to\infty$, while recalling that $J_{11}=0$
and using \eq{TjPss} and \eq{eq:kinb}
$$\Delta_2(0)
=-J_{1,n+3}k_{tr}a^{n}\kANC G_0=P_0k_{tr}^2(\gamma-\beta)a^{n}\kANC
=a^{n}k_{tr}k_{circ}(\gamma-\beta)\left(\frac{k_{in}}{G_0}-k_e\right).$$
From \eq{eq:kinb} we have $k_e<(k_e+\kANC N_0)=k_{in}/G_0$, thus $\Delta_2(0)<0$.
Thus the IVT guarantees that \eq{char.eq.quart.un.X2} has at least one real positive root and therefore $\mathbf{X}^*_2$ is unstable. 
\qed
\end{proof}

\subsection{Characteristic Equation of the discrete delay DDE Quartino model}
\label{app.jac.disc}

The discrete delay DDE model \eq{lctsd5}, obtained as the $n\to\infty$ limit of the time-rescaled Quartino model \eq{lctsd2} has characteristic equation \eq{char.eq.det}. The $3\times3$ linearisation matrices $\mathds{A}$ and $\mathds{B}$ from \eq{disc.jac}, are derived by differentiating the terms on the right-hand side of \eq{lctsd5}.
We will denote their entries by $A_{ij}$ and $B_{ij}$ for $i,j=\{1,2,3\}$.
Let us define the ratio between G-CSF concentrations $G_{r}\coloneqq G^{*}/G_{0}$  and let $\mathbf{Y}^*=(P^*,N^*,G^*)$ be a generic steady state then
\begin{equation}\label{AB}
\mathds{A}=\begin{bmatrix}
A_{11} & 0 & A_{13}\\
0 & A_{22} & A_{23} \\
0 & A_{32} & A_{33}\\
\end{bmatrix},
\qquad
\mathds{B}=\begin{bmatrix}
0 & 0 & 0\\
k_{tr} & 0 & 0\\
0 & 0 & 0\\
\end{bmatrix},
\end{equation}
where
\begin{gather*}
 A_{11}=\kP G_{r}^{\gamma-\beta}-k_{tr},\quad A_{13}=\frac{P^{*}}{G^{*}}(\gamma-\beta)\kP G_{r}^{\gamma-\beta},\quad
 A_{22}=-\kcirc G_{r}^{-\beta},\\
 A_{23}=\frac{N^{*}}{G^{*}}\beta\kcirc G_{r}^{-\beta},\quad
 A_{32}=-\kANC G^{*}G_{r}^{-\beta},\quad A_{33}=-(k_{e}+\kANC N^{*})G_{r}^{-\beta}.
\end{gather*}
Using the matrices \eqref{AB} we rearrange the determinant from equation \eqref{char.eq.det} and obtain
\begin{equation}\nonumber\label{A.disc}
\begin{vmatrix}
\lambda-A_{11} & 0 & -A_{13}\\
-k_{tr}e^{-\lambda\tau} & \lambda-A_{22} & -A_{23}\\
0 & -A_{32} & \lambda-A_{33}\\
\end{vmatrix}=0,
\end{equation}
from where we get the characteristic equation \eqref{char.eq.disc} with the coefficients $a_2$, $a_1$, $a_0$ and $b$ given by
\begin{equation}\label{coef.disc}
\begin{aligned}
& a_2=-(A_{11}+A_{22}+A_{33}),\\
& a_1=A_{11}(A_{22}+A_{33})+A_{22}A_{33}-A_{23}A_{32},\\
& a_0=A_{11}(A_{23}A_{32}-A_{22}A_{33}),\\
& b=A_{13}A_{32}k_{tr}.\\
\end{aligned}
\end{equation}
Recalling the steady states $\mathbf{Y}^*_1$ and  $\mathbf{Y}^*_2$ given, respectively, by \eqref{lctdsdeq1} and \eqref{lctsdeq2}, calculated with the parameter constraints \eqref{eq:kinc} and \eqref{eq:kinb}, we can evaluate the coefficients in \eq{coef.disc} at each steady state.
For $\mathbf{Y}^*_1$ we have $P^*=N^*=0$ which implies that $A_{13}=A_{23}=0$ and we have
\begin{equation}\label{coef.disc1}
\begin{aligned}
& a_2=\kP(1-G_{r}^{\gamma-\beta})+(\kcirc+k_{e})G_{r}^{-\beta},\\
& a_1=\kP(1-G_{r}^{\gamma-\beta})(\kcirc+k_{e})G_{r}^{-\beta}+\kcirc k_{e}G_{r}^{-\beta},\\
& a_0=\kP\kcirc k_{e}(1-G_{r}^{\gamma-\beta})G_{r}^{-2\beta},\\
& b=0,\\
\end{aligned}
\end{equation}
For $\gamma>\beta$ we have $a_{0}<0$ in \eqref{coef.disc1}, since $k_{in}>k_{e}G_{0}$ from the constraint \eqref{eq:kinb} and consequently $G_{r}>1$. Similarly $a_{0}>0$ when $\gamma<\beta$.

For $\mathbf{Y}_2^*$ we have $G^*=G_0$, hence  $G_{r}=1$ which implies $A_{11}=0$.
\begin{equation}\label{coef.disc2}
\begin{aligned}
& a_2=\kcirc+k_{e}+\kANC N_{0},\\
& a_1=\kcirc(k_{e}+\kANC N_{0}(1+\beta)),\\
& a_0=0,\\
& b= (\beta-\gamma)k_{p}\kcirc\kANC N_{0}.\\
\end{aligned}
\end{equation}

\subsection{Characteristic Equation of the QSP granulopoiesis model}
\label{app.jacDDE}

The QSP granulopoiesis model \eq{eq:DDEmodel} has characteristic equation \eq{char.eq.det.qsp}.
From the DDE of the QSP granulopoiesis model \eq{eq:DDEmodel} we compute the linearization matrices
and indicate which terms are symmetric and asymmetric. We obtain
\begin{equation}\label{MatrixAB}
\mathds{A}=\begin{bmatrix}
A_{11} & A_{12} & 0 & 0 & 0\\
0 & A_{22} & A_{23} & A_{24} & A_{25}\\
0 & A_{32} & A_{33} & 0 & A_{35}\\
0 & A_{42} & A_{43} & A_{44} & A_{45}\\
0 & A_{52} & A_{53} & A_{54} & A_{55}\\
\end{bmatrix},\quad
\mathds{B}=\begin{bmatrix}
B_{11} & 0 & 0 & 0 & 0\\
0 & 0 & 0 & 0 & 0\\
0 & 0 & 0 & 0 & 0\\
0 & 0 & 0 & 0 & 0\\
0 & 0 & 0 & 0 & 0\\
\end{bmatrix},
\end{equation}
\begin{equation}\label{MatrixCDE}
\mathds{C}=\begin{bmatrix}
0 & 0 & 0 & 0 & 0\\
C_{21} & 0 & 0 & C_{24} & 0\\
0 & 0 & 0 & 0 & 0\\
0 & 0 & 0 & 0 & 0\\
0 & 0 & 0 & 0 & 0\\
\end{bmatrix},\quad
\mathds{D}=\begin{bmatrix}
0 & 0 & 0 & 0 & 0\\
0 & 0 & 0 & D_{24} & 0\\
0 & 0 & 0 & 0 & 0\\
0 & 0 & 0 & 0 & 0\\
0 & 0 & 0 & 0 & 0\\
\end{bmatrix},\quad
\mathds{E}=\begin{bmatrix}
0 & 0 & 0 & 0 & 0\\
0 & 0 & 0 & E_{24} & 0\\
0 & 0 & 0 & 0 & 0\\
0 & 0 & 0 & 0 & 0\\
0 & 0 & 0 & 0 & 0\\
\end{bmatrix},
\end{equation}
with
\begin{gather*}
A_{11} = -\kappa^{*}-\kappa_{\delta}-\beta(Q^{*})-Q^{*}\frac{\textrm{d}\beta}{\textrm{d}Q}(Q^{*}), \quad
A_{12} =  -Q^{*}\frac{\textrm{d}\kappa}{\textrm{d}G_{1}}(G_{1}^{*}),\\
A_{22} = -\gammaNR-\ftrans^{*}-\NR^{*}\frac{\textrm{d}\ftrans}{\textrm{d} G_{BF}}(G_{BF}^{*})\frac{\partial G_{BF}}{\partial \NR}(\NR^{*},N^{*},G_{2}^{*}),\\
A_{23} = -\NR^{*}\frac{\textrm{d}\ftrans}{\textrm{d}G_{BF}}(G_{BF}^{*})\frac{\partial G_{BF}}{\partial N}(\NR^{*},N^{*},G_{2}^{*}), \quad
A_{24} = \frac{A_{N}^{*}\kappa^{*}Q^{*}}{\VN(G_{1}^{*})}\dfrac{d\VN}{\textrm{d}G_{1}}(G_{1}^{*}),\\
A_{25} = -A_{35} = -\NR^{*}\frac{\textrm{d}\ftrans}{\textrm{d}G_{BF}}(G_{BF}^{*})\frac{\partial G_{BF}}{\partial G_{2}}(\NR^{*},N^{*},G_{2}^{*}), \\
A_{32} = \NR^{*}\frac{\textrm{d}\ftrans}{\textrm{d}G_{BF}}(G_{BF}^{*})\frac{\partial G_{BF}}{\partial \NR}(\NR^{*},N^{*},G_{2}^{*}) + \ftrans^{*},\\
A_{33} = \NR^{*}\frac{\textrm{d}\ftrans}{\textrm{d}G_{BF}}(G_{BF}^{*})\frac{\partial G_{BF}}{\partial N}(\NR^{*},N^{*},G_{2}^{*}) - \gammaNR,\\
A_{42} = A_{43} = -A_{52} = -A_{53} =-k_{12}V(G_{1}^{*})^{s_{G}},\\
A_{44} = -k_{ren}-k_{12}[(\NR^{*}+N^{*})V-G_{2}^{*}]s_{G}(G_{1}^{*})^{s_{G}-1}, \quad
A_{45} = k_{12}(G_{1}^{*})^{s_{G}}+k_{21},\\
A_{54} = k_{12}[(\NR^{*}+N^{*})V-G_{2}^{*}]s_{G}(G_{1}^{*})^{s_{G}-1},\quad
A_{55} = - k_{int}-k_{12}(G_{1}^{*})^{s_{G}}-k_{21},\\
B_{11} =  A_{Q}\left(\beta(Q^{*})+Q^{*}\frac{\textrm{d}\beta}{\textrm{d}Q}(Q^{*})\right), \quad
C_{21} = \kappa^{*}A_{N}^{*},\quad
C_{24} = A_{N}^{*}Q^{*}\dfrac{d\kappa}{\textrm{d}G_{1}}(G_{1}^{*}),\\
D_{24} = -\frac{A_{N}^{*}\kappa^{*}Q^{*}}{\VN(G_{1}^{*})}\dfrac{d\VN}{\textrm{d}G_{1}}(G_{1}^{*}) = -A_{24},\quad
E_{24} = \kappa^{*}Q^{*}\frac{\textrm{d}\tilde{A}_{N}}{\textrm{d}\tilde{G}_{1}} (G_{1}^{*}) =\kappa^{*}Q^{*}\tauNP A_{N}^{*}\frac{\textrm{d}\etaNP}{\textrm{d}G_{1}}(G_{1}^{*}).
\end{gather*}


\end{document}